\documentclass[letter,11pt]{article}

\usepackage{setspace}
\usepackage{subcaption}
\usepackage{enumerate}
\usepackage{amsmath}
\usepackage{amsfonts}
\usepackage{graphicx}
\usepackage{tcolorbox}
\usepackage{amssymb}
\usepackage{multirow}
\usepackage{geometry}
\usepackage{soul}
\usepackage{colortbl}
\usepackage{wrapfig}
\usepackage{algorithm}
\usepackage{algorithmicx}
\usepackage{algpseudocode}
\usepackage{amsthm}
\usepackage{todonotes}
\usepackage{mathtools}
\usepackage{mathrsfs}

\usepackage[pdftex, plainpages = false, pdfpagelabels, 
                 bookmarks=false,
                 bookmarksopen = true,
                 bookmarksnumbered = true,
                 breaklinks = true,
                 linktocpage,
                 pagebackref,
                 colorlinks = true,  
                 linkcolor = blue,
                 urlcolor  = blue,
                 citecolor = red,
                 anchorcolor = green,
                 hyperindex = true,
                 hyperfigures
                 ]{hyperref}

\usepackage{thmtools} 
\usepackage{thm-restate}

\algtext*{EndWhile}
\algtext*{EndIf}
\algtext*{EndFor}
\algtext*{EndProcedure}
\algtext*{EndFunction}

\newtheorem{lemma}{Lemma}

\numberwithin{lemma}{section}
\newtheorem{theorem}[lemma]{Theorem}
\newtheorem{theorem*}[lemma]{Theorem*}
\newtheorem{corollary}[lemma]{Corollary}
\newtheorem{definition}[lemma]{Definition}
\newtheorem{observation}[lemma]{Observation}

\newtheorem{fact}[lemma]{Fact}

\usepackage{wrapfig}
\usepackage{bm}

\title{Subexponential Parameterized Algorithms for Hitting Subgraphs}
\author{Daniel Lokshtanov\footnote{University of California, Santa Barbara, USA. Email: \texttt{daniello@ucsb.edu}} \and Fahad Panolan\footnote{University of Leeds, UK. Email: \texttt{F.panolan@leeds.ac.uk}} \and Saket Saurabh\footnote{Institute of Mathematical Sciences, Chennai, India. Email: \texttt{saket@imsc.res.in}} \and Jie Xue\footnote{New York University Shanghai, China. Email: \texttt{jiexue@nyu.edu}} \and Meirav Zehavi\footnote{Ben-Gurion University, Israel. Email: \texttt{meiravze@bgu.ac.il}}}
\date{}

\begin{document}


\maketitle

\begin{abstract}
For a finite set $\mathcal{F}$ of graphs, the \textsc{$\mathcal{F}$-Hitting} problem aims to compute, for a given graph $G$ (taken from some graph class $\mathcal{G}$) of $n$ vertices (and $m$ edges) and a parameter $k \in \mathbb{N}$, a set $S$ of vertices in $G$ such that $|S| \leq k$ and $G-S$ does not contain any subgraph isomorphic to a graph in $\mathcal{F}$.
As a generic problem, \textsc{$\mathcal{F}$-Hitting} subsumes many fundamental vertex-deletion problems that are well-studied in the literature.
The \textsc{$\mathcal{F}$-Hitting} problem admits a simple branching algorithm with running time $2^{O(k)} \cdot n^{O(1)}$, while it cannot be solved in $2^{o(k)} \cdot n^{O(1)}$ time on general graphs assuming the ETH, follows from the seminal work of Lewis and Yannakakis.

In this paper, we establish a general framework to design subexponential parameterized algorithms for the \textsc{$\mathcal{F}$-Hitting} problem on a broad family of graph classes.
Specifically, our framework yields algorithms that solve \textsc{$\mathcal{F}$-Hitting} with running time $2^{O(k^c)} \cdot n + O(m)$ for a constant $c < 1$ on any graph class $\mathcal{G}$ that admits balanced separators whose size is (strongly) sublinear in the number of vertices and polynomial in the size of a maximum clique.
Examples include all graph classes of polynomial expansion (e.g., planar graphs, bounded-genus graphs, minor-free graphs, etc.) and many important classes of geometric intersection graphs (e.g., map graphs, intersection graphs of any fat geometric objects, pseudo-disks, etc.).
Our algorithms also apply to the \textit{weighted} version of \textsc{$\mathcal{F}$-Hitting}, where each vertex of $G$ has a weight and the goal is to compute the set $S$ with a minimum weight that satisfies the desired conditions.

The core of our framework, which is our main technical contribution, is an intricate subexponential branching algorithm that reduces an instance of \textsc{$\mathcal{F}$-Hitting} (on the aforementioned graph classes) to $2^{O(k^c)}$ general hitting-set instances, where the Gaifman graph of each instance has treewidth $O(k^c)$, for some constant $c < 1$ depending on $\mathcal{F}$ and the graph class.
\end{abstract}


\section{Introduction}

A vertex-deletion problem takes as input a graph $G$, and aims to delete from $G$ a minimum number of vertices such that the resulting graph satisfies some property.
In many vertex-deletion problems, the desired properties can be expressed as excluding a finite set of ``forbidden'' structures.
This motivates the so-called \textsc{$\mathcal{F}$-Hitting} problem.
Formally, for a finite set $\mathcal{F}$ of graphs (which represent the forbidden structures), the \textsc{$\mathcal{F}$-Hitting} problem is defined as follows.

\begin{tcolorbox}[colback=gray!5!white,colframe=gray!75!black]
        \textsc{$\mathcal{F}$-Hitting} \hfill \textbf{Parameter:} $k$
        \vspace{0.1cm} \\
        \textbf{Input:} A graph $G$ of $n$ vertices and $m$ edges (taken from some class $\mathcal{G}$) and a number $k \in \mathbb{N}$.
        \vspace{-0.4cm} \\
        \textbf{Goal:} Compute a set $S \subseteq V(G)$ of vertices such that $|S| \leq k$ and $G - S$ does not contain any graph in $\mathcal{F}$ as a subgraph.
\end{tcolorbox}

As a generic problem, \textsc{$\mathcal{F}$-Hitting} subsumes many fundamental vertex-deletion problems that have been well-studied in the literature, such as \textsc{Vertex Cover}, \textsc{Path Transversal}~\cite{brevsar2013vertex,gupta2019losing,lee2017partitioning}, \textsc{Short Cycle Hitting}~\cite{groshaus2009cycle,lokshtanov2022subexponential,lokshtanov2023framework,pilipczuk2011problems}, \textsc{Component Order Connectivity}~\cite{drange2016computational,gross2013survey,kumar20172lk}, \textsc{Degree Modulator}~\cite{betzler2012bounded,ganian2021structural,gupta2019losing}, \textsc{Treedepth Modulator}~\cite{bougeret2019much,eiben2022lossy,gajarsky2017kernelization}, \textsc{Clique Hitting}~\cite{berthe2024kick,fiorini2018approximability}, \textsc{Biclique Hitting}~\cite{goldmann2021parameterized}, etc.
We present in the appendix detailed descriptions of these problems as well as their background (see also the work~\cite{lokshtanov2023hitting}).

On general graphs (i.e., when the graph class $\mathcal{G}$ contains all graphs), the complexity of \textsc{$\mathcal{F}$-Hitting} is well-understood.
For any (finite) $\mathcal{F}$, the problem admits a simple branching algorithm with running time $2^{O(k)} \cdot n^{O(1)}$~\cite{Cai96}.
On the other hand, it follows from the reductions given in the work of Lewis and Yannakakis~\cite{lewis1980node}, \textsc{$\mathcal{F}$-Hitting} cannot be solved in $2^{o(k)} \cdot n^{O(1)}$ time or even $2^{o(n)}$ time 
for any $\mathcal{F}$ such that there are infinitely many graphs that do not contain any graphs in $\mathcal{F}$ as a subgraph\footnote{Note that if there are only a finite number of graphs that do not contain any graphs in $\mathcal{F}$ as a subgraph, then \textsc{$\mathcal{F}$-Hitting} can be trivially solved in polynomial time.}, assuming the Exponential-Time Hypothesis.
However, this does not rule out the existence of subexponential (parameterized) algorithms for \textsc{$\mathcal{F}$-Hitting} on restrictive graph classes.
In fact, many special cases of \textsc{$\mathcal{F}$-Hitting} have been solved in subexponential time on various graph classes.
For example, \textsc{Vertex Cover}, the simplest (nontrivial) case of \textsc{$\mathcal{F}$-Hitting}, was known to have $2^{o(k)} \cdot n^{O(1)}$-time algorithms on many graph classes, such as 
planar graphs~\cite{DemaineFHT05}, 
minor-free graphs~\cite{demaine2005subexponential}, 
unit-disk graphs~\cite{FominLPSZ19,FominLS18}, 
map graphs~\cite{DemaineFHT05,fomin2019decomposition,FominLS18}, and 
disk graphs~\cite{lokshtanov2022subexponential}.
Some other cases, such as \textsc{Path Transversal}, \textsc{Clique Hitting}, \textsc{Component Order Connectivity}, etc., were also known to admit subexponential parameterized algorithms on specific graph classes, mostly planar graphs, minor-free graphs, and (unit-)disk graphs~\cite{berthe2024kick,DemaineFHT05,FominLPSZ19,FominLS18,lokshtanov2022subexponential}.
Nevertheless, there has not been a comprehensive study on subexponential algorithms for the general \textsc{$\mathcal{F}$-Hitting} problem.

In this paper, we systematically study the \textsc{$\mathcal{F}$-Hitting} problem on a broad family of graph classes, in the context of subexponential parameterized algorithms.
We show that on many graph classes where subexponential algorithms have been well-studied, \textsc{$\mathcal{F}$-Hitting} admits a subexponential parameterized algorithm for any $\mathcal{F}$.
Before presenting our result in detail, we first discuss a difficulty in designing subexponential algorithms for \textsc{$\mathcal{F}$-Hitting}.

\paragraph{Difficulty of solving $\mathcal{F}$-Hitting in subexponential time.}
Compared to most vertex-deletion problems, designing subexponential algorithms for \textsc{$\mathcal{F}$-Hitting} seems especially challenging, due to lack of efficient algorithms for the problem on small-treewidth graphs.
In fact, many fundamental vertex-deletion problems can be solved in $2^{O(t)} n^{O(1)}$ time or $t^{O(t)} n^{O(1)}$ time on graphs of treewidth $t$, by standard dynamic programming on tree decompositions.
When the treewidth $t$ is (strongly) sublinear, the running time is subexponential.
Using this result as a cornerstone, a common approach for designing subexponential algorithms for such problems is to first reduce the treewidth $t$ of the graph to sublinear, and then solve the problem in $2^{O(t)} n^{O(1)}$ time or $t^{O(t)} n^{O(1)}$ time.
Unfortunately, this kind of approach does not work for \textsc{$\mathcal{F}$-Hitting}.
Indeed, the best known algorithm for \textsc{$\mathcal{F}$-Hitting} on graphs of treewidth $t$, given by Cygan et al.~\cite{cygan2017hitting}, requires $2^{t^{f(\mathcal{F})}} n$ time for some function $f$.
Even worse, Cygan et al.~\cite{cygan2017hitting} also showed that such running time is necessary assuming the ETH.
Therefore, even if the treewidth $t$ is already sublinear, the algorithm for \textsc{$\mathcal{F}$-Hitting} can still run in superexponential time.
This fact makes the task of solving \textsc{$\mathcal{F}$-Hitting} in subexponential time rather difficult and require new insights to the problem.

\paragraph{Other related work on \textsc{$\mathcal{F}$-Hitting}.}
The general \textsc{$\mathcal{F}$-Hitting} problem has been studied on graphs of bounded treewidth by Cygan et al. \cite{cygan2017hitting}, who gave FPT algorithms (parameterized by treewidth) and lower bounds for the problem and its variants.
Recently, Dvo{\v{r}}{\'a}k et al.~\cite{lokshtanov2023hitting} studied \textsc{$\mathcal{F}$-Hitting} in the perspective of approximation, showing that the local-search approach of Har-Peled and Quanrud~\cite{har2017approximation} for \textsc{Vertex Cover} and \textsc{Dominating Set} on graph classes of polynomial expansion can be extended to obtain PTASes for \textsc{$\mathcal{F}$-Hitting} as well.
They also gave a $(1+\varepsilon)$-approximation reduction from the problem on any bounded-expansion graph class to the same problem on bounded degree graphs within the class, resulting in lossy kernels for the problem.
Bougeret et al.~\cite{bougeret2024kernelization} considered the kernelization of \textsc{$\mathcal{F}$-Hitting} with various structural parameters.
Finally, \textsc{$\mathcal{F}$-Hitting} is a special case of \textsc{$d$-Hitting Set}, and hence the algorithms for \textsc{$d$-Hitting Set}~\cite{CyganFKLMPPS15,williamson2011design} also apply to \textsc{$\mathcal{F}$-Hitting}.
However, as the set system considered in \textsc{$\mathcal{F}$-Hitting} is implicitly defined and can have size $n^{O(\gamma)}$ where $\gamma = \max_{F \in \mathcal{F}} |V(F)|$, transferring the algorithms for \textsc{$d$-Hitting Set} to \textsc{$\mathcal{F}$-Hitting} usually results in an overhead of $n^{O(\gamma)}$.

\subsection{Our result}

Our main result is a general framework to design subexponential parameterized algorithms for \textsc{$\mathcal{F}$-Hitting}.
In fact, our framework applies to the weighted version of \textsc{$\mathcal{F}$-Hitting}, where the vertices are associated with weights and we want to find the solution with the minimum total weight.

\begin{tcolorbox}[colback=gray!5!white,colframe=gray!75!black]
        \textsc{Weighted $\mathcal{F}$-Hitting} \hfill \textbf{Parameter:} $k$
        \vspace{0.1cm} \\
        \textbf{Input:} A graph $G$ of $n$ vertices and $m$ edges (taken from some class $\mathcal{G}$) together with a weight function $w:V(G) \rightarrow \mathbb{R}_{\geq 0}$, and a number $k \in \mathbb{N}$.
        \vspace{0.1cm} \\
        \textbf{Goal:} Compute a minimum-weight (under the function $w$) set $S \subseteq V(G)$ of vertices such that $|S| \leq k$ and $G - S$ does not contain any graph in $\mathcal{F}$ as a subgraph.
\end{tcolorbox}

Another nice feature of our framework is that the algorithms it yields have running time not only subexponential in $k$ but also linear in the size of the input graph $G$.
More precisely, the running time of our algorithms is of the form $2^{O(k^c)} \cdot n + O(m)$ for some constant $c < 1$.

To present our result, we need to define the graph classes to which our framework applies.
We define these graph classes in the language of \textit{balanced separators}, which serve as a key ingredient in many efficient graph algorithms.
For $\eta,\mu,\rho \geq 0$, we say a graph $G$ admits \textit{balanced $(\eta,\mu,\rho)$-separators} if for every induced subgraph $H$ of $G$, there exists $S \subseteq V(H)$ of size at most $\eta \cdot \omega^\mu(H) \cdot |V(H)|^\rho$ such that every connected component of $H - S$ contains at most $\frac{1}{2} |V(H)|$ vertices, where $\omega(H)$ denotes the size of a maximum clique in $H$.
Throughout this paper, we denote by $\mathcal{G}(\eta,\mu,\rho)$ the class of all graphs admitting balanced $(\eta,\mu,\rho)$-separators.
Trivially, every graph admits balanced $(1,0,1)$-separators.
Thus, we only consider the case $\rho < 1$.
The main result of this paper is the following theorem.

\begin{theorem} \label{thm-main}
    Let $\mathcal{G} \subseteq \mathcal{G}(\eta,\mu,\rho)$ where $\eta,\mu \geq 0$ and $0 \leq \rho < 1$.
    Also, let $\mathcal{F}$ be a finite set of graphs.
    Then there exists a constant $c < 1$ (depending on $\eta$, $\mu$, $\rho$, and $\mathcal{F}$) such that the \textsc{Weighted $\mathcal{F}$-Hitting} problem on $\mathcal{G}$ can be solved in $2^{O(k^c)} \cdot n + O(m)$ time.
\end{theorem}

We now briefly discuss why the graph classes in Theorem~\ref{thm-main} are interesting.
The sub-classes of $\mathcal{G}(\eta,0,\rho)$ for $\rho < 1$ are known as graph classes with \textit{strongly sublinear separators}.
Dvo\v{r}\'{a}k and Norin~\cite{DorakN2016} showed that these classes are equivalent to graph classes of \textit{polynomial expansion}, which is a well-studied special case of \textit{bounded-expansion} graph classes introduced by Ne{\v{s}}et{\v{r}}il and Ossona De Mendez \cite{nevsetvril2008grad,nevsetvril2008grad2}.
Important examples of polynomial-expansion graph classes include planar graphs, bounded-genus graphs, minor-free graphs, $k$-nearest neighbor graphs~\cite{miller1997separators}, greedy Euclidean spanners~\cite{le2022greedy}, etc.
For a general $\mu \geq 0$, the classes $\mathcal{G}(\eta,\mu,\rho)$ with $\rho < 1$ in addition subsume a broad family of geometric intersection graphs.
A \textit{geometric intersection graph} is defined by a set of geometric objects in a Euclidean space, where the objects are the vertices and two vertices are connected by an edge if their corresponding objects intersect.
Intersection graphs of any fat objects (i.e., convex geometric objects whose diameter-width ratio is bounded) and intersection graphs of pseudo-disks (i.e., topological disks in the plane satisfying that the boundaries of every pair of them are either disjoint or intersect twice) admit balanced $(\eta,\mu,\rho)$-separators for $\rho < 1$~\cite{lokshtanov2023hitting,miller1997separators}.
These two families, in turn, cover many interesting cases of geometric intersection graphs, such as (unit-)disk graphs, ball graphs, hypercube graphs, map graphs, intersection graphs of non-crossing rectangles, etc.
To summarize, the graph classes $\mathcal{G}(\eta,\mu,\rho)$ for $\rho < 1$ cover many common graph classes on which subexponential algorithms have been well-studied.

It is natural to ask whether the result in Theorem~\ref{thm-main} can be extended to \textsc{$\mathcal{F}$-Hitting} with an infinite $\mathcal{F}$.
A typical example is \textsc{Feedback Vertex Set}, where $\mathcal{F}$ consists of all cycles.
Unfortunately, this seems impossible.
Indeed, it was known~\cite{FominLS18} that \textsc{Feedback Vertex Set} does not admit subexponential FPT algorithms on unit-ball graphs with bounded ply, which belong to the class $\mathcal{G}(\eta,0,\rho)$ for some $\eta$ and $\rho < 1$.
Another extension one may consider is to the \textsc{Induced $\mathcal{F}$-Hitting} problem, where the goal is to hit all \textit{induced} subgraphs of $G$ that are isomorphic to some graph in $\mathcal{F}$.
Again, Theorem~\ref{thm-main} is unlikely to hold for \textsc{Induced $\mathcal{F}$-Hitting}.
Indeed, when $\mathcal{F}$ consists of a single edgeless graph of $p$ vertices, \textsc{Induced $\mathcal{F}$-Hitting} (with $k=0$) is equivalent to detecting an independent set of size $p$ and thus an algorithm for \textsc{Induced $\mathcal{F}$-Hitting} on a graph class $\mathcal{G}$ with running time as in Theorem~\ref{thm-main} would imply an FPT algorithm for \textsc{Independent Set} on $\mathcal{G}$.
But \textsc{Independent Set} is W[1]-hard even on unit-disk graphs~\cite{CyganFKLMPPS15}.
Note that while this rules out the possibility of extending Theorem~\ref{thm-main} to \textsc{Induced $\mathcal{F}$-Hitting}, it might still be possible to obtain a weaker bound of $2^{O(k^c)} \cdot n^{f(\mathcal{F})}$ for $c<1$ and some function $f$, which we leave as an interesting open question for future study.


\subsection{Our framework and building blocks}
In this section, we discuss our algorithmic framework for achieving Theorem~\ref{thm-main} and its technical components.
Let $\mathcal{G} \subseteq \mathcal{G}(\eta,\mu,\rho)$ where $\rho < 1$.
In a high level, our framework solves a \textsc{Weighted $\mathcal{F}$-Hitting} instance $(G,w,k)$ with $G \in \mathcal{G}$ through three general steps:
\begin{enumerate}
    \item[\textbf{1.}] Reduce the size of the problem instance to $k^{O(1)}$ in linear time.
    \item[\textbf{2.}] Reduce the \textsc{$\mathcal{F}$-Hitting} instance to a subexponential number of general (weighted) hitting-set instances, where the Gaifman graph of each instance has treewidth sublinear in $k$.
    \item[\textbf{3.}] Solve each hitting-set instance efficiently using the sublinear treewidth of its Gaifman graph.
\end{enumerate}

\noindent
Among the three steps, the second one is the main step which is also the most difficult one, while the third step is standard.
In what follows, we discuss these steps in detail.

To achieve Step~1, we give a polynomial kernel for the \textsc{$\mathcal{F}$-Hitting} problem on $\mathcal{G}$ that runs in \textit{linear} time.
Specifically, we show that one can compute in linear time an induced subgraph $G'$ of the input graph $G$ with size $k^{O(1)}$ such that solving the problem on $G'$ is already sufficient for solving the problem on $G$.
We say a set $S \subseteq V(G)$ is an \textit{$\mathcal{F}$-hitting set} of $G$ if $G - S$ does not contain any graph in $\mathcal{F}$ as a subgraph.
Formally, we prove the following theorem.

\begin{restatable}{theorem}{kernel} \label{thm-kernel}
Let $\mathcal{G} \subseteq \mathcal{G}(\eta,\mu,\rho)$ where $\eta,\mu \geq 0$ and $\rho < 1$.
Also, let $\mathcal{F}$ be a finite set of graphs.
There exists an algorithm that, for a given graph $G \in \mathcal{G}$ of $n$ vertices and $m$ edges together with a number $k \in \mathbb{N}$, computes in $k^{O(1)} \cdot n+O(m)$ time an induced subgraph $G'$ of $G$ with $|V(G')| = k^{O(1)}$ such that any $\mathcal{F}$-hitting set $S \subseteq V(G')$ of $G'$ with $|S| \leq k$ is also an $\mathcal{F}$-hitting set of $G$.
\end{restatable}

Note that polynomial kernels for \textsc{$d$-Hitting Set} are well-known~\cite{CyganFKLMPPS15}.
Unfortunately, we cannot apply these kernels to obtain Theorem~\ref{thm-kernel}.
The main reason is what we already mentioned in the introduction: in the \textsc{$\mathcal{F}$-Hitting} problem, the sets to be hit are implicit, and the number of these sets can be $n^{O(\gamma)}$, where $\gamma = \max_{F \in \mathcal{F}} |V(F)|$, so that we cannot afford to compute all of them.
Besides this, another (less serious) difficulty  here is that the instance obtained by the kernelization algorithm is required to be another \textsc{$\mathcal{F}$-Hitting} instance (whose underlying graph is an induced subgraph of $G$) rather than a general \textsc{$d$-Hitting Set} instance.
Our approach for Theorem~\ref{thm-kernel} is a variant of the sunflower-based kernel for \textsc{$d$-Hitting Set}, which can overcome these difficulties when properly combined with the linear-time first-order model checking algorithm of Dvo\v{r}\'{a}k et al.~\cite{dvovrak2013testing}.
Thanks to Theorem~\ref{thm-kernel}, it suffices to design algorithms with subexponential \textit{XP} running time, i.e., $n^{O(k^c)}$ time for $c < 1$.

Step~2 is the core of our framework.
It is achieved by an intricate branching algorithm, which is the main technical contribution of this paper.
Let $\mathcal{X}$ be a collection of sets (in which the elements belong to the same universe).
The \textit{Gaifman graph} of $\mathcal{X}$ is the graph with vertex set $\bigcup_{X \in \mathcal{X}} X$ where two vertices $u$ and $v$ are connected by an edge if $u,v \in X$ for some $X \in \mathcal{X}$.
We say a set $S$ \textit{hits} $\mathcal{X}$ if $S \cap X \neq \emptyset$ for all $X \in \mathcal{X}$.
The algorithm for Step~2 is stated in the following theorem.

\begin{restatable}{theorem}{branching} \label{thm-subtwbranch}
Let $\mathcal{G} \subseteq \mathcal{G}(\eta,\mu,\rho)$ where $\eta,\mu \geq 0$ and $\rho < 1$.
Also, let $\mathcal{F}$ be a finite set of graphs.
Then there exists a constant $c < 1$ (depending on $\eta$, $\mu$, $\rho$, and $\mathcal{F}$) such that for a given graph $G \in \mathcal{G}$ of $n$ vertices and a parameter $k \in \mathbb{N}$, one can construct in $n^{O(k^c)}$ time $t = 2^{O(k^c)}$ collections $\mathcal{X}_1,\dots,\mathcal{X}_t$ of subsets of $V(G)$ satisfying the following conditions.
\begin{itemize}
    \item For any $S \subseteq V(G)$ with $|S| \leq k$, $S$ is an $\mathcal{F}$-hitting set of $G$ iff $S$ hits $\mathcal{X}_i$ for some $i \in [t]$.
    \item The Gaifman graph of $\mathcal{X}_i$ has treewidth $O(k^c)$, for all $i \in [t]$.
    \item $|\mathcal{X}_i| = k^{O(1)}$ for all $i \in [t]$.
\end{itemize}
\end{restatable}

\vspace{0.1cm}

\noindent
The proof of Theorem~\ref{thm-subtwbranch} is highly nontrivial, which combines theories of sparse graphs, sunflowers, tree decomposition, branching algorithms, together with novel insights to the \textsc{$\mathcal{F}$-Hitting} problem itself.
In the proof, we introduce interesting combinatorial results (e.g., Lemma~\ref{lem-twGaifman}) for the set systems considered in the \textsc{$\mathcal{F}$-Hitting} problem (i.e., the vertex sets of all subgraphs isomorphic to some graph in $\mathcal{F}$), which are of independent interest and might be useful for understanding the structure of such set systems.
Using Theorem~\ref{thm-kernel}, one can straightforwardly improve the running time of Theorem~\ref{thm-subtwbranch} to $2^{O(k^c)} \cdot n + O(m)$, yielding the following result.
\begin{corollary} \label{cor-branching}
Let $\mathcal{G} \subseteq \mathcal{G}(\eta,\mu,\rho)$ where $\eta,\mu \geq 0$ and $\rho < 1$.
Also, let $\mathcal{F}$ be a finite set of graphs.
Then there exists a constant $c < 1$ (depending on $\eta$, $\mu$, $\rho$, and $\mathcal{F}$) such that for a given graph $G \in \mathcal{G}$ of $n$ vertices and a parameter $k \in \mathbb{N}$, one can construct in $2^{O(k^c)} \cdot n + O(m)$ time $t = 2^{O(k^c)}$ collections $\mathcal{X}_1,\dots,\mathcal{X}_t$ of subsets of $V(G)$ satisfying the following conditions.
\begin{itemize}
    \item For any $S \subseteq V(G)$ with $|S| \leq k$, $S$ is an $\mathcal{F}$-hitting set of $G$ iff $S$ hits $\mathcal{X}_i$ for some $i \in [t]$.
    \item The Gaifman graph of $\mathcal{X}_i$ has treewidth $O(k^c)$, for all $i \in [t]$.
    \item $|\mathcal{X}_i| = k^{O(1)}$ for all $i \in [t]$.    
\end{itemize}
\end{corollary}
\vspace{0.01cm}
\begin{proof}
We first apply Theorem~\ref{thm-kernel} on $G$ and $k$ to obtain the induced subgraph $G'$.
Then we apply Theorem~\ref{thm-subtwbranch} on $G'$ and $k$ to obtain the collections $\mathcal{X}_1,\dots,\mathcal{X}_t$ of subsets of $V(G')$, which are also subsets of $V(G)$.
Since $|V(G')| = k^{O(1)}$ by Theorem~\ref{thm-kernel}, the total time cost is $k^{O(1)} \cdot n + O(m) + k^{O(k^p)}$ for some constant $p < 1$.
Choosing an arbitrary constant $c \in (p,1)$, the running time is bounded by $2^{O(k^c)} \cdot n + O(m)$.
The bounds on the sizes of $\mathcal{X}_1,\dots,\mathcal{X}_t$ and the treewidth of the Gaifman graphs directly follow from Theorem~\ref{thm-subtwbranch}.
It suffices to show that a set $S \subseteq V(G)$ with $|S| \leq k$ is an $\mathcal{F}$-hitting set of $G$ iff $S$ hits $\mathcal{X}_i$ for some $i \in [t]$.
If $S$ is an $\mathcal{F}$-hitting set of $G$, then $S \cap V(G')$ is an $\mathcal{F}$-hitting set of $G'$.
Thus $S \cap V(G')$ hits $\mathcal{X}_i$ for some $i \in [t]$ by Theorem~\ref{thm-subtwbranch}, which implies that $S$ hits $\mathcal{X}_i$.
On the other hand, if $S$ hits $\mathcal{X}_i$ for some $i \in [t]$, then $S$ is an $\mathcal{F}$-hitting set of $G'$ by Theorem~\ref{thm-subtwbranch}, which is in turn an $\mathcal{F}$-hitting set of $G$ by Theorem~\ref{thm-kernel} since $|S| \leq k$.
\end{proof}

Step~3 is achieved by standard dynamic programming on tree decomposition.
For each hitting-set instance $\mathcal{X}_i$ obtained in Corollary~\ref{cor-branching}, we build a tree decomposition of the Gaifman graph of $\mathcal{X}_i$ of width $O(k^c)$, which can be done using standard constant-approximation algorithms for treewidth~\cite{CyganFKLMPPS15}.
Then we apply dynamic programming to compute a minimum-weight (under the weight function $w$) hitting set $S_i$ for $\mathcal{X}_i$ satisfying $|S_i| \leq k$.
Finally, we return the set $S_i$ with the minimum total weight among $S_1,\dots,S_t$.
The first condition in Corollary~\ref{cor-branching} guarantees that $S_i$ is an optimal solution for the \textsc{Weighted $\mathcal{F}$-Hitting} instance $(G,w,k)$.

\paragraph{Organization.}
In Section~\ref{sec-overview}, we give an informal overview of our techniques.
In Section~\ref{sec-pre}, we give the formal definitions of the basic notions used in this paper.
Sections~\ref{sec-kernel} and~\ref{sec-branching} present the proofs of Theorems~\ref{thm-kernel} and~\ref{thm-subtwbranch}, respectively.
Section~\ref{sec-subhit} proves our main result, Theorem~\ref{thm-main}.
Finally, in Section~\ref{sec-conclusion}, we give a conclusion for the paper and propose some open problems.

\section{Overview of our techniques} \label{sec-overview}

In this section, we give an overview of the ideas/techniques underlying our proofs. 
The kernelization algorithm in Theorem~\ref{thm-kernel} is relatively simple.
So we only focus on our main technical result, the reduction algorithm in Theorem~\ref{thm-subtwbranch}.
Before the discussion, let us recall a standard notion called \textit{weak coloring numbers}.
Let $G$ be a graph and $\sigma$ be an ordering of $V(G)$.
For $u,v \in V(G)$, we write $u <_\sigma v$ if $u$ is before $v$ under the ordering $\sigma$.
The notations $>_\sigma$, $\leq_\sigma$, $\geq_\sigma$ are defined similarly.
For an integer $r \geq 0$, $u$ is \textit{weakly $r$-reachable} from $v$ \textit{under $\sigma$} if there is a path $\pi$ between $v$ and $u$ of length at most $r$ such that $u$ is the largest vertex on $\pi$ under the ordering $\sigma$, i.e., $u \geq_\sigma w$ for all $w \in V(\pi)$.
Let $\text{WR}_r(G,\sigma,v)$ denote the set of vertices in $G$ that are weakly $r$-reachable from $v$ under $\sigma$.
The \textit{weak $r$-coloring number} of $G$ \textit{under $\sigma$} is defined as $\mathsf{wcol}_r(G,\sigma) = \max_{v \in V(G)} |\text{WR}_r(G,\sigma,v)|$.
Then the \textit{weak $r$-coloring number} of $G$ is defined as $\mathsf{wcol}_r(G) = \min_{\sigma \in \Sigma(G)} \mathsf{wcol}_r(G,\sigma)$ where $\Sigma(G)$ is the set of all orderings of $V(G)$.
It is well-known~\cite{nevsetvril2008grad,zhu2009colouring} that a graph class $\mathcal{G}$ is of bounded expansion iff there is a function $f:\mathbb{N} \rightarrow \mathbb{N}$ such that $\mathsf{wcol}_r(G) \leq f(r)$ for all $G \in \mathcal{G}$ and all $r \in \mathbb{N}$.

For simplicity, we only discuss the algorithm of Theorem~\ref{thm-subtwbranch} in a special case where $\mu = 0$ (i.e., $\mathcal{G}$ is of polynomial expansion) and $\mathcal{F}$ only contains a single graph $F$.
The same algorithm directly extends to the case where $\mathcal{F}$ consists of multiple graphs.
Further extending it to a general $\mathcal{G} \subseteq \mathcal{G}(\eta,\mu,\rho)$ is also not difficult, by adapting some standard techniques.
Note that if $\mathcal{G} \subseteq \mathcal{G}(\eta,0,\rho)$, then for any fixed $r \in \mathbb{N}$, we have $\mathsf{wcol}_r(G) = O(1)$ for all $G \in \mathcal{G}$.

The first simple observation we have is that one can assume $F$ is connected without loss of generality.
To see this, consider a graph $G \in \mathcal{G}$.
We say a graph is \textit{$F$-free} if it does not contain any subgraph isomorphic to $F$.
Let $G^+$ (resp., $F^+$) be the graph obtained from $G$ (resp., $F$) by adding a new vertex with edges to all other vertices.
Now $G^+ \in \mathcal{G}(\eta+1,0,\rho)$ and $F^+$ is connected.
Furthermore, one can easily verify that for any $S \subseteq V(G)$, $G^+ - S$ is $F^+$-free iff $G - S$ is $F$-free.
As long as Theorem~\ref{thm-subtwbranch} works for $\mathcal{G}(\eta+1,0,\rho)$ and $\mathcal{F}^+ = \{F^+\}$, we can apply it to compute the collections $\mathcal{X}_1,\dots,\mathcal{X}_t$ of subsets of $V(G^+)$ satisfying the desired conditions.
Define $\mathcal{X}_i' = \{X \cap V(G): X \in \mathcal{X}_i\}$ for $i \in [t]$.
It turns out that $\mathcal{X}_1',\dots,\mathcal{X}_t'$ are collections of subsets of $V(G)$ that satisfy the desired conditions for $G$ and $\mathcal{F}$.

An \textit{$F$-copy} in $G$ refers to a pair $(H,\pi)$ where $H$ is a subgraph of $G$ and $\pi:V(H) \rightarrow V(F)$ is an isomorphism between $H$ and $F$.
As is usual in hitting-set problems~\cite{impagliazzo2001problems,santhanam2012limits}, we consider and branch on \textit{sunflowers} in the set system (which in our setting consists of the vertex sets of the $F$-copies in $G$).
Recall that sets $V_1,\dots,V_r$ form a \textit{sunflower} if there exists a set $X$ such that $X \subseteq V_i$ for all $i \in [r]$ and $V_1 \backslash X,\dots,V_r \backslash X$ are disjoint; $X$ is called the \textit{core} of the sunflower and $r$ is the size of the sunflower. 

As the sets in our problem are vertex sets of subgraphs of $G$, it is more convenient to consider sunflowers with additional structures related to the graph.
We say $F$-copies $(H_1,\pi_1),\dots,(H_r,\pi_r)$ in $G$ form a \textit{structured sunflower} if $V(H_1),\dots,V(H_r)$ form a sunflower with core $X$ and $(\pi_1)_{|X} = \cdots = (\pi_r)_{|X}$; the core of the structured sunflower is the pair $(X,f)$ where $f:X \rightarrow V(F)$ is the unique map satisfying $f = (\pi_1)_{|X} = \cdots = (\pi_r)_{|X}$.
A pair $(X,f)$ is a \textit{heavy core} if it is the core of a structured sunflower of size $\gamma^{|X|} \delta$, where $\gamma = |V(F)|$ and $\delta$ is a parameter to be determined.
The reason why we pick $\gamma^{|X|} \delta$ as the threshold will be clear later.
Note that we do not require the $F$-copies in a structured sunflower to be distinct.
Therefore, for an $F$-copy $(H,\pi)$ in $G$, the pair $(V(H),\pi)$ is a heavy core, because it is the core of the structured sunflower $(H_1,\pi_1),\dots,(H_\Delta,\pi_\Delta)$ where $\Delta = \gamma^{|V(H)|} \delta = \gamma^\gamma \delta$ and $(H_1,\pi_1)=\cdots=(H_\Delta,\pi_\Delta) = (H,\pi)$.

To get some basic idea about how $\mathcal{X}_1,\dots,\mathcal{X}_t$ in Theorem~\ref{thm-subtwbranch} are generated, let us first consider a trivial branching algorithm, which can generate $\mathcal{X}_1,\dots,\mathcal{X}_t$ satisfying the first condition in Theorem~\ref{thm-subtwbranch} without any guarantee on the running time, the number $t$, or the treewidth.
Imagine there is some (unknown) $\mathcal{F}$-hitting set $S$ of $G$.
The branching algorithm essentially guesses whether each heavy core is hit by the solution or not.
It maintains a set $U \subseteq V(G)$ and a collection $\mathcal{X}$ of subsets of $V(G)$.
The vertices in $U$ are ``undeletable'' vertices, namely, the vertices that are not supposed to be in $S$.
On the other hand, the sets in $\mathcal{X}$ are supposed to be hit by $S$.
Initially, $U = \emptyset$ and $\mathcal{X} = \emptyset$.
Then it calls the function $\textsc{Branch}(U,\mathcal{X})$, which works as follows.
\begin{itemize}
    \item Pick a heavy core $(X,f)$ satisfying $X \nsubseteq U$ and $X \notin \mathcal{X}$.
    If such a heavy core does not exist and $G[U]$ is $F$-free, then create a new $\mathcal{X}_i = \{X \backslash U: X \in \mathcal{X}\}$, and return to the last level.
    \item Branch on $(X,f)$ in two ways (i.e., guess whether $S$ hits $X$ or not):
    \begin{itemize}
        \item ``Yes'' branch (guess $S \cap X \neq \emptyset$): recursively call $\textsc{Branch}(U,\mathcal{X} \cup \{X\})$.
        \item ``No'' branch (guess $S \cap X = \emptyset$): recursively call $\textsc{Branch}(U \cup X,\mathcal{X})$.
    \end{itemize}
\end{itemize}

Let $\mathcal{X}_1,\dots,\mathcal{X}_t$ be the collection generated by the above procedure.
One can easily check that a subset $S \subseteq V(G)$ is an $\mathcal{F}$-hitting set of $G$ iff $S$ hits $\mathcal{X}_i$ for some $i \in [t]$.
Indeed, if $S$ is an $\mathcal{F}$-hitting set of $G$ and the algorithm makes the right decision in each step (i.e., makes the ``yes'' decision whenever $S \cap X \neq \emptyset$ and the ``no'' decision whenever $S \cap X = \emptyset$), at the end of that branch a collection $\mathcal{X}_i = \{Y \backslash U: Y \in \mathcal{X}\}$ is created and $S$ hits $\mathcal{X}_i$.
On the other hand, if $S \subseteq V(G)$ is a set that hits some $\mathcal{X}_i$, then $S$ hits all $F$-copies in $G$.
Why? Note that $\mathcal{X}_i = \{X \backslash U: X \in \mathcal{X}\}$.
At the point we create $\mathcal{X}_i$, we have $X \in \mathcal{X}$ for all heavy cores $(X,f)$ with $X \nsubseteq U$, and in particular $V(H) \in \mathcal{X}$ for all $F$-copies $(H,\pi)$ with $V(H) \nsubseteq U$.
By assumption, $S$ hits these $F$-copies.
Furthermore, as $G[U]$ is $F$-free, there is no $F$-copy $(H,\pi)$ with $V(H) \subseteq U$.
Thus, $S$ hits all $F$-copies in $G$.

Of course, this trivial branching procedure can only provide us some intuition about how the collections $\mathcal{X}_1,\dots,\mathcal{X}_t$ are generated.
It guarantees neither sublinear treewidth of the Gaifman graphs of $\mathcal{X}_1,\dots,\mathcal{X}_t$ nor subexponential bound on $t$.
Thus, we are still far from proving Theorem~\ref{thm-subtwbranch}.
In the following two sections, we shall focus on how to achieve sublinear treewidth (Section~\ref{sec-subltw}) and subexponential branching (Section~\ref{sec-subexpbranch}), respectively.
Both parts are quite technical.
Interestingly, to obtain sublinear treewidth, we only need to slightly modify the trivial branching algorithm, and the main challenge lies in the proof of a structural lemma for the Gaifman graph of certain heavy cores (Lemma~\ref{Olem-treewidth}).
For subexponential branching, however, we have to further elaborate the branching algorithm significantly, with an involved analysis.

In this overview, we ignore the requirements of Theorem~\ref{thm-subtwbranch} on the running time of the algorithm and the sizes of $\mathcal{X}_1,\dots,\mathcal{X}_t$.
It turns out that these requirements can be achieved almost for free as long as the algorithm admits a subexponential branching tree.

\subsection{How to achieve sublinear treewidth} \label{sec-subltw}
Recall that we want the Gaifman graph of each $\mathcal{X}_i$ to have treewidth sublinear in $k$.
To see how to achieve the sublinear treewidth bound, let us introduce some additional notions.
For two heavy cores $(X,f)$ and $(Y,g)$, we write $(X,f) \prec (Y,g)$ if $X \subsetneq Y$ and $f = g_{|X}$, and write $(X,f) \preceq (Y,g)$ if $(X,f) \prec (Y,g)$ or $(X,f) = (Y,g)$.
Clearly, $\prec$ is a partial order among all heavy cores.
For a subset $U \subseteq V(G)$, we say a heavy core $(X,f)$ is \textit{$U$-minimal} if $X \nsubseteq U$ and for any heavy core $(Y,g)$ with $(Y,g) \prec (X,f)$, we have $Y \subseteq U$.
The key to achieve sublinear treewidth is the following important structural lemma for $U$-minimal heavy cores.

\begin{lemma}[simplified version of Lemma~\ref{lem-twGaifman}] \label{Olem-treewidth}
    Suppose $\delta > \textnormal{wcol}_\gamma(G)$.
    Then for any subset $U \subseteq V(G)$ and any $U$-minimal heavy cores $(X_1,f_1),\dots,(X_r,f_r)$ in $G$, the Gaifman graph of $\{X_1 \backslash U,\dots,X_r \backslash U\}$ has treewidth $\delta^{O(1)} \cdot k^c$ for some constant $c < 1$, where $k$ is the size of a minimum hitting set of $\{X_1 \backslash U,\dots,X_r \backslash U\}$.
    Here $c$ and the constant hidden in $O(\cdot)$ only depend on $F$ and the polynomial-expansion graph class $\mathcal{G}$ from which $G$ is drawn.
\end{lemma}


The proof of Lemma~\ref{Olem-treewidth} is technical.
Before giving a sketch of the proof, we first explain how this lemma helps us.
Recall the branching procedure discussed before.
The first observation is that we actually only need to branch on $U$-minimal heavy cores.
Specifically, we require the heavy core $(X,f)$ picked in the first step of $\textsc{Branch}(U,\mathcal{X})$ to be \textit{$U$-minimal}.
With this modification, the collections $\mathcal{X}_1,\dots,\mathcal{X}_t$ generated still satisfy the first condition in Theorem~\ref{thm-subtwbranch}, because a set $S \subseteq V(G) \backslash U$ hits all heavy cores if and only if it hits all $U$-minimal heavy cores.
More importantly, the Gaifman graph of each $\mathcal{X}_i$ has treewidth sublinear in the size of a minimum hitting set of $\mathcal{X}_i$, by Lemma~\ref{Olem-treewidth}.
The second observation is that if the size of a minimum hitting set of a collection $\mathcal{X}_i$ is larger than $k$, then we can simply discard $\mathcal{X}_i$.
This is because the first condition in Theorem~\ref{thm-subtwbranch} only considers $S \subseteq V(G)$ with $|S| \leq k$.
Therefore, only keeping the collections $\mathcal{X}_i$ whose minimum hitting set has size at most $k$ is sufficient.
In this way, we can guarantee that the Gaifman graphs of all $\mathcal{X}_i$ have treewidth sublinear in $k$, as required in Theorem~\ref{thm-subtwbranch}.

\paragraph{Proof sketch of Lemma~\ref{Olem-treewidth}.}
In the rest of this section, we provide a high-level overview for the proof of Lemma~\ref{Olem-treewidth}.
We first need to show the following auxiliary lemma, which essentially states that whenever there are many heavy cores $(X_1,f_1),\dots,(X_p,f_p)$ forming a large sunflower, one can find certain smaller heavy cores inside each $X_i$.
This lemma heavily relies on the threshold $\gamma^{|X|} \delta$ we chose for a heavy core $(X,f)$.

\begin{lemma}[simplified version of Lemma~\ref{lem-heavysep}] \label{Olem-component}
    Let $p = \gamma^\gamma \delta$.
    Suppose $(X_1,f_1),\dots,(X_p,f_p)$ are heavy cores in $G$ satisfying that $|X_1| = \cdots = |X_p|$ and $X_1,\dots,X_p$ form a sunflower with core $K$ where $(f_1)_{|K} = \cdots = (f_p)_{|K}$.
    Define $f: K \rightarrow V(F)$ as the unique map satisfying $f = (f_1)_{|K} = \cdots = (f_p)_{|K}$.
    Then for any $i \in [p]$ and any set $\mathcal{C}$ of connected components of $F - f(K)$, $(X_i^{\mathcal{C}},f_i^{\mathcal{C}})$ is a heavy core in $G$ where $X_i^{\mathcal{C}} = K \cup (\bigcup_{C \in \mathcal{C}} f_i^{-1}(V(C)))$ and $f_i^{\mathcal{C}} = (f_i)_{|X_i^{\mathcal{C}}}$.
\end{lemma}
\begin{proof}[Proof sketch.]
We only need to consider the pair $(X_1^{\mathcal{C}},f_1^{\mathcal{C}})$.
If $\mathcal{C}$ contains all connected components of $F - f(K)$, then $(X_1^{\mathcal{C}},f_1^{\mathcal{C}}) = (X_1,f_1)$ and we are done.
Otherwise, $|X_1^{\mathcal{C}}| < |X_1| = \cdots = |X_p|$.
Our goal is to find a structured sunflower $(H_1,\pi_1),\dots,(H_\Delta,\pi_\Delta)$ in $G$ for $\Delta = \gamma^{|X_1^{\mathcal{C}}|} \delta$ whose core is $(X_1^{\mathcal{C}},f_1^{\mathcal{C}})$.
We say an $F$-copy $(H,\pi)$ is a \textit{candidate} if $X_1^{\mathcal{C}} \subseteq V(H)$ and $f_1^{\mathcal{C}} = \pi_{|X_1^{\mathcal{C}}}$.
Then our task becomes finding $\Delta$ candidates whose vertex sets are disjoint outside $X_1^{\mathcal{C}}$.

Where do these candidates come from?
In fact, we can construct them from the structured sunflowers that witness the heavy cores $(X_1,f_1),\dots,(X_p,f_p)$.
Let $A = f(K) \cup (\bigcup_{C \in \mathcal{C}} V(C))$ and $B = f(K) \cup (V(F) \backslash A)$.
Then $V(F)$ is the disjoint union of $A \backslash B$, $f(K)$, and $B \backslash A$.
Also, note that there is no edge in $F$ between $A \backslash B$ and $B \backslash A$.
Consider an $F$-copy $(P,\phi)$ in the structured sunflower that witnesses $(X_1,f_1)$ and another $F$-copy $(Q,\psi)$ in the structured sunflower that witnesses $(X_i,f_i)$ for some $i \in [p]$.
The key observation is the following: if $\phi^{-1}(A \backslash B) \cap \psi^{-1}(B \backslash A) = \emptyset$, then $H = P[\phi^{-1}(A)] \cup Q[\psi^{-1}(B)]$ is isomorphic to $F$ with the isomorphism $\pi: V(H) \rightarrow V(F)$ defined as $\pi(v) = \phi(v)$ for $v \in \pi^{-1}(A)$ and $\pi(v) = \psi(v)$ for $v \in \pi^{-1}(B)$ --- one can easily verify that $\pi$ is well-defined and is an isomorphism --- and furthermore $(H,\pi)$ is a candidate (for convenience, we call $P[\phi^{-1}(A)]$ the \textit{$A$-half} and $Q[\psi^{-1}(B)]$ the \textit{$B$-half} of the candidate).
We use this observation to construct the candidates.
The structured sunflower that witnesses $(X_1,f_1)$ has size $\gamma^{|X_1|} \delta > \Delta$, so we can take $\Delta$ $F$-copies $(P_1,\phi_1),\dots,(P_\Delta,\phi_\Delta)$ from it.
Then we use $P_1[\phi_1^{-1}(A)],\dots,P_\Delta[\phi_\Delta^{-1}(A)]$ as the $A$-halves of the candidates, each of which will be ``glued'' with a $B$-half to obtain a complete candidate.
We briefly discuss how to find the $B$-halves.
Outside $K$, the $B$-halves should be disjoint from each other and also disjoint from the $A$-halves.
Suppose we already found the $B$-halves for $P_1[\phi_1^{-1}(A)],\dots,P_{i-1}[\phi_{i-1}^{-1}(A)]$ and are now looking for the $B$-half for $P_i[\phi_i^{-1}(A)]$.
We say a set $S \subseteq V(G)$ is \textit{safe} if it is disjoint from all $A$-halves and the $B$-halves we have found.
We first find $j \in [p]$ such that $X_j \backslash K$ is safe.
Such an index exists since $X_1 \backslash K,\dots,X_p \backslash K$ are disjoint and $p$ is sufficiently large.
Then we further find an $F$-copy $(Q,\psi)$ in the structured sunflower that witnesses $(X_j,f_j)$ such that $V(Q) \backslash X_j$ is safe.
This is possible because the size of the structured sunflower is $\gamma^{|X_j|} \delta$, which is much larger than $\Delta = \gamma^{|X_1^{\mathcal{C}}|} \delta$ as $|X_1^{\mathcal{C}}| < |X_j|$.
Now we use $Q[\psi^{-1}(B)]$ as the $B$-half for $P_i[\phi_i^{-1}(A)]$.
In this way, we can successfully find all $B$-halves.
\end{proof}

Note that the above lemma directly implies the ``sparseness'' of $U$-minimal heavy cores outside $U$: every vertex in $V(G) \backslash U$ hits at most $\delta^{O(1)}$ (distinct) $U$-minimal heavy cores in $G$.
Why? Suppose a vertex $v \in V(G) \backslash U$ hits too many $U$-minimal heavy cores.
By the sunflower lemma and Pigeonhole principle, among these $U$-minimal heavy cores, we can find $(X_1,f_1),\dots,(X_p,f_p)$ satisfying the condition in the lemma.
The core $K$ of the sunflower $X_1,\dots,X_p$ is nonempty as $v \in \bigcap_{i=1}^p X_i = K$.
Also, $K \nsubseteq U$, since $v \in K \backslash U$.
Applying the lemma with $\mathcal{C} = \emptyset$, we see that $(K,f)$ is a heavy core, which contradicts the $U$-minimality of $(X_1,f_1),\dots,(X_p,f_p)$ because $(K,f) \prec (X_i,f_i)$ for all $i \in [p]$.
We omit the calculation for the maximum number of $U$-minimal heavy cores $v$ can hit, but the number turns out to be $\delta^{O(1)}$.

Now we sketch the proof of Lemma~\ref{Olem-treewidth}.
Let $(X_1,f_1),\dots,(X_r,f_r)$ be as in the lemma, and $k$ be the size of a minimum hitting set of $\{X_1 \backslash U,\dots,X_r \backslash U\}$.
As argued above, each vertex can hit $\delta^{O(1)}$ sets in $\{X_1 \backslash U,\dots,X_r \backslash U\}$, which implies $|\bigcup_{i=1}^r (X_i \backslash U)| = \delta^{O(1)} k$.
Fix an ordering $\sigma$ of the vertices of $G$ such that $\text{wcol}_\gamma(G,\sigma) = \text{wcol}_\gamma(G)$.
Let $G'$ be a supergraph of $G$ obtained by adding edges to connect pairs of vertices in which one is weakly $\gamma$-reachable (under $\sigma$) from the other, i.e., $V(G') = V(G)$ and $E(G') = E(G) \cup \{(u,v): u \in \text{WR}_\gamma(G,\sigma,v)\}$.
It turns out that the graph $G'$ also admits strongly sublinear separators, which implies the treewidth of $G'[\bigcup_{i=1}^r (X_i \backslash U)]$ is sublinear in $|\bigcup_{i=1}^r (X_i \backslash U)| = \delta^{O(1)} k$, i.e., bounded by $\delta^{O(1)} k^c$ for some $c < 1$.
Let $(T,\beta)$ be a minimum-width tree decomposition of $G'[\bigcup_{i=1}^r (X_i \backslash U)]$.
Our goal is to modify $(T,\beta)$ to a tree decomposition of the Gaifman graph $G^*$ of $\{X_1 \backslash U,\dots,X_r \backslash U\}$, without increasing its width too much.
The modification is done as follows.
For each $i \in [r]$, we pick a node $t_i \in V(T)$ such that $\beta(t_i) \cap (X_i \backslash U) \neq \emptyset$.
For two nodes $t,t' \in V(T)$, denote by $\pi_{t,t'}$ as the (unique) path in $T$ connecting $t$ and $t'$.
Then for each $t \in V(T)$ and each $i \in [r]$, define $\beta_i^*(t)$ as the set of all vertices $v \in X_i \backslash U$ such that $t$ is on the path $\pi_{t_i,t'}$ for some node $t' \in V(T)$ with $v \in \beta(t')$.
Set $\beta^*(t) = \bigcup_{i=1}^r \beta_i^*(t)$ for all $t \in V(T)$.

It is easy to verify that $(T,\beta^*)$ is a tree decomposition of $G^*$.
The tricky part is to bound its width.
We want $|\beta^*(t)| = \delta^{O(1)} \cdot |\beta(t)|$ for every $t \in V(T)$. Fix a node $t\in V(T)$. 
Let $I^* = \{i \in [r]: \beta_i^*(t) \neq \emptyset\}$.
Since $|\beta_i^*(t)| \leq |X_i \backslash U| \leq \gamma$, we have $|\beta^*(t)| \leq \gamma |I^*|$ and thus it suffices to show $|I^*| = \delta^{O(1)} \cdot |\beta(t)|$.
Now let $I = \{i \in [r]: \beta(t) \cap (X_i \backslash U) \neq \emptyset\}$.
We have seen that each vertex in $\beta(t)$ can hit at most $\delta^{O(1)}$ sets in $\{X_1 \backslash U,\dots,X_r \backslash U\}$.
So $|I| = \delta^{O(1)} \cdot |\beta(t)|$ and we only need to show $|I^* \backslash I| = \delta^{O(1)} \cdot |\beta(t)|$.

The high-level plan for bounding $|I^* \backslash I|$ is to apply a charging argument as follows.
For each $i \in I^* \backslash I$, we shall pick two vertices $v_i,v_i' \in X_i \backslash U$ and charge $i$ to a set $Y_i \subseteq X_i$ satisfying that $f_i(Y_i)$ separates $f_i(v_i)$ and $f_i(v_i')$ in $F$, i.e., $f_i(v_i)$ and $f_i(v_i')$ belong to different connected components of $F - f_i(Y_i)$.
By a careful construction, we can guarantee that the number of distinct $Y_i$'s is small.
Then using Lemma~\ref{Olem-component}, we can show that each set $Y \subseteq V(G)$ does not get charged too many times.
These two conditions together bound the size of $I^* \backslash I$.

Consider an index $i \in I^* \backslash I$.
We have $\beta_i^*(t) \neq \emptyset$ but $\beta(t) \cap (X_i \backslash U) = \emptyset$.
By the choice of $t_i$, $\beta(t_i) \cap (X_i \backslash U) \neq \emptyset$ and so we pick a vertex $v_i \in \beta(t_i) \cap (X_i \backslash U)$.
On the other hand, as $\beta_i^*(t) \neq \emptyset$, we can pick another vertex $v_i' \in \beta_i^*(t) \subseteq X_i \backslash U$.
Note that $v_i,v_i' \notin \beta(t)$, since $\beta(t) \cap (X_i \backslash U) = \emptyset$.
By the properties of a tree decomposition, the nodes $s \in T$ with $v_i \in \beta(s)$ (resp., $v_i' \in \beta(s)$) are connected in $T$, and we call the subtree of $T$ formed by these nodes the \textit{$v_i$-area} (resp., \textit{$v_i'$-area}) for convenience.
Why do $v_i,v_i'$ appear in $\beta_i^*(t)$ but not $\beta(t)$?
The only reason is that the $v_i$-area and the $v_i'$-area belong to different connected components in the forest $T - \{t\}$.
We can show that if $\delta > \text{wcol}_\gamma(G)$, then this situation happens only when $f_i(X_i \cap \beta(t))$ separates $f_i(v_i)$ and $f_i(v_i')$ in $F$.
We omit the details of this argument.
Now a natural idea is to directly set $Y_i = X_i \cap \beta(t)$.
But this seems a bad idea, as the number of distinct $Y_i$'s cannot be bounded with this definition.
Therefore, we need to construct $Y_i$ from $X_i \cap \beta(t)$ with an additional step as follows.
Let $\varPi$ be the set of all simple paths in $F$ from $f_i(v_i)$ to a vertex in $f_i(X_i \cap \beta(t))$ in which all internal nodes are in $V(F) \backslash f_i(X_i \cap \beta(t))$.
For every $u \in X_i$, we include $u$ in $Y_i$ if there exists $\pi \in \varPi$ such that $u$ is the largest vertex (under the ordering $\sigma$) in $f_i^{-1}(V(\pi))$.
It turns out that $f_i(Y_i)$ also separates $f_i(v_i)$ and $f_i(v_i')$ in $F$.
Furthermore, $Y_i$ has a very nice property: $Y_i \subseteq \text{WR}_\gamma(G,\sigma,v)$ for some $v \in \beta(t)$.
Again, we omit the proof of this property in this overview.

With the above construction, how many distinct $Y_i$'s can there be?
For each $v \in \beta(t)$, we have $|\text{WR}_\gamma(G,\sigma,v)| \leq \text{wcol}_\gamma(G,\sigma) = \text{wcol}_\gamma(G)$.
Thus, the nice property of each $Y_i$ and the fact $|Y_i| \leq \gamma$ guarantee that the number of distinct $Y_i$'s is at most $\text{wcol}_\gamma^\gamma(G) \cdot |\beta(t)|$, which is $O(|\beta(t)|)$.
Now it suffices to bound the number of times a set $Y \subseteq V(G)$ gets charged.
The intuition is the following.
Assume there are too many indices $i \in I^* \backslash I$ that are charged to the same set $Y$, in order to deduce a contradiction.
Then among the heavy cores $(X_i,f_i)$ corresponding to the indices charged to $Y$, we can find $p = \gamma^\gamma \delta$ of them satisfying the conditions in Lemma~\ref{Olem-component}, by the sunflower lemma and Pigeonhole principle.
Without loss of generality, suppose they are $(X_1,f_1),\dots,(X_p,f_p)$, where $X_1,\dots,X_p$ form a sunflower with core $K$ and $f = (f_1)_{|K} = \cdots = (f_p)_{|K}$.
Then $Y = Y_1 = \cdots = Y_p$.
Applying Lemma~\ref{Olem-component} with $\mathcal{C} = \emptyset$, we see that $(K,f)$ is a heavy core in $G$.
Since $(K,f) \prec (X_i,f_i)$ for all $i \in [p]$, we must have $K \subseteq U$, for otherwise $(X_1,f_1),\dots,(X_p,f_p)$ are not $U$-minimal.
As all $i \in [p]$ are charged to $Y$, we have $Y \subseteq \bigcap_{i=1}^p X_i = K$.
Recall the vertices $v_i,v_i' \in X_i \backslash U$ we picked when constructing $Y_i$.
We just consider $v_1$ and $v_1'$.
Neither $f(v_1)$ nor $f(v_1')$ is contained in $f(K)$, because $K \subseteq U$ and $v_i,v_i' \in X_i \backslash U$.
Let $C$ (resp., $C'$) be the connected component of $F - f(K)$ containing $v_1$ (resp., $v_1'$).
Note that $C \neq C'$.
Indeed, $f(v_1)$ and $f(v_1')$ lie in different connected components of $F - f_1(Y_1) = F - f_1(Y)$ and thus lie in different connected components of $F - f(K)$, since $f_1(Y) \subseteq f_1(K) = f(K)$.
Set $\mathcal{C} = \{C\}$ and let $(X_1^\mathcal{C},f_1^\mathcal{C})$ as defined in Lemma~\ref{Olem-component}.
Lemma~\ref{Olem-component} shows that $(X_1^\mathcal{C},f_1^\mathcal{C})$ is a heavy core.
Observe that $X_1^\mathcal{C} \nsubseteq U$, since $v_1 \in X_1^\mathcal{C} \backslash U$.
On the other hand, $v_1' \notin X_1^\mathcal{C}$, as $C' \notin \mathcal{C}$.
Thus, $X_1^\mathcal{C} \subsetneq X_1$ and $(X_1^\mathcal{C},f_1^\mathcal{C}) \prec (X_1,f_1)$.
But this contradicts the fact that $(X_1,f_1)$ is $U$-minimal.
As a result, we see that $Y$ cannot get charged too many times.
This bounds $|I^* \backslash I|$ and hence bounds $|\beta^*(t)|$.
The final bound we achieve is $|\beta^*(t)| = \delta^{O(1)} \cdot |\beta(t)|$ for all $t \in V(T)$.
As the width of $(T,\beta)$ is $\delta^{O(1)} k^c$ for $c < 1$, the width of $(T,\beta^*)$ is also $\delta^{O(1)} k^c$, so is the treewidth of $G^*$.
This completes the overview of the proof of Lemma~\ref{Olem-treewidth}.

\subsection{How to do subexponential branching} \label{sec-subexpbranch}

Next, we discuss how to achieve the subexponential bound on the number $t$ of collections generated by our algorithm.
Note that $t$ is at most the number of leaves of the branching tree.
So the key here is to have a branching tree with subexponential size.

To get some intuition, consider a stage of our branching procedure, where we are branching on a $U$-minimal heavy core $(X,f)$.
There is a structured sunflower $(H_1,\pi_1),\dots,(H_\Delta,\pi_\Delta)$ that witnesses $(X,f)$, where $\Delta = \gamma^{|X|} \delta$.
When we make the ``yes'' decision for $(X,f)$, what we gain is that the size of $\mathcal{X}$ increases by 1.
When we make the ``no'' decision for $(X,f)$, we add all vertices in $X$ to $U$, and by doing this we also gain something: originally the $F$-copies $(H_1,\pi_1),\dots,(H_\Delta,\pi_\Delta)$ can be hit by a single vertex in $X \backslash U$, but after the vertices in $X$ are added to $U$, we have to use $\Delta$ vertices outside $U$ to hit these $F$-copies since $V(H_1),\dots,V(H_\Delta)$ form a sunflower with core $X$.
So ideally, this could make the size of a minimum $\mathcal{F}$-hitting set of $G$ contained in $V(G) \backslash U$ increase by $\Delta > \delta$.
(Clearly, this is not the case in general. But let us assume it is true just for explaining the intuition.)
It turns out that, by a corollary of Lemma~\ref{Olem-component}, the size of $\mathcal{X}$ on any successful branch path cannot exceed $\delta^{O(1)} k$.
As such, the number of ``yes'' decisions along a successful path in the branching tree is at most $\theta_\mathsf{yes} = \delta^{O(1)} k$.
On the other hand, the number of ``no'' decisions along a successful path is at most $\theta_\mathsf{no} = k/\delta$.
Indeed, every ``no'' decision increases the size of a minimum $\mathcal{F}$-hitting set of $G$ contained in $V(G) \backslash U$ by at least $\delta$, and when we need more than $k$ vertices in $V(G) \backslash U$ to hit all $F$-copies in $G$, we know that the current path is not successful.
Thus, during the branching procedure, if we have made more than $\theta_\mathsf{yes}$ ``yes'' decisions (resp., $\theta_\mathsf{no}$ ``no'' decisions), we can stop branching further.
In this way, the branching tree has size $\binom{\theta_\mathsf{yes}+\theta_\mathsf{no}}{\theta_\mathsf{no}} = (\delta k)^{O(k/\delta)}$, which is subexponential in $k$ when setting $\delta = k^\varepsilon$ for a sufficiently small $\varepsilon > 0$.
This is the intuition about where the subexponential bound comes from.
Of course, the analysis does not actually work, because we cheated when bounding the number of ``no'' decisions.
In fact, the branching tree of our current algorithm does not have a subexponential size, and we have to further elaborate it.

The above intuition has been used (implicitly) in several subexponential branching algorithms for general hitting set~\cite{impagliazzo2001problems,santhanam2012limits}, which aim to sparsify the input set system.
A crucial reason why these algorithms have subexponential-size branching trees is that, during the branching procedure, they keep cleaning out ``redundant'' sets from the set system and only consider large sunflowers formed by the sets that survive.
In our setting, an $F$-copy $(H,\pi)$ is \textit{$U$-redundant} for $U \subseteq V(G)$ if there exists another $F$-copy $(H',\pi')$ such that $V(H')\nsubseteq U$ and $V(H') \backslash U \subsetneq V(H) \backslash U$.
Note that if $(H',\pi')$ is hit by a set $S \subseteq V(G) \backslash U$, then $(H,\pi)$ must also be hit by $S$.
Therefore, if $U$ is the ``undeletable'' set maintained in the branching procedure, intuitively, $U$-redundant $F$-copies are useless and can be ignored (because they will be anyway hit as long as the non-$U$-redundant ones are hit).
As aforementioned, the branching algorithms of~\cite{impagliazzo2001problems,santhanam2012limits} keep cleaning out the redundant sets and only branch on the cores of large sunflowers formed by the sets that are non-redundant.
By doing this, they can guarantee that the number of ``no'' decisions along any successful path is sublinear (which in turn implies the subexponential bound on the size of the branching tree).
Unfortunately, this is not a good idea for our problem.
Of course, in each step of our branching algorithm, we can choose to branch on only the heavy cores witnessed by a structured sunflower formed by non-$U$-redundant $F$-copies.
This can still give us the collections $\mathcal{X}_1,\dots,\mathcal{X}_t$ satisfying the first condition in Theorem~\ref{thm-subtwbranch}.
The main issue is the treewidth bound: Lemma~\ref{Olem-treewidth} heavily relies on the fact that $(X_1,f_1),\dots,(X_r,f_r)$ are \textit{$U$-minimal} heavy cores. 
If we restrict ourselves to heavy cores witnessed by non-$U$-redundant $F$-copies, we cannot guarantee the heavy cores in $\mathcal{X}$ to be $U$-minimal during the branching procedure.
Indeed, there can exist heavy cores $(X,f)$ and $(Y,g)$ with $(X,f) \prec (Y,g)$ such that $(Y,g)$ can be witnessed by non-$U$-redundant $F$-copies but $(X,f)$ cannot.
Thus, at some stage of the branching, it might happen that every $U$-minimal heavy core cannot be witnessed by non-$U$-redundant $F$-copies while there are still (non-$U$-minimal) heavy cores witnessed by non-$U$-redundant $F$-copies; in this situation, we are forced to consider heavy cores that are not $U$-minimal (and possibly add them to $\mathcal{X}$).
This entirely ruins the sublinear treewidth of the Gaifman graphs.

However, the insight that ``$U$-redundant $F$-copies can be ignored'' still turns out to be useful.
But we have to be more careful about which heavy cores should be considered and which can be ignored.
As aforementioned, only branching on the heavy cores witnessed by non-$U$-redundant $F$-copies does not work.
Thus, we try to loosen the condition of ``being witnessed by non-$U$-redundant $F$-copies'' as follows.
We say a heavy core $(X,f)$ is \textit{$U$-redundant} if \textit{every} $F$-copy $(H,\pi)$ with $(X,f) \preceq (V(H),\pi)$ is $U$-redundant.
Note that for a heavy core, ``being not $U$-redundant'' is strictly weaker than ``being witnessed by non-$U$-redundant $F$-copies''.
Furthermore, the former condition has a nice \textit{hereditary} property (which the latter condition does not have): for heavy cores $(X,f)$ and $(Y,g)$ such that $(X,f) \prec (Y,g)$, if $(Y,g)$ is not $U$-redundant, then $(X,f)$ is also not $U$-redundant.
This property is important and allows us to avoid the issue we had before (when trying to branch on heavy cores witnessed by non-$U$-redundant $F$-copies).
Now we say a heavy core $(X,f)$ is \textit{$U$-active} if it is $U$-minimal and \textit{not} $U$-redundant.
We modify our algorithm so that it only branches on $U$-active heavy cores.
We show that the collections $\mathcal{X}_1,\dots,\mathcal{X}_t$ generated still satisfy both conditions in Theorem~\ref{thm-subtwbranch} after this modification.
Clearly, the treewidth bound for the Gaifman graphs of $\mathcal{X}_1,\dots,\mathcal{X}_t$ still holds, because $U$-active heavy cores are $U$-minimal.
Also, we still have that if $S \subseteq V(G)$ is an $\mathcal{F}$-hitting set of $G$, then it hits some $\mathcal{X}_i$.
The only part slightly different is to show that if $S$ hits some $\mathcal{X}_i$, then it is an $\mathcal{F}$-hitting set of $G$.
Consider the point we generate $\mathcal{X}_i$.
At this point, $X \in \mathcal{X}$ for all $U$-active heavy cores $(X,f)$, $G[U]$ is $F$-free, and we set $\mathcal{X}_i = \{X \backslash U:X \in \mathcal{X}\}$.
Without loss of generality, we only need to consider the case $S \subseteq V(G) \backslash U$, as the sets in $\mathcal{X}_i$ are disjoint from $U$.
Assume $S$ hits $\mathcal{X}_i$ but $S$ is not an $\mathcal{F}$-hitting set of $G$, for the sake of contradiction.
Then there must exist a non-$U$-redundant $F$-copy $(H,\pi)$ with $S \cap V(H) = \emptyset$.
Indeed, as $S \subseteq V(G) \backslash U$, if $S$ hits all non-$U$-redundant $F$-copies in $G$, then it is an $\mathcal{F}$-hitting set of $G$.
Now $(V(H),\pi)$ is a heavy core, which is also not $U$-redundant.
As $S$ hits $\mathcal{X}_i$ (thus hits $\mathcal{X}$) and $\mathcal{X}$ contains all $U$-active heavy cores, we know that $(V(H),\pi)$ is not $U$-active, which further implies it is not $U$-minimal.
Hence, there exists a $U$-minimal heavy core $(X,f)$ such that $(X,f) \prec (V(H),\pi)$.
By the hereditary property of $U$-redundancy, $(X,f)$ is also not $U$-redundant and is thus $U$-active.
It follows that $X \in \mathcal{X}$ and $X \backslash U \in \mathcal{X}_i$.
However, as $S \cap V(H) = \emptyset$, we have $S \cap X = \emptyset$, which contradicts our assumption that $S$ hits $\mathcal{X}_i$.
Thus, if $S$ hits $\mathcal{X}_i$, $S$ must be an $\mathcal{F}$-hitting set of $G$.

Now we see that only branching on $U$-active heavy cores can still give us the collections satisfying the conditions in Theorem~\ref{thm-subtwbranch}.
This yields a smaller branching tree, as $U$-active heavy cores form a subset of $U$-minimal heavy cores.
However, this has not yet given us a sublinear bound on the number of ``no'' decisions.
We need the last elaboration on our branching algorithm.
Still, we keep branching on $U$-active heavy cores $(X,f)$.
When making a ``yes'' decision, we recursively call $\textsc{Branch}(U,\mathcal{X} \cup \{X\})$ as before.
The changes are made to the ``no'' decisions.
When making a ``no'' decision, instead of simply adding the vertices in $X$ to $U$, we further guess an additional set $P$ of vertices that are not in the (unknown) solution $S$ and add the vertices in $P$ to $U$ as well.
As such, at each stage, we have one ``yes'' decision and \textit{multiple} ``no'' decisions corresponding to different choices of $P$.
The choices of $P$ are as follows.
Take $F$-copies $(H_1,\pi_1),\dots,(H_\Delta,\pi_\Delta)$ that witness $(X,f)$, where $\Delta = \gamma^{|X|}\delta$.
Suppose $C_1,\dots,C_t$ are the connected components of $F - f(X)$.
Define $V_{i,j} = \pi_j^{-1}(V(C_i))$ for $(i,j) \in [t] \times [\Delta]$, which is the copy of $C_i$ in $H_j$.
For each $i \in [t]$, we pick at most $\gamma$ sets in $\{V_{i,1},\dots,V_{i,\Delta}\}$ and include them in $P$.
Formally, let 
\begin{equation*}
    \mathcal{P} = \left\{\bigcup_{i=1}^t \bigcup_{j \in J_i} V_{i,j}: J_1,\dots,J_t \in \mathcal{J}\right\},
\end{equation*}
where $\mathcal{J} = \{J \subseteq [\Delta]: |J| \leq \gamma\}$.
Then the choices of $P$ are just the sets in $\mathcal{P}$.
Specifically, for each $P \in \mathcal{P}$, we recursively call $\textsc{Branch}(U \cup X \cup P,\mathcal{X})$, which corresponds to a ``no'' decision.
Note that $|\mathcal{P}| = \Delta^{O(\gamma)} = \delta^{O(1)}$.
Thus, the degree of the branching tree becomes $\delta^{O(1)}$, but this does not influence the entire size of the branching tree too much, as $\delta = k^\varepsilon$.
Surprisingly, with such a twist, we can in fact make the number of ``no'' decisions sublinear in $k$.


Finally, we briefly discuss our analysis for the number of ``no'' decisions.
Unfortunately, as our branching algorithm is already rather different from the ones in~\cite{impagliazzo2001problems,santhanam2012limits}, their arguments is not applicable here.
Instead, we use a very different analysis, which takes advantage of the graph structure, or more specifically, the bounded weak coloring number of $G$, as well as the assumption that the graph $F$ is connected.
Consider an $\mathcal{F}$-hitting set $S$ of $G$.
When branching on a ($U$-active) heavy core $(X,f)$, we define \textit{$S$-correct} decisions as follows.
If $S \cap X \neq \emptyset$, then the ``yes'' decision is the only $S$-correct decision.
If $S \cap X = \emptyset$, then a ``no'' decision is $S$-correct iff its corresponding set $P \in \mathcal{P}$ satisfies \textbf{(i)} $S \cap P = \emptyset$ and \textbf{(ii)} $S \cap P' \neq \emptyset$ for any $P' \in \mathcal{P}$ with $P \subsetneq P'$ (in other words, the set $P$ we guess is a maximal set in $\mathcal{P}$ that is disjoint from $S$).
We say a path in the branching tree from the root is \textit{$S$-successful} if every decision on the path is $S$-correct.
It is clear that at any node of an $S$-successful path, the sets $U$ and $\mathcal{X}$ always satisfy that $S \cap U = \emptyset$ and $S$ is a hitting set of $\mathcal{X}$.

Let $S$ be an $\mathcal{F}$-hitting set of $G$ with $|S| \leq k$.
Our goal is to show that along any $S$-successful path in the branching tree, the number of ``no'' decisions is bounded by $O(k/\delta)$.
This is done by a subtle charging argument.
Fix an ordering $\sigma$ of $V(G)$ such that $\text{wcol}_\gamma(G,\sigma) = \text{wcol}_\gamma(G)$.
For $v \in V(G)$, define $\lambda_S(v) = |\{u \in S: v \in \text{WR}_\gamma(G,\sigma,u)\}|$.
Then we define a set
\begin{equation*}
    R = \{v \in V(G): \lambda_S(v) \geq \delta - \gamma - \text{wcol}_\gamma(G)\}.
\end{equation*}
We have $\sum_{v \in V(G)} \lambda_S(v) = \sum_{u \in S} |\text{WR}_\gamma(G,\sigma,u)| \leq \text{wcol}_\gamma(G,\sigma)\cdot |S| = \text{wcol}_\gamma(G) \cdot |S|$.
By an averaging argument, we deduce that $|R| \leq \frac{\text{wcol}_\gamma(G)}{\delta - \gamma - \text{wcol}_\gamma(G)} \cdot |S|$.
Note that $G$ is taken from a graph class of polynomial expansion and thus $\text{wcol}_\gamma(G) = O(1)$.
Therefore, if we choose $\delta$ much larger than $\gamma+\text{wcol}_\gamma(G)$, then $|R| = O(|S|/\delta) = O(k/\delta)$.
Our plan is (essentially) to charge every ``no'' decision on the $S$-successful path to a vertex in $R$, with the guarantee that each vertex in $R$ only gets charged $O(1)$ times.
If this can be done, then the number of ``no'' decisions is $O(k/\delta)$.
\begin{observation}
    Let $(X,f)$ be a heavy core in $G$, and $x \in X$ be the largest vertex under $\sigma$.
    If $S \cap X = \emptyset$, then $x \in R$.
    Furthermore, if $\delta > \textnormal{wcol}_\gamma(G)$, then $x \in \textnormal{WR}_\gamma(G,\sigma,u)$ for all $u \in X$.
\end{observation}
\begin{proof}[Proof sketch.]
There exist $F$-copies $(H_1,\pi_1),\dots,(H_\Delta,\pi_\Delta)$ such that $V(H_1),\dots,V(H_\Delta)$ form a sunflower with core $X$, where $\Delta = \gamma^{|X|} \delta$.
Since $S \cap X = \emptyset$ and $S$ is an $\mathcal{F}$-hitting set, there exists $u_i \in S \cap (V(H_i) \backslash X)$ for all $i \in [\Delta]$.
Note that $u_1,\dots,u_\Delta$ are distinct, as $V(H_1) \backslash X,\dots,V(H_\Delta) \backslash X$ are disjoint.
Using the assumptions that $F$ is connected and $x$ is the largest vertex in $X$, we can deduce that $x \in \text{WR}_\gamma(G,\sigma,u_i)$ for at least $\Delta - \text{wcol}_\gamma(G)$ indices $i \in [\Delta]$.
As $u_1,\dots,u_\Delta \in S$, this implies $\lambda_S(x) \geq \Delta - \text{wcol}_\gamma(G) \geq \delta - \gamma - \text{wcol}_\gamma(G)$.
Thus, $x \in R$.
The second statement also follows easily from the facts that $F$ is connected and $x$ is the largest vertex in $X$.
\end{proof}

Consider a ``no'' decision on the $S$-successful path, and let $(X,f)$ be the heavy core on which the decision was made.
At the time we made the decision for $(X,f)$, we have $X \nsubseteq U$.
Pick an arbitrary vertex $y \in X \backslash U$.
As we made a ``no'' decision for $(X,f)$, we have $S \cap X = \emptyset$.
The above observation then implies $X \cap R \cap \text{WR}_\gamma(G,\sigma,y) \neq \emptyset$, since it contains the largest vertex in $X$.
We then charge the ``no'' decision for $(X,f)$ to the \textit{smallest} vertex in $X \cap R \cap \text{WR}_\gamma(G,\sigma,y)$.

It is non-obvious that why each vertex in $R$ only gets charged $O(1)$ times.
The intuition is roughly as follows.
Assume there are too many ``no'' decisions charged to the same vertex in $R$.
Let $(X_1,f_1),\dots,(X_r,f_r)$ be the heavy cores these ``no'' decisions are made on, where $r$ is very large.
By the sunflower lemma, we may assume that $X_1,\dots,X_r$ form a sunflower, without loss of generality.
Suppose the decisions for $(X_1,f_1),\dots,(X_r,f_r)$ are made in order.
After we made the ``no'' decisions for $(X_1,f_1),\dots,(X_{r-1},f_{r-1})$, the vertices in $X_1,\dots,X_{r-1}$ are all added to $U$.
In addition, each ``no'' decision here also adds a set $P$ of vertices to $U$, where $P$ is maximal among all choices that are disjoint from $S$ (because the ``no'' decision is $S$-correct).
We somehow show that using these vertices added to $U$, for every $F$-copy $(H,\pi)$ with $(X_r,f_r) \preceq (V(H),\pi)$, we can construct another $F$-copy $(H',\pi')$ satisfying $V(H') \backslash U \subsetneq V(H) \backslash U$; the construction of $H'$ is done by carefully replacing a part of $H$ that is not totally contained in $U$ with an isomorphic one that consists of those vertices added to $U$ (of course this construction relies on our charging rule as well).
It then follows that when we branch on $(X_r,f_r)$, every $F$-copy $(H,\pi)$ with $(X_r,f_r) \preceq (V(H),\pi)$ is $U$-redundant and hence $(X_r,f_r)$ is $U$-redundant.
But this contradicts the fact that we only branch on $U$-active heavy cores, and thus each vertex in $R$ cannot get charged too many times.
As a result, the number of ``no'' decisions on any $S$-successful path is bounded by $O(k/\delta)$.

According to the above discussion, we can set $\theta_\mathsf{no} = O(k/\delta)$ as the budget for ``no'' decisions, which makes the branching tree have subexponential size.
This completes the overview of how we achieve the subexponential bounds in Theorem~\ref{thm-subtwbranch}.



\section{Preliminaries} \label{sec-pre}


\paragraph{Basic notations.}
We use $\mathbb{N}$ to denote the set $\{1,2,3,\ldots\}$, and write $[n]=\{1,2,\ldots,n\}$ for short.
For a graph $G$, the notations $V(G)$ and $E(G)$ denote the set of vertices and the set of edges in $G$, respectively.
For a vertex subset $S\subseteq V(G)$, $N_G(S)=\{u\in V(G)\setminus S~\colon~ (u,x)\in E(G) \mbox{ for some }x\in S\}$ and 
$N_G[S]=N_G(S)\cup S$. 
A \textit{subgraph} of $G$ is a graph $G' = (V',E')$, denoted by $G'\subseteq G$, where $V' \subseteq V(G)$ and $E' \subseteq E(G)$.
The graph $G'$ is an {\em induced subgraph} of $G$ if $E'=\{(u,v)\in E(G)~\colon~u\in V' \mbox{ and } v\in V'\}$. 
For a set $V \subseteq V(G)$, the notation $G[V]$ denotes the subgraph of $G$ induced by $V$.
The notation $\omega(G)$ denotes the size (i.e., number of vertices) of a maximum clique in $G$.
Two graphs $G$ and $G'$ are \textit{isomorphic}, denoted by $G \simeq G'$, if there exists a bijective map $\pi:V(G) \rightarrow V(G')$ such that $(u,v) \in E(G)$ iff $(f(u),f(v)) \in E(G')$.
The function $\pi$ is called an \textit{isomorphism} between $G$ and $G'$.
For graphs $F$ and $G$, an \textit{$F$-copy} in $G$ is a pair $(H,\pi)$ where $H \subseteq G$ and $\pi: V(H) \rightarrow V(F)$ is an isomorphism between $H$ and $F$.
We say $G$ is \textit{$F$-free} if there is no $F$-copy in $G$.
For a set $\mathcal{F}$ of graphs, $G$ is \textit{$\mathcal{F}$-free} if it is $F$-free for all $F \in \mathcal{F}$.
Furthermore, we define $\mathcal{V}_\mathcal{F}(G) = \{V(H): H \text{ is a subgraph of } G \text{ isomorphic to some } F \in \mathcal{F}\}$. For a graph $G$, and subgraphs $G_1$ and $G_2$ of $G$, $G_1\cup G_2$ is the graph with vertex set $V(G_1)\cup V(G_2)$ and edge set $E(G_1)\cup E(G_2)$. 

\begin{lemma} \label{lem-isom}
    Let $G$ and $G'$ be two graphs, and $\pi:V(G) \rightarrow V(G')$ be a bijective map.
    Also, let $V_1,\dots,V_r \subseteq V(G)$ and $V_i' = \pi(V_i)$ for $i \in [r]$.
    If $G = \bigcup_{i=1}^r G[V_i]$, $G' = \bigcup_{i=1}^r G'[V_i']$, and $\pi_{|V_i}$ is an isomorphism between $G[V_i]$ and $G'[V_i']$ for all $i \in [r]$, then $\pi$ is an isomorphism between $G$ and $G'$.
\end{lemma}
\begin{proof}
Let $u,v$ be two distinct vertices in $V(G)$. Let $u'=\pi(u)$ and $v'=\pi(v)$. Notice that $(u,v)\in E(G)$ if and only if there is an integer $i\in [r]$ such that $(u,v)\in E(G_i)$. Thus, since $\pi_{|V_j}$ is an isomorphism between $G[V_j]$ and $G'[V_j']$ for all $j \in [r]$, $(u,v)\in E(G)$ if and only if $(u',v')\in E(G')$. This completes the proof of the lemma.      
\end{proof}

\paragraph{Hitting set and Gaifman graphs.}
Let $\mathcal{A}$ be a collection of sets whose elements belong to a common universe $U$.
A set $S \subseteq U$ \textit{hits} $\mathcal{A}$ (or is a \textit{hitting set} of $\mathcal{A}$) if $S \cap A \neq \emptyset$ for all $A \in \mathcal{A}$.
The \textit{Gaifman graph} of $\mathcal{A}$ is a graph with vertex set $\bigcup_{A \in \mathcal{A}} A$ where two vertices $u$ and $v$ are connected by an edge if $u,v \in A$ for some $A \in \mathcal{A}$.
For a graph $G$ and a set $\mathcal{F}$ of graphs, an \textit{$\mathcal{F}$-hitting set} of $G$ is a hitting set of $\{V(H): H \subseteq G \text{ and } H \simeq F \text{ for some } F \in \mathcal{F}\}$.
Equivalently, an $\mathcal{F}$-hitting set of $G$ is a set $S \subseteq V(G)$ such that $G-S$ is $\mathcal{F}$-free.

\paragraph{Sunflowers and the sunflower lemma.}
Let $V_1,\dots,V_r$ be sets whose elements belong to a common universe $U$.
We say $V_1,\dots,V_r$ form a \textit{sunflower} if $V_1 \backslash X,\dots,V_r \backslash X$ are disjoint, where $X = \bigcap_{i=1}^r V_i$.
We call $X$ the \textit{core} of the sunflower.
\begin{lemma}[sunflower lemma] \label{lem-sunflower}
Let $\mathcal{A}$ be a collection of sets whose sizes are at most $p$.
If $|\mathcal{A}| > (r-1)^p \cdot p!$, then there exist $r$ sets in $\mathcal{A}$ which form a sunflower and such a sunflower can be computed in time polynomial in $|\mathcal{A}|$ and $r$. 
\end{lemma}

\paragraph{Graph separators and treewidth.}
We say a graph $G$ admits \textit{balanced $(\eta,\mu,\rho)$-separators} if for every induced subgraph $H$ of $G$, there exists $S \subseteq V(H)$ of size at most $\eta \cdot \omega^\mu(H) \cdot |V(H)|^\rho$ such that every connected component of $H - S$ contains at most $\frac{1}{2} |V(H)|$ vertices, where $\omega(H)$ denotes the size of a maximum clique in $H$.
We denote by $\mathcal{G}(\eta,\mu,\rho)$ the class of all graphs admitting balanced $(\eta,\mu,\rho)$-separators.
We only focus on the case $\rho < 1$, since all graphs are contained in $\mathcal{G}(1,0,1)$.

A \textit{tree decomposition} of a graph $G$ is a pair $(T,\beta)$ where $T$ is a tree and $\beta: T \subseteq 2^{V(G)}$ maps each node $t \in T$ to a set $\beta(t) \subseteq V(G)$ called the \textit{bag} of $t$ such that
\textbf{(i)} $\bigcup_{t \in T} \beta(t) = V(G)$, \textbf{(ii)} for any edge $(u,v) \in E(G)$, there exists $t \in V(T)$ with $u,v \in \beta(t)$, and \textbf{(iii)} for any $v \in V(G)$, the nodes $t \in V(T)$ with $v \in \beta(t)$ induce a subtree in $T$.
The \textit{width} of $(T,\beta)$ is $\max_{t \in T} |\beta(t)| - 1$.
The \textit{treewidth} of a graph $G$, denoted by $\mathbf{tw}(G)$, is the minimum width of a tree decomposition of $G$.

\begin{lemma}[\cite{DorakN2016}] \label{lem-separator=tw}
    Let $\mathcal{G} \subseteq \mathcal{G}(\eta,\mu,\rho)$ where $\eta, \mu \geq 0$ and $\rho < 1$.
    Then for any graph $G \in \mathcal{G}$ of $n$ vertices, we have $\mathbf{tw}(G) \leq \frac{2^\rho \eta \cdot \omega^{\mu}(G)}{2^\rho - 1} \cdot n^\rho$.
\end{lemma}
\begin{proof}
We prove the lemma by induction on $|V(G)|$.
Suppose the statement holds when all $G \in \mathcal{G}$ with $|V(G)| \leq n-1$.
Consider a graph $G \in \mathcal{G}$ with $|V(G)| = n$.
By assumption, there is a balanced separator $S \subseteq V(G)$ of $G$ with $|S| \leq \eta \cdot \omega^{\mu}(G)\cdot n^\rho$.
Let $C_1,\dots,C_r$ be the connected components of $G-S$.
Then $|V(C_i)| \leq \frac{n}{2}$ for all $i \in [r]$.
By our induction hypothesis, 
\begin{equation*}
    \mathbf{tw}(C_i) \leq \frac{2^\rho \eta \cdot \omega^{\mu}(C_i)}{2^\rho - 1} \cdot (\frac{n}{2})^\rho \leq \frac{\eta \cdot \omega^{\mu}(G)}{2^\rho-1} \cdot n^\rho.
\end{equation*}
For $i \in [r]$, let $(T_i,\beta_i)$ be a tree decomposition of $C_i$ such that $|\beta_i(t)| \leq \frac{\eta \cdot \omega^{\mu}(G)}{2^\rho-1} \cdot n^\rho + 1$ for all nodes $t \in V(T_i)$.
For convenience, we view each $T_i$ as a rooted tree by picking an arbitrary node $t_i \in V(T_i)$ as its root.
Let $T$ be a tree consisting of a root node with subtrees $T_1,\dots,T_r$.
For each node $t \in V(T)$, we define $\beta(t) = S$ if $t$ is the root of $T$ and $\beta(t) = \beta_i(t) \cup S$ if $t \in V(T_i)$.
We now verify that $(T,\beta)$ is a tree decomposition of $G$.
For a vertex $v \in S$, the bags of all nodes of $T$ contain $v$.
For a vertex $v \in V(C_i)$, the nodes of $T$ whose bags contain $v$ are exactly those in $T_i$, which are connected in $T$.
Consider an edge $(u,v) \in E(G)$.
If $u,v \in S$, then $u,v \in \beta(t)$ for all $t \in V(T)$.
Assume $u \notin S$ without loss of generality.
Then $u \in V(C_i)$ for some $i \in [r]$.
If $v \in V(C_i)$, then $u,v \in \beta_i(t)$ for some $t \in V(T_i)$ and hence $u,v \in \beta(t)$.
On the other hand, if $v \notin V(C_i)$, then $v \in S$.
Let $t \in V(T_i)$ be a node such that $u \in \beta_i(t)$.
Then $u,v \in \beta(t)$.
Therefore, $(T,\beta)$ is a tree decomposition of $G$.
Furthermore, 
\begin{equation*}
    |\beta(t)| \leq \frac{\eta \cdot \omega^{\mu}(G)}{2^\rho-1} \cdot n^\rho +1 + |S| \leq \frac{2^\rho \eta \cdot \omega^{\mu}(G)}{2^\rho - 1} \cdot n^\rho +1
\end{equation*}
for all $t \in V(T)$.
So we have $\mathbf{tw}(G) \leq \frac{2^\rho \eta \cdot \omega^{\mu}(G)}{2^\rho - 1} \cdot n^\rho$.
\end{proof}

\paragraph{Degeneracy and smallest-last ordering.}
Let $G$ be a graph where $|V(G)| = n$.
A \textit{smallest-last ordering} of $G$ is an ordering $(v_1,\dots,v_n)$ of the vertices of $G$ such that for every $i \in [n]$, $v_i$ is a vertex in $G[\{v_1,\dots,v_i\}]$ with the minimum degree.
We say $G$ is \textit{$d$-degenerate} if there exists an ordering $(v_1,\dots,v_n)$ of vertices of $G$ such that $|N_G(v_i) \cap \{v_1,\dots,v_i\}| \leq d$ for all $i \in [n]$.

\begin{fact}[\cite{matula1983smallest}] \label{fact-degen}
The following statements are equivalent.
\begin{itemize}
    \item $G$ is $d$-degenerate.
    \item There is a smallest-last ordering $(v_1,\dots,v_n)$ of $G$ satisfying $|N_G(v_i) \cap \{v_1,\dots,v_i\}| \leq d$ for all $i \in [n]$.
    \item Every smallest-last ordering $(v_1,\dots,v_n)$ of $G$ satisfies $|N_G(v_i) \cap \{v_1,\dots,v_i\}| \leq d$ for all $i \in [n]$.
\end{itemize}
\end{fact}

\begin{theorem}[\cite{matula1983smallest}] \label{thm-smallast}
    Given a graph $G$ of $n$ vertices and $m$ edges, a smallest-last ordering of $G$ can be computed in $O(n+m)$ time.
\end{theorem}

\paragraph{Weak coloring numbers.}
Let $G$ be a graph and $\sigma$ be an ordering of $V(G)$.
For $u,v \in V(G)$, we write $u <_\sigma v$ if $u$ is before $v$ under the ordering $\sigma$.
The notations $>_\sigma$, $\leq_\sigma$, $\geq_\sigma$ are defined similarly.
For an integer $r \geq 0$, $u$ is \textit{weakly $r$-reachable} from $v$ \textit{under $\sigma$} if there is a path $\pi$ between $v$ and $u$ of length at most $r$ such that $u$ is the largest vertex on $\pi$ under the ordering $\sigma$, i.e., $u \geq_\sigma w$ for all $w \in V(\pi)$.
Let $\text{WR}_r(G,\sigma,v)$ denote the set of vertices in $G$ that are weakly $r$-reachable from $v$ under $\sigma$.
The \textit{weak $r$-coloring number} of $G$ \textit{under $\sigma$} is defined as $\mathsf{wcol}_r(G,\sigma) = \max_{v \in V(G)} |\text{WR}_r(G,\sigma,v)|$.
Finally, the \textit{weak $r$-coloring number} of $G$ is defined as $\mathsf{wcol}_r(G) = \min_{\sigma \in \Sigma(G)} \mathsf{wcol}_r(G,\sigma)$.

\begin{lemma}[\cite{DorakN2016}] \label{lem-edgenumber}
    Let $\mathcal{G} \subseteq \mathcal{G}(\eta,0,\rho)$ where $\eta \geq 0$ and $\rho < 1$.
    For any graph $G \in \mathcal{G}$ of $n$ vertices, $|E(G)| = \eta^{O(1)} \cdot n$ and $G$ is $\eta^{O(1)}$-degenerate.
    The constants hidden in $O(\cdot)$ only depend on $\rho$.
\end{lemma}

\begin{lemma}[\cite{nevsetvril2008grad}] \label{lem-nablawcol}
    Let $\mathcal{G} \subseteq \mathcal{G}(\eta,0,\rho)$ where $\eta \geq 0$ and $\rho < 1$.
    Then for any $G \in \mathcal{G}$ and $r \geq 0$, $\textnormal{wcol}_r(G) \leq (\eta r)^{O(r^2)}$.
    The constant hidden in $O(\cdot)$ only depends on $\rho$.
\end{lemma}

\paragraph{Augmentations.}
Let $\mathcal{G}$ be a graph class, and $\gamma \geq 0$ be an integer.
The \textit{$\gamma$-augmentation} of $\mathcal{G}$ is another graph class $\mathcal{G}'$ defined as follows.
For every graph $G \in \mathcal{G}'$ and every ordering $\sigma$ of $V(G)$ such that $\text{wcol}_\gamma(G,\sigma) = \text{wcol}_\gamma(G)$, we include in $\mathcal{G}'$ a graph $G'$, which is obtained from $G$ by adding edges $(u,v)$ with $v \in \text{WR}_\gamma(G,\sigma,u)$.
The augmentation has the following property.

\begin{lemma} \label{lem-augmentation}
    Let $\mathcal{G} \subseteq \mathcal{G}(\eta,0,\rho)$ where $\eta \geq 0$ and $\rho < 1$.
    Also, let $\mathcal{G}'$ be the $\gamma$-augmentation of $\mathcal{G}$.
    Then $\mathcal{G}' \subseteq \mathcal{G}(\eta',0,\rho)$ for $\eta' = (\eta \gamma)^{O(\gamma^2)}$.
    The constant hidden in $O(\cdot)$ only depends on $\rho$.
\end{lemma}

\begin{proof}
Let $G' \in \mathcal{G}'$.
By the construction of $\mathcal{G}'$, there exists $G \in \mathcal{G}$ and an ordering $\sigma$ of $V(G)$ satisfying $\text{wcol}_\gamma(G) = \text{wcol}_\gamma(G,\sigma)$ such that $G'$ is obtained from $G$ by adding edges $(u,v)$ with $v \in \text{WR}_\gamma(G,\sigma,u)$.
We have $V(G) = V(G')$.
Then set $\eta' = \eta \cdot \text{wcol}_\gamma(G)$.
Thus, $\eta' = (\eta \gamma)^{O(\gamma^2)}$ by Lemma~\ref{lem-nablawcol}.
For any subset $V \subseteq V(G')$, we want to show that $G'[V]$ has a balanced separator of size $\eta' \cdot n^{\rho}$, where $n = |V|$.
By assumption $G[V]$ has a balanced separator $S$ of size at most $\eta \cdot n^{\rho}$.
Define $S' = \bigcup_{s \in S} (V \cap \text{WR}_\gamma(G,\sigma,s))$.
We claim that $S'$ is a balanced separator of $G'[V]$.
It suffices to show that every connected component of $G'[V] - S'$ is contained in a connected component of $G[V] - S$.
Consider a connected component $C'$ of $G'[V] - S'$.
As $S \subseteq S'$, $V(C') \subseteq V \backslash S' \subseteq V \backslash S$.
Thus, there exists a connected component $C$ of $G[V]-S$ such that $V(C) \cap V(C') \neq \emptyset$.
Assume $V(C') \nsubseteq V(C)$, in order to deduce a contradiction.
Let $u \in V(C) \cap V(C')$ and $v \in V(C') \backslash V(C)$.
Since $u,v \in V(C')$ and $C'$ is connected, there is a path $(v_0,v_1,\dots,v_r)$ in $C'$ with $v_0 = u$ and $v_r = v$.
Define $i \in \{0\} \cup [r]$ as the smallest index such that $v_i \notin V(C)$.
Note that such an index exists because $v_r \notin V(C)$.
Furthermore, $i \geq 1$, as $v_0 \in V(C)$.
So we have $v_{i-1} \in V(C)$.
As $(v_{i-1},v_i) \in E(G')$, we have either $v_{i-1} \in \text{WR}_\gamma(G,\sigma,v_i)$ or $v_i \in \text{WR}_\gamma(G,\sigma,v_{i-1})$.
Thus, there exists a path $\pi$ in $G$ of length at most $\gamma$ between $v_{i-1}$ and $v_i$ on which the largest vertex (under the ordering $\sigma$) is either $v_{i-1}$ or $v_i$.
But $v_{i-1} \in V(C)$ and $v_i \notin V(C)$.
So $\pi$ must contain a vertex $s \in S$.
If $v_{i-1}$ is the largest vertex on $\pi$, then $v_{i-1} \in \text{WR}_\gamma(G,\sigma,s)$, since the sub-path $\pi'$ of $\pi$ between $v_{i-1}$ and $s$ has length at most $\gamma$ and the largest vertex on $\pi'$ is $v_{i-1}$.
Similarly, if $v_i$ is the largest vertex on $\pi$, then $v_i \in \text{WR}_\gamma(G,\sigma,s)$.
In other words, we have either $v_{i-1} \in \text{WR}_\gamma(G,\sigma,s)$ or $v_i \in \text{WR}_\gamma(G,\sigma,s)$, which implies that either $v_{i-1} \in S'$ or $v_i \in S'$.
This contradicts the fact that $(v_0,v_1,\dots,v_r)$ is a path in $C'$.
Therefore, $S'$ is a balanced separator of $G'[V]$.
Finally, since $|\text{WR}_\gamma(G,\sigma,s)| \leq \text{wcol}_\gamma(G,\sigma) = \text{wcol}_\gamma(G)$ for all $s \in S$, we have $|S'| \leq \text{wcol}_\gamma(G) \cdot |S| \leq \eta' \cdot n^\rho$.
\end{proof}

\paragraph{First-order model checking.}
A \textit{first-order formula} on a graph $G$ has variables representing the vertices of $G$, and it is built by combining two types of {\em atomic formulas}: (i) equality for variables $u,v$, denoted by $u = v$, which is true iff $u$ and $v$ represent the same vertex of $G$, and (ii) adjacency for variables $u,v$, denoted $\mathsf{Adj}(u,v)$, which is true if $u$ and $v$ represent two vertices that are adjacent in $G$.
Atomic formulas are combined using boolean connectives $\wedge, \vee, \neg$ and quantifiers $\exists,\forall$.
We shall use the following well-known result for first-order model-checking on sparse graphs.

\begin{theorem}[\cite{dvovrak2013testing}] \label{thm-fologic}
Let $\mathcal{G} \subseteq \mathcal{G}(\eta,0,\rho)$ where $\eta \geq 0$ and $\rho < 1$.
Also, let $\ell > 0$ be a fixed number. 
Given a graph $G \in \mathcal{G}$ with $n$ vertices and a first-order formula $\varphi$ on $G$ of length at most $\ell$, one can find in $\eta^{O(1)} \cdot n$ time a satisfying assignment of $\varphi$ (or conclude that $\varphi$ does not admit a satisfying assignment).
The constant hidden in $O(\cdot)$ only depends on $\rho$ and $\ell$.
\end{theorem}

\section{Linear-time polynomial kernel} \label{sec-kernel}

In this section, we present our linear-time polynomial kernel for \textsc{$\mathcal{F}$-Hitting} on a graph class $\mathcal{G} \subseteq \mathcal{G}(\eta,\mu,\rho)$ for $\rho < 1$.
Specifically, we prove the following theorem.

\kernel*

\begin{algorithm}[b]
    \caption{\textsc{Kernel}$(G,k)$}
    \begin{algorithmic}[1]
        \State $(v_1,\dots,v_n) \leftarrow \textsc{SmallestLastOrder}(G)$
        \State $I \leftarrow \{i \in [n]: |N_G(v_i) \cap \{v_1,\dots,v_i\}| \geq (\eta (k+\gamma)^\mu)^c+1 \} \cup \{n\}$
        \State $G_0 \leftarrow G[\{v_1,\dots,v_{\min(I)}\}]$
        \State $K \leftarrow \emptyset$ and $\mathcal{X} \leftarrow \emptyset$
        \While{there exists $V \in \mathcal{V}_\mathcal{F}(G_0)$ such that $V \nsubseteq K$ and $X \nsubseteq V$ for all $X \in \mathcal{X}$}
            \State $K \leftarrow K \cup V$
            \For{every $X \subseteq K$ with $|X| < \gamma$}
                \State $U_1,\dots,U_p \leftarrow$ a maximal sunflower in $\mathcal{V}_\mathcal{F}(G_0[K])$ with core $X$
                \If{$p > k$}{ $\mathcal{X} \leftarrow \mathcal{X} \cup \{X\}$}
                \EndIf
            \EndFor
        \EndWhile
        \State $G' \leftarrow G_0[K]$
        \State \textbf{return} $G'$
    \end{algorithmic}
    \label{alg-kernel}
\end{algorithm}

Consider a graph class $\mathcal{G} \subseteq \mathcal{G}(\eta,\mu,\rho)$ where $\rho < 1$.
By Lemma~\ref{lem-edgenumber}, there exists a constant $c > 0$ depending on $\rho$ such that for any $\tilde{\eta} \geq 0$, every graph in $\mathcal{G}(\tilde{\eta},0,\rho)$ is $\tilde{\eta}^c$-degenerate.
Let $G \in \mathcal{G}$ be the input graph, and set $\gamma = \max_{F \in \mathcal{F}} |V(F)|$.
The overall framework of our algorithm is presented in Algorithm~\ref{alg-kernel}.
In line~1, the sub-routing $\textsc{SmallestLastOrder}(G)$ computes a smallest-last ordering $(v_1,\dots,v_n)$ of $G$.
Then line~2-3, we compute an induced subgraph $G_0$ of $G$ as follows.
We find the smallest index $i \in [n]$ satisfying $|N_G(v_i) \cap \{v_1,\dots,v_i\}| \geq (\eta (k+\gamma)^\mu)^c+1$; if such an index does not exist, we set $i = n$.
Then we define $G_0 \leftarrow G[\{v_1,\dots,v_i\}]$.
Line~4-10 computes the desired graph $G'$, which will be an induced subgraph of $G_0$.
In this procedure, we maintain a subset $K \subseteq V(G_0)$ and a collection $\mathcal{X}$ of subsets of $K$.
Roughly speaking, $\mathcal{X}$ consists of the cores of large sunflowers in $\mathcal{V}_\mathcal{F}(G_0[K])$.
Initially, $K = \emptyset$ and $\mathcal{X} = \emptyset$ (line~4).
In each round, we try to find a set $V \in \mathcal{V}_\mathcal{F}(G_0)$ such that $V \nsubseteq K$ and $X \nsubseteq V$ for all $X \in \mathcal{X}$ (line~5).
We add the vertices in $V$ to $K$ (line~6).
Then for every subset $X \subseteq K$ of size smaller than $\gamma$, we compute a maximal sunflower in $\mathcal{V}_\mathcal{F}(G_0[K])$ with core $X$, and add $X$ to $\mathcal{X}$ if the size of this sunflower is larger than $k$ (line 7-9).
The procedure terminates when there does not exist such a set $V$.
After this, we simply define $G' = G_0[K]$ (line~10) and return $G'$ as the output of the algorithm (line~11).


\paragraph{Correctness.}
We first prove the correctness of Algorithm~\ref{alg-kernel}.
Specifically, Theorem~\ref{thm-kernel} requires \textbf{(i)} $|V(G')| = k^{O(1)}$ and \textbf{(ii)} any $\mathcal{F}$-hitting set of $G'$ of size at most $k$ is also an $\mathcal{F}$-hitting set of $G$.
To see \textbf{(i)}, the key observation is that the sets $V$ considered in the while-loop (line~5-9) cannot form a large sunflower.
Suppose the while-loop has in total $r$ iterations and let $V_i$ denote the set $V$ in the $i$-th iteration for $i \in [r]$.
\begin{observation} \label{obs-polyksize}
The set system $\{V_1,\dots,V_r\}$ does not contain a sunflower of size $\gamma (k+1) + 1$.
In particular, we have $r \leq \gamma! \gamma^\gamma (k+1)^\gamma = k^{O(1)}$ and $|V(G')| \leq \gamma! \gamma^{\gamma+1} (k+1)^\gamma = k^{O(1)}$.
\end{observation}
\begin{proof}
Assume $\{V_1,\dots,V_r\}$ contains a sunflower of size $\gamma (k+1) + 1$, and let $V_i$ be the set in this sunflower with the largest index and $X$ be the core of the sunflower.
Consider the $(i-1)$-th iteration.
After line~6, we have $K = \bigcup_{j=1}^{i-1} V_j$.
Note that $\{V_1,\dots,V_{i-1}\}$ contains a sunflower of size $\gamma (k+1)$ with core $X$, by the choice of $i$.
In particular, $\mathcal{V}_\mathcal{F}(G_0[K])$ contains a sunflower of size $\gamma (k+1)$ with core $X$.
Since the sets in $\mathcal{V}_\mathcal{F}(G_0[K])$ are of size at most $\gamma$, a maximal sunflower in $\mathcal{V}_\mathcal{F}(G_0[K])$ with core $X$ has size at least $k+1$.
Therefore, the algorithm adds $X$ to $\mathcal{X}$ in line~9.
In other words, at the beginning of the $i$-th iteration, we have $X \in \mathcal{X}$.
However, this contradicts the fact that $X \subseteq V_i$, since $V_i$ should not contain any set in $\mathcal{X}$.
Thus, $\{V_1,\dots,V_r\}$ does not contain a sunflower of size $\gamma (k+1) + 1$.
By Lemma~\ref{lem-sunflower}, this implies $r \leq \gamma! \gamma^\gamma (k+1)^\gamma = k^{O(1)}$.
As a result, $|V(G')| = |\bigcup_{i=1}^r V_i| \leq \gamma r = \gamma! \gamma^{\gamma+1} (k+1)^\gamma = k^{O(1)}$.
\end{proof}

To see \textbf{(ii)}, we shall first establish the relation between the $\mathcal{F}$-hitting sets of $G'$ and the $\mathcal{F}$-hitting sets of $G_0$, and then consider the relation between $G_0$ and $G$.

\begin{observation} \label{obs-G'toG0}
    Any $\mathcal{F}$-hitting set $S \subseteq V(G')$ of $G'$ with $|S| \leq k$ is also an $\mathcal{F}$-hitting set of $G_0$.
\end{observation}
\begin{proof}
Recall that $G' = G_0[K]$ where $K \subseteq V(G_0)$ is the subset generated by the while-loop of Algorithm~\ref{alg-kernel}.
When the while-loop terminates, there does not exist $V \in \mathcal{V}_\mathcal{F}(G_0)$ with $V \nsubseteq K$ and $X \nsubseteq V$ for all $X \in \mathcal{X}$.
Let $S \subseteq V(G') = K$ be an $\mathcal{F}$-hitting set of $G'$ and suppose $|S| \leq k$.
Consider a set $V \in \mathcal{V}_\mathcal{F}(G_0)$.
We need to show that $S \cap V \neq \emptyset$.
If $V \subseteq K$, then $V \in \mathcal{V}_\mathcal{F}(G')$ and $S \cap V \neq \emptyset$ by construction.
So assume $V \nsubseteq K$.
As the while-loop terminates with $K$, there must exist $X \in \mathcal{X}$ with $X \subseteq V$.
Note that $X$ is the core of some sunflower $U_1,\dots,U_p \in \mathcal{V}_\mathcal{F}(G')$ with $p > k$.
As an $\mathcal{F}$-hitting set of $G'$, we have $S \cap U_i \neq \emptyset$ for all $i \in [p]$.
But $|S| \leq k$, while $p > k$.
Thus, $S \cap X \neq \emptyset$, which implies $S \cap V \neq \emptyset$.
\end{proof}

\begin{observation} \label{obs-G0toG}
    If $G_0 \neq G$, then $G_0$ does not admit an $\mathcal{F}$-hitting set of size at most $k$.
\end{observation}
\begin{proof}
Recall that $I = \{i \in [n]: |N_G(v_i) \cap \{v_1,\dots,v_i\}| \geq (\eta (k+\gamma)^\mu)^c+1 \} \cup \{n\}$ and $G_0 = G[\{v_1,\dots,v_i\}]$ with $i = \min(I)$.
If $G_0 \neq G$, then $i < n$ and $|N_G(v_i) \cap \{v_1,\dots,v_i\}| \geq (\eta(k+\gamma)^\mu)^c+1$.
Note that $(v_1,\dots,v_i)$ is a smallest-last ordering of $G_0$.
Therefore, $G_0$ is not $(\eta(k+\gamma)^\mu)^c$-degenerate by Fact~\ref{fact-degen}.
We claim that $\omega(G_0) > k+\gamma$.
Assume $\omega(G_0) \leq k+\gamma$.
Since $G_0 \in \mathcal{G}(\eta,\mu,\rho)$, we have $G_0 \in \mathcal{G}(\eta \omega^{\mu}(G_0),0,\rho)$.
By Lemma~\ref{lem-edgenumber}, $G_0$ is $(\eta(k+\gamma)^\mu)^c$-degenerate, which contradicts what we observed above.
Thus, $\omega(G_0) > k+\gamma$.
It follows that for any $S \subseteq V(G_0)$ with $|S| \leq k$, $G_0-S$ contains a clique of size $\gamma$ and hence contains any graph in $\mathcal{F}$ as a subgraph.
We then conclude that $G_0$ does not admit an $\mathcal{F}$-hitting set of size at most $k$.
\end{proof}

\noindent
The above observation implies that any $\mathcal{F}$-hitting set $S \subseteq V(G_0)$ of $G_0$ with $|S| \leq k$ is also an $\mathcal{F}$-hitting set of $G$. 
Combining this with Observation~\ref{obs-G'toG0}, the correctness of Algorithm~\ref{alg-kernel} is proved.

\paragraph{Linear-time implementation.}
Next, we show how to implement Algorithm~\ref{alg-kernel} in $k^{O(1)} \cdot n+O(m)$ time.
Computing the smallest-last ordering $(v_1,\dots,v_n)$ of $G$ takes $O(n+m)$ time by Theorem~\ref{thm-smallast}, and within the same time we can compute the graph $G_0$.
Note that $G_0$ is $k^{O(1)}$-degenerate and hence $\omega(G_0) = k^{O(1)}$.
Since $G_0 \in \mathcal{G}(\eta,\mu,\rho)$, we have $G_0 \in \mathcal{G}(\eta \omega^{\mu}(G_0),0,\rho) = \mathcal{G}(\eta',0,\rho)$, where $\eta' = k^{O(1)}$.
Observation~\ref{obs-polyksize} shows that the while-loop has $k^{O(1)}$ iterations and $|K| = k^{O(1)}$.
So it suffices to implement each iteration of the while-loop in $k^{O(1)} \cdot n$ time.
Since $|K| = k^{O(1)}$ and $\gamma = O(1)$, the for-loop (line 7-9) only has $k^{O(1)}$ iterations.
For each $X$ considered in the for-loop, computing a maximal sunflower in $\mathcal{V}_\mathcal{F}(G_0[K])$ takes $k^{O(1)}$ time because $|K| = k^{O(1)}$.

Now the only missing part is how to efficiently find the set $V \in \mathcal{V}_\mathcal{F}(G_0)$ satisfying $V \nsubseteq K$ and $X \nsubseteq V$ for all $X \in \mathcal{X}$, in each iteration.
To this end, we guess the graph $F \in \mathcal{F}$ corresponding to $V$ and the intersection $I = V \cap K$.
Since $V \nsubseteq K$, $I$ is a subset of $K$ with size at most $|V(F)|-1$.
Furthermore, we have $X \nsubseteq V$ for all $X \in \mathcal{X}$ iff $X \nsubseteq I$ for all $X \in \mathcal{X}$, because the sets in $\mathcal{X}$ are all contained in $K$.
Therefore, we say a pair $(F,I)$ where $F \in \mathcal{F}$ and $I \subseteq K$ is \textit{feasible} if $|I| \leq |V(F)|-1$ and $X \nsubseteq I$ for all $X \in \mathcal{X}$.
Since $|K| = k^{O(1)}$ and $|V(F)| \leq \gamma$ for all $F \in \mathcal{F}$, the number of feasible pairs is $k^{O(1)}$.
Also, as $|\mathcal{X}| = k^{O(1)}$, one can compute all feasible pairs in $k^{O(1)}$ time.
For every feasible pair $(F,I)$, we try to find an $F$-copy $(H,\pi)$ in $G_0$ satisfying $V(H) \cap K = I$.
Equivalently, we search for an $F$-copy $(H,\pi)$ in $G_0-(K \backslash I)$ satisfying $I \subseteq V(H)$.
This problem can be formulated as a first-order formula of constant length on $G_0-(K \backslash I)$ as follows.
Recall that $\mathsf{Adj}(u,v)$ is the atomic formula which is true if $(u,v)$ is an edge of $G_0-(K \backslash I)$ and is false otherwise.
Suppose $V(F) = \{a_1,\dots,a_r\}$ and $I = \{x_1,\dots,x_p\}$.
Set $E = \{(i,j) \in [r]^2: (a_i,a_j) \in E(F)\}$.
For vertices $b_1,\dots,b_r \in V(G_0) \backslash (K \backslash I)$, testing whether there exists an $F$-copy $(H,\pi)$ in $G_0-(K \backslash I)$ with $V(H) = \{b_1,\dots,b_r\}$ and $\pi(b_i) = a_i$ for all $i \in [r]$ can be formulated as the first-order formula
\begin{equation*}
    \phi(b_1,\dots,b_r) =  \left(\bigwedge_{(i,j) \in [r]^2} \left( (i \neq j) \rightarrow (b_i \neq b_j) \right) \right) \wedge \left(\bigwedge_{(i,j) \in E} \mathsf{Adj}(b_i,b_j) \right).
\end{equation*}
Furthermore, for vertices $b_1,\dots,b_r \in V(G_0) \backslash (K \backslash I)$, testing whether $I = \{x_1,\dots,x_p\} \subseteq \{b_1,\dots,b_r\}$ can be formulated as the first-order formula
\begin{equation*}
    \psi(b_1,\dots,b_r) = \bigvee_{\sigma \in \varSigma} \left( \bigwedge_{i \in [p]} \left(x_i = b_{\sigma(i)}\right) \right),
\end{equation*}
where $\varSigma$ is the set of all maps from $[p]$ to $[r]$.
Define $\tau(b_1,\dots,b_r) = \phi(b_1,\dots,b_r) \wedge \psi(b_1,\dots,b_r)$.
Now it suffices to find vertices $b_1,\dots,b_r \in V(G_0) \backslash (K \backslash I)$ satisfying $\tau(b_1,\dots,b_r)$, and the set $V = \{b_1,\dots,b_r\}$ is exactly what we want.
Note that $\tau$ is a first-order formula of constant length.
Since $G \in \mathcal{G}(\eta',0,\rho)$ where $\eta' = k^{O(1)}$, by Theorem~\ref{thm-fologic}, we can find a satisfying assignment of $\tau$ in $k^{O(1)} \cdot n$ time (if it exists).
By considering all $k^{O(1)}$ feasible pairs $(F,I)$, we can find the desired set $V$ or conclude that it does not exist.
As such, each iteration of the while-loop can be done in $k^{O(1)} \cdot n$ time, and the overall running time of Algorithm~\ref{alg-kernel} is $k^{O(1)} \cdot n + O(m)$.




\section{Sublinear-treewidth branching} \label{sec-branching}

In this section, we prove our branching theorem for \textsc{$\mathcal{F}$-Hitting}, which is the following.

\branching*

In Section~\ref{sec-heavy}, we introduce a key notion used in our proof, called \textit{heavy cores}, and prove important structural properties of heavy cores.
Then in Sections~\ref{sec-branchalg} and~\ref{sec-branchana}, we present and analyze our algorithm for a special case of Theorem~\ref{thm-subtwbranch}, where $\mu = 0$ (i.e., $\mathcal{G}$ is a graph class of polynomial expansion) and $\mathcal{F}$ consists of only connected graphs.
This special case is the most important part of Theorem~\ref{thm-subtwbranch}, and extending it to the general case is easy (based on a simple observation and known results in the literature), which is done in Section~\ref{sec-generalbranch}.

\subsection{Heavy cores and their properties} \label{sec-heavy}

Throughout this section, $G$ is a graph, $\mathcal{F}$ is a finite set of graphs, and $\gamma = \max_{F \in \mathcal{F}} |V(F)|$.
A \textit{standard triple} refers to a triple $(X,F,f)$, where $X \subseteq V(G)$, $F \in \mathcal{F}$, and $f:X \rightarrow V(F)$ is an \textit{injective} map.
For two standard triples $(X,F,f)$ and $(Y,F,g)$, we write $(X,F,f) \prec (Y,F,g)$ if $X \subsetneq Y$ and $f = g_{|X}$.
Also, we write $(X,F,f) \preceq (Y,F,g)$ if $(X,F,f) \prec (Y,F,g)$ or $(X,F,f) = (Y,F,g)$.
Clearly, $\prec$ defines a partial order among all standard triples.
The following fact follows directly from the sunflower lemma (Lemma~\ref{lem-sunflower}).
\begin{fact} \label{fact-goodsunflower}
    Let $p > 0$ be an integer, and $(X_1,F,f_1),\dots,(X_r,F,f_r)$ be standard triples where $r \geq (\gamma^\gamma p)^\gamma \gamma!$.
    Then there exists $I \subseteq [r]$ with $|I| = p$ such that $\{X_i: i \in I\}$ form a sunflower with core $K$ and $(f_i)_{|K} = (f_j)_{|K}$ for all $i,j \in I$.
\end{fact}
\begin{proof}
By the sunflower lemma, there exists a sunflower of size $q = \gamma^\gamma p$ among the sets in $\{X_1,\dots,X_r\}$.
Without loss of generality, assume the sunflower is $X_1,\dots,X_q$ and let $K \subseteq V(G)$ be its core.
For each $i \in [q]$, $(f_i)_{|K}$ is a map from $K$ to $V(F)$.
There are at most $\gamma^\gamma$ different maps $f: K \rightarrow V(F)$, since $|K| \leq \gamma$ and $|V(F)| \leq \gamma$.
So by the Pigeonhole principle, there exists $I \in [q]$ with $|I| = p$ such that $(f_i)_{|K} = (f_j)_{|K}$ for all $i,j \in I$.
\end{proof}


Next, we define the key notion used in our proof, called \textit{heavy cores}, which are essentially standard triples corresponding to the cores of large sunflowers formed by $F$-copies in $G$.

\begin{definition}[heavy cores]
For an integer $\delta>0$, a \textbf{$(\mathcal{F},\delta)$-heavy core} in $G$ is a standard triple $(X,F,f)$ satisfying the following condition: for $\Delta = \gamma^{|X|}\delta$, there exist $F$-copies $(H_1,\pi_1),\dots,(H_\Delta,\pi_\Delta)$ in $G$ such that $V(H_1),\dots,V(H_\Delta)$ is a sunflower with core $X$ and $(\pi_i)_{|X} = f$ for all $i \in [\Delta]$.
We call the $F$-copies $(H_1,\pi_1),\dots,(H_\Delta,\pi_\Delta)$ a \textbf{witness} of $(X,F,f)$.
\end{definition}

For convenience, when $\mathcal{F}$ and $\delta$ are clear from the context, we shall simply use the term \textit{heavy cores} instead of $(\mathcal{F},\delta)$-heavy cores.
Note that in when defining heavy cores, we \textit{do} not require the $F$-copies $(H_1,\pi_1),\dots,(H_\Delta,\pi_\Delta)$ to be distinct.
However, if two of them are identical, all of them have to be identical in order to satisfy the sunflower condition.
In this case, suppose
\begin{equation*}
    (H_1,\pi_1) = \cdots = (H_\Delta,\pi_\Delta) = (H,\pi)
\end{equation*}
and then $(X,F,f) = (V(H),F,\pi)$.
On the other hand, if $(X,F,f)$ is a heavy core in $G$ where $f$ is surjective, then any witness of $(X,F,f)$ must consist of identical $F$-copies.
\begin{fact} \label{fact-Fcopyheavy}
    Let $\delta > 0$ be an integer.
    For each $F \in \mathcal{F}$ and each $F$-copy $(H,\pi)$ in $G$, the triple $(V(H),F,\pi)$ is a $(\mathcal{F},\delta)$-heavy core in $G$.
    Conversely, if $(X,F,f)$ is a $(\mathcal{F},\delta)$-heavy core in $G$ where $f$ is surjective, then there exists a unique $F$-copy $(H,\pi)$ in $G$ such that $(X,F,f) = (V(H),F,\pi)$.
\end{fact}

\noindent
We observe the following basic property of a heavy core.

\begin{fact} \label{fact-connect}
    Let $\sigma$ be an ordering of $V(G)$, $\delta > \textnormal{wcol}_\gamma(G,\sigma)$ be an integer, and $(X,F,f)$ be a $(\mathcal{F},\delta)$-heavy core in $G$.
    For $x,y \in X$, if there exists a path $\pi$ in $F$ between $f(x)$ and $f(y)$ of length $\ell$ such that $y$ is the largest vertex in $f^{-1}(V(\pi))$ under the ordering $\sigma$, then $y \in \textnormal{WR}_\ell(G,\sigma,x)$.
\end{fact}


\begin{proof}
Let $(H_1,\phi_1),\dots,(H_\Delta,\phi_\Delta)$ be a witness of $(X,F,f)$, where $\Delta = \gamma^{|X|} \delta > \text{wcol}_\gamma(G,\sigma)$.
Suppose there exists a path $\pi$ in $F$ between $f(x)$ and $f(y)$ of length $\ell$ such that $y$ is the largest vertex in $f^{-1}(V(\pi))$ under the ordering $\sigma$.
Without loss of generality, we can assume that $\pi$ is simple and thus $\ell \leq \gamma$.
In each $H_i$, there exists a path $\pi_i$ between $x$ and $y$, where $\phi_i$ induces an isomorphism between $\pi_i$ and $\pi$.
Let $v_i = \max(V(\pi_i))$ for $i \in [\Delta]$ and $I = \{i \in [\Delta]: v_i \neq y\}$.
Observe that $v_i \notin f^{-1}(V(\pi))$ for all $i \in I$.
Indeed, if $v_i \in f^{-1}(V(\pi))$, then we must have $v_i = y$ and thus $i \notin I$, because $y$ is the largest vertex in $f^{-1}(V(\pi))$.
Since $f^{-1}(V(\pi)) = V(\pi_i) \cap X$, we have $v_i \in V(\pi_i) \backslash X$ for all $i \in I$.
Therefore, these $v_i$'s are distinct.
Furthermore, $v_i \in \textnormal{WR}_\ell(G,\sigma,x)$ for all $i \in [\Delta]$ and in particular all $i \in I$, as the largest vertex on the sub-path of $\pi_i$ between $x$ and $v_i$ is $v_i$ itself.
It follows that $|I| \leq |\textnormal{WR}_\ell(G,\sigma,x)| \leq \text{wcol}_\ell(G,\sigma)$.
Since $\Delta \geq \delta > \text{wcol}_\gamma(G,\sigma) \geq \text{wcol}_\ell(G,\sigma)$, we have $[\Delta] \backslash I \neq \emptyset$.
Let $i \in [\Delta] \backslash I$.
Then $y = v_i \in \textnormal{WR}_\ell(G,\sigma,x)$.
\end{proof}

Let $U \subseteq V(G)$ be a subset.
We say a $(\mathcal{F},\delta)$-heavy core $(X,F,f)$ in $G$ is \textit{$U$-minimal} if $X \nsubseteq U$ and for any $(\mathcal{F},\delta)$-heavy core $(Y,F,g)$ in $G$ such that $(Y,F,g) \prec (X,F,f)$, we have $Y \subseteq U$.
The most important result in this section is a structural lemma for $U$-minimal heavy cores (in graphs that admit small separators), which is Lemma~\ref{lem-twGaifman}.
The lemma essentially states that if $G$ admits $(\eta,0,\rho)$-separators for $\rho < 1$, then for any $U$-minimal heavy cores $(X_1,F_1,f_1),\dots,(X_r,F_r,f_r)$ in $G$, the Gaifman graph of $\{X_1 \backslash U,\dots,X_r \backslash U\}$ has treewidth sublinear in the size of a minimum hitting set of $\{X_1 \backslash U,\dots,X_r \backslash U\}$.
The proof of Lemma~\ref{lem-twGaifman} is highly technical, and thus we need to first establish some auxiliary results, i.e., Lemma~\ref{lem-heavysep}, Corollary~\ref{cor-heavycore}, and Corollary~\ref{cor-hitfew}.
These results work for any $G$ (i.e., do not require $G$ to have small separators).

\begin{lemma} \label{lem-heavysep}
    Let $\delta > 0$ and $p \geq \gamma^\gamma \delta$ be integers.
    Suppose $(X_1,F,f_1),\dots,(X_p,F,f_p)$ are $(\mathcal{F},\delta)$-heavy cores in $G$ satisfy the following.
    \begin{itemize}
        \item $|X_1| = \cdots = |X_p|$ and $X_1,\dots,X_p$ form a sunflower with core $K$.
        \item There exists $f: K \rightarrow V(F)$ such that $f = (f_1)_{|K} = \cdots = (f_p)_{|K}$.
    \end{itemize}
    Then for any $i \in [p]$ and any set $\mathcal{C}$ of connected components of $F - f(K)$, $(X_i^\mathcal{C},F,f_i^\mathcal{C})$ is a $(\mathcal{F},\delta)$-heavy core in $G$, where $X_i^\mathcal{C} = K \cup (\bigcup_{C \in \mathcal{C}} f_i^{-1}(V(C)))$ and $f_i^\mathcal{C} = (f_i)_{|X_i^\mathcal{C}}$.
\end{lemma}
\begin{proof}
Let $\mathcal{C}$ be a set of connected components of $F - f(K)$.
Without loss of generality, we only need to show that $(X_1^\mathcal{C},F,f_1^\mathcal{C})$ is a heavy core in $G$.
If $X_1^\mathcal{C} = X_1$, then $(X_1^\mathcal{C},F,f_1^\mathcal{C}) = (X_1,F,f_1)$ and we are done.
So assume $X_1^\mathcal{C} \subsetneq X_1$.
Thus, $|X_1^\mathcal{C}| < |X_1| \leq \gamma$ and thus $|X_1^\mathcal{C}| \leq \gamma-1$.
Define $A = f(K) \cup (\bigcup_{C \in \mathcal{C}} V(C))$ and $B = (V(F) \backslash A) \cup f(K)$.
Clearly, $A \cup B = V(F)$ and $A \cap B = f(K)$.
Also, $X_1^\mathcal{C} = f_1^{-1}(A)$.
Since $f_1^{-1}(A) = X_1^\mathcal{C} \subsetneq X_1 = f_1^{-1}(V(F))$, we have $A \subsetneq V(F)$ and thus $B \neq \emptyset$.

Write $X = X_1^\mathcal{C} = f_1^{-1}(A)$ and $f = f_1^\mathcal{C} = (f_1)_{|X}$ for convenience.
Let $\Delta \in \mathbb{N}$ be the largest integer such that there exist $F$-copies $(H_1,\pi_1),\dots,(H_\Delta,\pi_\Delta)$ satisfying that $V(H_1),\dots,V(H_\Delta)$ form a sunflower with core $X$ and $(\pi_i)_{|X} = f$ for all $i \in [\Delta]$.
We shall prove that $\Delta \geq \gamma^{|X|} \delta$, and thus by definition $(X,F,f)$ is a heavy core in $G$.
Assume $\Delta < \gamma^{|X|} \delta$ in order to deduce a contradiction.
In particular, $\Delta < \gamma^{\gamma-1} \delta = p/\gamma$, as $|X| \leq \gamma-1$.
Set $Z = \bigcup_{i=1}^\Delta (V(H_i) \backslash X)$.
We have $|Z| \leq \gamma \Delta$.
Since $(X_1,F,f_1)$ is a heavy core in $G$, it has a witness, which consists of $F$-copies $(P_1,\phi_1),\dots,(P_{\Delta'},\phi_{\Delta'})$ in $G$ where $\Delta' = \gamma^{|X_1|} \delta$.
The sets $P_1,\dots,P_{\Delta'}$ form a sunflower with core $X_1$ and $(\phi_1)_{|X_1} = \cdots = (\phi_{\Delta'})_{|X_1}$.
As $|X_1| \geq |X|+1$, we have $\Delta' \geq \gamma \cdot \gamma^{|X|} \delta > \gamma \Delta \geq |Z|$.
Because $V(P_1) \backslash X_1,\dots,V(P_{\Delta'}) \backslash X_1$ are disjoint and $\Delta' > |Z|$, there exists an index $i \in [\Delta']$ such that $(V(P_i) \backslash X_1) \cap Z = \emptyset$.
Without loss of generality, we can assume $(V(P_1) \backslash X_1) \cap Z = \emptyset$.
Set $Z' = Z \cup \phi_1^{-1}(A)$.
Since $A \subsetneq V(F)$, $|\phi_1^{-1}(A)| = |A| < \gamma$ and thus $|Z'| < \gamma \Delta + \gamma = \gamma (\Delta+1)$.
We have $\Delta +1 \leq \gamma^{|X|} \delta$, because $\Delta < \gamma^{|X|} \delta$.
Therefore, $|Z'| < \gamma^{|X|+1} \delta \leq \Delta' \leq p$.
As $X_1 \backslash K,\dots,X_p \backslash K$ are disjoint and $|Z'|<p$, there exists an index $i \in [p]$ such that $(X_i \backslash K) \cap Z' = \emptyset$.
Now $(X_i,F,f_i)$ is a heavy core in $G$.
By assumption, we have $|X_i| = |X_1|$.
So $(X_i,F,f_i)$ has a witness consisting of $F$-copies $(Q_1,\psi_1),\dots,(Q_{\Delta'},\psi_{\Delta'})$ in $G$.
Again, the sets $V(Q_1),\dots,V(Q_{\Delta'})$ form a sunflower with core $X_i$ and $(\psi_1)_{|X_i} = \cdots = (\psi_{\Delta'})_{|X_i}$.
Since $V(Q_1) \backslash X_i,\dots,V(Q_{\Delta'}) \backslash X_i$ are disjoint and $|Z'| < \Delta'$, there exists an index $j \in [\Delta']$ such that $(V(Q_j) \backslash X_i) \cap Z' = \emptyset$.
Without loss of generality, we can assume $(V(Q_1) \backslash X_i) \cap Z' = \emptyset$.
Now we have
\begin{equation*}
    (V(Q_1) \backslash K) \cap Z' = ((V(Q_1) \backslash X_i) \cap Z') \cup ((X_i \backslash K) \cap Z') = \emptyset,
\end{equation*}
and in particular $(\psi_1^{-1}(B) \backslash K) \cap \phi_1^{-1}(A) = \emptyset$, which implies $\phi_1^{-1}(A) \cap \psi_1^{-1}(B) = K$ as $K \subseteq \phi_1^{-1}(A)$ and $K \subseteq \psi_1^{-1}(B)$.
Next, we construct an $F$-copy $(H,\pi)$ as follows.
Define $H = P_1[\phi_1^{-1}(A)] \cup Q_1[\psi_1^{-1}(B)]$ and $\pi: V(H) \rightarrow V(F)$ as
\begin{equation*}
    \pi(v) = \left\{
    \begin{array}{ll}
        \phi_1(v) & \text{if } v \in \phi_1^{-1}(A), \\
        \psi_1(v) & \text{if } v \in \psi_1^{-1}(B).
    \end{array}
    \right.
\end{equation*}
Note that $\pi$ is well-defined, as $\phi_1^{-1}(A) \cap \psi_1^{-1}(B) = K$ and $(\phi_1)_{|K} = (f_1)_{|K} = f = (f_i)_{|K} = (\psi_1)_{|K}$.
We show that $\pi$ is bijective.
The sets $A$ and $B$ are both in the image of $\pi$, as $\phi_1$ and $\psi_1$ are bijective.
Since $A \cup B = V(F)$, $\pi$ is surjective.
Furthermore, $\pi$ is injective on $\phi_1^{-1}(A)$ and on $\psi_1^{-1}(B)$, because both $\phi_1$ and $\psi_1$ are injective.
Thus, to see $\pi$ is injective, it suffices to show that $\pi(u) \neq \pi(v)$ for any $u \in \phi_1^{-1}(A) \backslash K$ and $v \in \psi_1^{-1}(B) \backslash K$.
In this case, $\pi(u) \in A \backslash f(K)$ and $\pi(v) \in B \backslash f(K)$, which implies $\pi(u) \neq \pi(v)$ since $A \cap B = f(K)$.
We then further show that $\pi$ is an isomorphism between $H$ and $F$.
We have $H = P_1[\phi_1^{-1}(A)] \cup Q_1[\psi_1^{-1}(B)]$.
On the other hand, it is easy to see that $F = F[A] \cup F[B]$.
Furthermore, $\pi_{|\phi_1^{-1}(A)} = (\phi_1)_{|\phi_1^{-1}(A)}$ and $\pi_{|\psi_1^{-1}(B)} = (\psi_1)_{|\psi_1^{-1}(B)}$.
Thus, $\pi_{|\phi_1^{-1}(A)}$ (resp., $\pi_{|\psi_1^{-1}(B)}$) is an isomorphism between $P_1[\phi_1^{-1}(A)]$ (resp., $Q_1[\psi_1^{-1}(B)]$) and $F[A]$ (resp., $F[B]$).
By Lemma~\ref{lem-isom}, $\pi$ is an isomorphism between $H$ and $F$, which implies that $(H,\pi)$ is an $F$-copy.

We claim that $V(H_1),\dots,V(H_\Delta),V(H)$ form a sunflower with core $X$, and $\pi_{|X} = (\pi_1)_{|X} = \cdots = (\pi_\Delta)_{|X}$.
It suffices to show that $V(H) \cap V(H_i) = X$ and $\pi_{|X} = (\pi_i)_{|X}$ for all $i \in [\Delta]$.
By construction, $X = f_1^{-1}(A) \subseteq \phi_1^{-1}(A) \subseteq V(H)$ and $\pi_{|X} = (\phi_1)_{|X} = (f_1)_{|X} = (\pi_i)_{|X}$ for all $i \in [\Delta]$.
Therefore, we only need to show that $(\bigcup_{i=1}^\Delta (V(H_i) \backslash X)) \cap V(H) = \emptyset$, i.e., $Z \cap V(H) = \emptyset$.
Observe that $X = X_1 \cap \phi_1^{-1}(A)$, because $(\phi_1)_{|X_1} = f_1$ and thus $f_1^{-1}(A) = X_1 \cap \phi_1^{-1}(A)$.
It follows that $\phi_1^{-1}(A) \subseteq (V(P_1) \backslash X_1) \cup X$ and hence
\begin{equation*}
    Z \cap \phi_1^{-1}(A) \subseteq ((V(P_1) \backslash X_1) \cap Z) \cup (X \cap Z) 
    = \emptyset.
\end{equation*}
On the other hand, as $(V(Q_1) \backslash X_i) \cap Z' = \emptyset$ and $(X_i \backslash K) \cap Z' = \emptyset$, we can deduce that
\begin{equation*}
    Z \cap \psi_1^{-1}(B) \subseteq ((V(Q_1) \backslash X_i) \cap Z) \cup ((X_i \backslash K) \cap Z) \cup (K \cap Z) = \emptyset.
\end{equation*}
Applying the fact $Z \cap V(H) = (Z \cap \phi_1^{-1}(A)) \cup (Z \cap \psi_1^{-1}(B))$, we have $Z \cap V(H) = \emptyset$.
To summarize, $V(H_1),\dots,V(H_\Delta),V(H)$ form a sunflower with core $X$, and $\pi_{|X} = (\pi_1)_{|X} = \cdots = (\pi_\Delta)_{|X}$.
However, this contradicts with the choice of $\Delta$, which is supposed to be the \textit{largest} number of $F$-copies in $G$ whose vertex sets form a sunflower with core $X$ and isomorphisms coincide on $X$.
Therefore, we must have $\Delta \geq \gamma^{|X|} \delta$, and thus $(X,F,f)$ is a heavy core in $G$.
\end{proof}

\begin{corollary} \label{cor-heavycore}
    Let $\delta > 0$ and $r = \gamma^{\gamma+1} \delta$ be integers.
    Suppose $(X_1,F,f_1),\dots,(X_r,F,f_r)$ are $(\mathcal{F},\delta)$-heavy cores in $G$ such that $X_1,\dots,X_r$ form a sunflower with core $X$ and there exists $f: X \rightarrow V(F)$ satisfying $f = (f_1)_{|X} = \cdots = (f_r)_{|X}$.
    Then $(X,F,f)$ is also a $(\mathcal{F},\delta)$-heavy core in $G$.
\end{corollary}
\begin{proof}
Set $p = \gamma^\gamma \delta$.
For $t \in [\gamma]$, let $I_t = \{i \in [r]: |X_i| = t\}$.
We have $\sum_{t=1}^\gamma |I_t| = r$, and thus there exists $t \in [\gamma]$ such that $|I_t| \geq r/\gamma = p$.
Without loss of generality, assume $|X_1| = \cdots |X_p| = t$.
Now we apply Lemma~\ref{lem-heavysep} on the heavy cores $(X_1,F,f_1),\dots,(X_p,F,f_p)$ with an arbitrary $i \in [p]$ and $\mathcal{C} = \emptyset$.
Following the notation of Lemma~\ref{lem-heavysep}, we have $X_i^\mathcal{C} = X$ and $f_i^\mathcal{C} = (f_i)_{|X} = f$, and the lemma states that $(X,F,f)$ is a heavy core in $G$.
\end{proof}

\begin{corollary} \label{cor-hitfew}
    Let $\delta > 0$ and $r \geq |\mathcal{F}| \cdot (\gamma^{2 \gamma+1} \delta)^\gamma \gamma!$ be integers, and $U \subseteq V(G)$ be a subset.
    Suppose $(X_1,F_1,f_1),\dots,(X_r,F_r,f_r)$ are distinct $(\mathcal{F},\delta)$-heavy cores in $G$ such that $\bigcap_{i=1}^r (X_i \backslash U) \neq \emptyset$.
    Then $(X_i,F_i,f_i)$ is not $U$-minimal for some $i \in [r]$.
\end{corollary}
\begin{proof}
Set $s = r/|\mathcal{F}| = (\gamma^{2 \gamma+1} \delta)^\gamma \gamma!$.
Then by the Pigeonhole principle, there exist $F \in \mathcal{F}$ and $I \subseteq [r]$ with $|I| = s$ such that $F_i = F$ for all $i \in I$.
Without loss of generality, we can assume that $F_1 = \cdots F_s = F$.
Set $p = \gamma^{\gamma+1} \delta$.
By Fact~\ref{fact-goodsunflower}, there exists $I \subseteq [s]$ with $|I| = p$ such that $\{X_i: i \in I\}$ form a sunflower with core $X$ and $(f_i)_{|X} = (f_j)_{|X}$ for all $i,j \in I$.
Again, we can assume $I = [p]$ without loss of generality.
Note that $X \backslash U \neq \emptyset$, since 
\begin{equation*}
    X \backslash U = \bigcap_{i=1}^p (X_i \backslash U) \supseteq \bigcap_{i=1}^r (X_i \backslash U) \neq \emptyset.
\end{equation*}


Let $f: X \rightarrow V(F)$ be the map such that $f = (f_1)_{|X} = \cdots = (f_p)_{|X}$.
By Corollary~\ref{cor-heavycore}, $(X,F,f)$ is a heavy core in $G$.
Note that $X \subsetneq X_i$ for some $i \in [p]$.
Indeed, if $X = X_i$ for all $i \in [p]$, then $(X,F,f) = (X_1,F_1,f_1) = \cdots = (X_p,F_p,f_p)$, which contradicts with the fact that the heavy cores are distinct.
Thus, $(X_i,F_i,f_i)$ is not $U$-minimal, as $(X,F,f) \prec (X_i,F_i,f_i)$ and $X \backslash U \neq \emptyset$.
\end{proof}

Now we are ready to prove our key structural lemma in this section, Lemma~\ref{lem-twGaifman}.
As this lemma is important, we state it in a self-contained form as below.

\begin{lemma} \label{lem-twGaifman}
    Let $\mathcal{G} \subseteq \mathcal{G}(\eta,0,\rho)$ where $\eta \geq 0$ and $0 \leq \rho < 1$.
    Also, let $\mathcal{F}$ be a finite set of graphs.
    Then for any $G \in \mathcal{G}$, any integer $\delta > \textnormal{wcol}_\gamma(G)$, any $U \subseteq V(G)$, and any $U$-minimal $(\mathcal{F},\delta)$-heavy cores $(X_1,F_1,f_1),\dots,(X_r,F_r,f_r)$ in $G$, the treewidth of the Gaifman graph of $\{X_1 \backslash U,\dots,X_r \backslash U\}$ is $(\eta\delta)^{O(1)} \cdot k^\rho$, where $k$ denotes the size of a minimum hitting set for $\{X_1 \backslash U,\dots,X_r \backslash U\}$.
    The constant hidden in $O(\cdot)$ only depends on $\rho$ and $\mathcal{F}$.
\end{lemma}
\begin{proof}
Without loss of generality, we can assume that $X_i \neq X_j$ for all distinct $i,j \in [r]$, for otherwise we can remove $(X_i,F_i,f_i)$ or $(X_j,F_j,f_j)$ while keeping the Gaifman graph unchanged.
We first show that $|\bigcup_{i=1}^r X_i| = \delta^{O(1)} \cdot k$.
Let $\gamma = \max_{F \in \mathcal{F}} |V(F)|$, and $S \subseteq V(G) \backslash U$ be a hitting set of $\{X_1 \backslash U,\dots,X_r \backslash U\}$ with $|S| = k$.
Then $S$ is also a hitting set of $\{X_1,\dots,X_r\}$.
By Corollary~\ref{cor-hitfew}, each vertex in $S$ can hit at most $|\mathcal{F}| \cdot (\gamma^{2 \gamma+1} \delta)^\gamma \gamma! - 1$ sets in $\{X_1,\dots,X_r\}$, since $(X_i,F_i,f_i)$ is $U$-minimal for all $i \in [r]$.
So we must have $r < |\mathcal{F}| \cdot (\gamma^{2 \gamma+1} \delta)^\gamma \gamma! \cdot k$, which implies that
\begin{equation*}
    \left|\bigcup_{i=1}^r X_i\right| \leq \gamma r < \gamma|\mathcal{F}| \cdot (\gamma^{2 \gamma+1} \delta)^\gamma \gamma! \cdot k = \delta^{O(1)} \cdot k.
\end{equation*}
Fix an ordering $\sigma$ of $V(G)$ such that $\text{wcol}_\gamma(G) = \text{wcol}_\gamma(G,\sigma)$.
Let $G'$ be the graph obtained from $G$ by adding edges $(u,v)$ for $u,v \in V(G)$ such that $v \in \text{WR}_\gamma(G,\sigma,u)$.
Then $G' \in \mathcal{G}'$, where $\mathcal{G}'$ is the $\gamma$-augmentation of $\mathcal{G}$.
By Lemma~\ref{lem-augmentation}, $\mathcal{G}' \subseteq \mathcal{G}(\eta^{O(1)},0,\rho)$.
Then Lemma~\ref{lem-separator=tw} implies that $\mathbf{tw}(G'[\bigcup_{i=1}^r X_i]) = (\eta \delta)^{O(1)} \cdot k^\rho$.
Fix a tree decomposition $(T,\beta)$ of $G'[\bigcup_{i=1}^r X_i]$ of width $(\eta \delta)^{O(1)} \cdot k^\rho$.
Our goal is to modify $(T,\beta)$ to a tree decomposition of the Gaifman graph $G^*$ of $\{X_1 \backslash U,\dots,X_r \backslash U\}$, without increasing its width too much.
The modification is done as follows.
For each $i \in [r]$, we pick a node $t_i \in T$ such that $\beta(t_i) \cap (X_i \backslash U) \neq \emptyset$.
For two nodes $t,t' \in T$, denote by $\pi_{t,t'}$ as the (unique) path in $T$ connecting $t$ and $t'$.
Then for each node $t \in T$ and each $i \in [r]$, we define $\beta_i^*(t)$ as the set of all vertices $v \in X_i \backslash U$ such that $t$ is on the path $\pi_{t_i,t'}$ for some node $t' \in T$ with $v \in \beta(t')$.
Set $\beta^*(t) = \bigcup_{i=1}^r \beta_i^*(t)$ for all $t \in T$.
We claim that $(T,\beta^*)$ is a tree decomposition of $G^*$.
Observe that $X_i \backslash U \subseteq \beta_i^*(t_i) \subseteq \beta^*(t_i)$ for all $i \in [r]$.
Thus, for every edge $(u,v) \in E(G^*)$, there exists $t \in T$ such that $u,v \in \beta^*(t)$.
For each vertex $v \in V = V(G^*)$, the set $T_v = \{t \in T: v \in \beta(t)\}$ of nodes are connected in $T$, because $(T,\beta)$ is a tree decomposition.
Let $T_v^* = \{t \in T: v \in \beta^*(t)\}$.
By construction, $T_v^*$ consists of exactly the nodes on the paths $\pi_{t_i,t}$ for all $t \in T_v$ and all $i \in [r]$ such that $v \in X_i \backslash U$.
These paths cover the entire $T_v$ (i.e., $T_v \subseteq T_v^*$) and each of them contains at least one vertex in $T_v$.
As $T_v$ is connected in $T$ and the paths are also connected, $T_v^*$ is connected in $T$.
Therefore, $(T,\beta^*)$ is a tree decomposition of $G^*$.

Next, we bound the width of $(T,\beta^*)$, which is the most challenging part of the proof.
Consider a node $t \in T$.
Observe that $\beta(t) \cap (\bigcup_{i=1}^r (X_i \backslash U)) \subseteq \beta^*(t)$, since $\beta(t) \cap (X_i \backslash U) \subseteq \beta_i^*(t)$ by our construction.
We shall prove that $|\beta^*(t)| = (\eta \delta)^{O(1)} \cdot |\beta(t)|$.
Let $I^* = \{i \in [r]: \beta_i^*(t) \neq \emptyset\}$.
Since $|\beta_i^*(t)| \leq |X_i \backslash U| \leq \gamma$, we have $|\beta^*(t)| \leq \gamma |I^*|$ and thus it suffices to show that $|I^*| = \delta^{O(1)} \cdot |\beta(t)|$.
Now let $I = \{i \in [r]: \beta(t) \cap (X_i \backslash U) \neq \emptyset\}$.
By Corollary~\ref{cor-hitfew}, each vertex in $\beta(t)$ can hit at most $|\mathcal{F}| \cdot (\gamma^{2 \gamma+1} \delta)^\gamma \gamma! - 1 = O(\delta^\gamma)$ sets in $\{X_1 \backslash U,\dots,X_r \backslash U\}$.
So we have $|I| = O(\delta^\gamma |\beta(t)|)$ and it suffices to show that $|I^* \backslash I| = (\eta \delta)^{O(1)} \cdot |\beta(t)|$.
We make the following claims.
\medskip

\noindent
\textit{Claim 1.}
For every $i \in I^* \backslash I$, there exist two connected components $C$ and $C'$ of $F - f_i(X_i \cap \beta(t))$ such that $f_i^{-1}(V(C)) \nsubseteq U$ and $f_i^{-1}(V(C')) \nsubseteq U$.
\medskip

\noindent
\textit{Proof.}
We have $\beta(t_i) \cap (X_i \backslash U) \neq \emptyset$ by the choice of $t_i$.
Pick a vertex $v \in \beta(t_i) \cap (X_i \backslash U)$.
As $i \in I^*$, $\beta_i^*(t) \neq \emptyset$ and thus there exists $v' \in \beta_i^*(t)$.
Furthermore, since $i \notin I$, we have $\beta(t) \cap (X_i \backslash U) = \emptyset$, which implies $v,v' \notin \beta(t)$ because $v,v' \in X_i \backslash U$.
Therefore, both $f_i(v)$ and $f_i(v')$ are vertices of $F - f_i(X_i \cap \beta(t))$.
Let $C$ and $C'$ be the connected components of $F - f_i(X_i \cap \beta(t))$ containing $f_i(v)$ and $f_i(v')$, respectively.
We have $V(C),V(C') \nsubseteq U$, simply because $v,v' \notin U$.
It now suffices to show that $C \neq C'$.
Assume $C = C'$, i.e., $f_i(v)$ and $f_i(v')$ lie in the same connected component of $F - f_i(X_i \cap \beta(t))$.
Let $\pi$ be a simple path from $f_i(v)$ to $f_i(v')$ in $F - f_i(X_i \cap \beta(t))$.
Suppose $a_0,a_1,\dots,a_p \in V(F)$ are the vertices on $\pi$ that are in $f_i(X_i)$, sorted along $\pi$, where $a_0 = f_i(v)$ and $a_p = f_i(v')$.
Define $x_0,x_1,\dots,x_p \in X_i$ as the pre-images of $a_0,a_1,\dots,a_p$ under $f_i$, respectively.
Then $x_0 = v$ and $x_p = v'$.
Since $\pi$ is a path in $F - f_i(X_i \cap \beta(t))$, $a_0,a_1,\dots,a_p \notin f_i(X_i \cap \beta(t))$ and thus $x_0,x_1,\dots,x_p \notin X_i \cap \beta(t)$, which implies $x_0,x_1,\dots,x_p \notin \beta(t)$.
For each $j \in [p]$, let $\pi_j$ be the sub-path of $\pi$ between $a_{j-1}$ and $a_j$.
Note that the internal nodes of $\pi_j$ are all in $V(F) \backslash f_i(X_i)$, and thus $f_i^{-1}(V(\pi_j)) = \{a_{j-1},a_j\}$.
Since $\delta > \text{wcol}_\gamma(G) = \text{wcol}_\gamma(G,\sigma)$ and the length of $\pi_j$ is at most $\gamma$, by Fact~\ref{fact-connect}, we have $x_j \in \text{WR}_\gamma(G,\sigma,x_{j-1})$ if $x_{j-1} <_\sigma x_j$ and $x_{j-1} \in \text{WR}_\gamma(G,\sigma,x_j)$ if $x_j <_\sigma x_{j-1}$.
In either case, we have $(x_{j-1},x_j) \in E(G')$.
It follows that $(x_0,x_1,\dots,x_p)$ is a path in $G'[X_i]$.
So the nodes $s \in T$ with $\beta(s) \cap \{x_0,x_1,\dots,x_p\} \neq \emptyset$ are connected in $T$, by the property of a tree decomposition.
As $v' \in \beta_i^*(t)$, $t$ is on the path $\pi_{t_i,t'}$ in $T$ for a node $t' \in T$ with $v' \in \beta(t')$.
We have $v \in \beta(t_i) \cap \{x_0,x_1,\dots,x_p\}$ and hence $\beta(t_i) \cap \{x_0,x_1,\dots,x_p\} \neq \emptyset$.
On the other hand, we have $v' \in \beta(t') \cap \{x_0,x_1,\dots,x_p\}$ and hence $\beta(t') \cap \{x_0,x_1,\dots,x_p\} \neq \emptyset$.
Since $t$ is on the path $\pi_{t_i,t'}$, we should also have $\beta(t) \cap \{x_0,x_1,\dots,x_p\} \neq \emptyset$.
However, this contradicts with the fact that $x_0,x_1,\dots,x_p \notin \beta(t)$.
Therefore, $C \neq C'$.
\hfill $\lhd$
\medskip

\noindent
\textit{Claim 2.}
For every $i \in I^* \backslash I$, there exist $Y_i \subseteq X_i$ and $v_i,v_i' \in X_i \backslash U$ satisfying the following.
\begin{itemize}
    \item $Y_i \subseteq \text{WR}_\gamma(G,\sigma,v)$ for some $v \in \beta(t)$.
    \item $f_i(v_i)$ and $f_i(v_i')$ do not lie in the same connected component of $F - f_i(Y_i)$.
\end{itemize}
\medskip

\noindent
\textit{Proof.}
By Claim~1, there exist two connected components $C$ and $C'$ of $F-f_i(X_i \cap \beta(t))$ such that $f_i^{-1}(V(C)) \nsubseteq U$ and $f_i^{-1}(V(C')) \nsubseteq U$.
We arbitrarily pick vertices $v_i \in f_i^{-1}(V(C)) \backslash U$ and $v_i' \in f_i^{-1}(V(C')) \backslash U$.
Furthermore, we define $Y_i \subseteq X_i$ as follows.
Let $\varPi$ be the set of all simple paths in $F$ from $f_i(v_i)$ to a vertex in $f_i(X_i \cap \beta(t))$ in which all internal nodes are in $V(F) \backslash f_i(X_i \cap \beta(t))$.
For every vertex $u \in X_i$, we include $u$ in $Y_i$ if there exists $\pi \in \varPi$ such that $u$ is the largest vertex (under the ordering $\sigma$) in $f_i^{-1}(V(\pi))$.
We show that $Y_i$, $v_i$, and $v_i'$ satisfies the desired two properties.
Define $u$ as the smallest vertex in $Y_i$ (under the ordering $\sigma$).
By construction, $u$ is the largest vertex in $f_i^{-1}(V(\pi))$ for some $\pi \in \varPi$.
The path $\pi$ connects $f_i(v_i)$ and a vertex $z \in f_i(X_i \cap \beta(t))$.
Let $v \in X_i \cap \beta(t)$ with $f_i(v) = z$.
We claim that $Y_i \subseteq \text{WR}_\gamma(G,\sigma,v)$.
Consider a vertex $y \in Y_i$.
There exists $\pi' \in \varPi$ such that $y$ is the largest vertex in $f_i^{-1}(V(\pi'))$.
Concatenating $\pi$ and $\pi'$, we obtain a path $\pi''$ between $f_i(v)$ and $f_i(y)$.
We have $u \leq_\sigma y$, as $u$ is the smallest vertex in $Y_i$.
Therefore, $y$ is the largest vertex in $f_i^{-1}(V(\pi''))$, since $u$ is the largest vertex in $f_i^{-1}(V(\pi))$ and $y$ is the largest vertex in $f_i^{-1}(V(\pi'))$.
Note that $\pi''$ is not necessarily a simple path.
But we can find a simple path $\hat{\pi}$ in $F$ between $f_i(v)$ and $f_i(y)$ such that $V(\hat{\pi}) \subseteq V(\pi'')$.
As $y \in f_i^{-1}(V(\hat{\pi})) \subseteq f_i^{-1}(V(\pi''))$, $y$ is also the largest vertex in $f_i^{-1}(V(\hat{\pi}))$.
Since the length of $\hat{\pi}$ is at most $\gamma$, by Fact~\ref{fact-connect}, we directly have $y \in \text{WR}_\gamma(G,\sigma,v)$.

To see the second property, we can assume that $v_i, v_i' \notin Y_i$, for otherwise one of $f_i(v_i)$ and $f_i(v_i')$ is not a vertex in $F - f_i(Y_i)$.
We need to show that every path $\pi$ in $F$ between $f(v_i)$ and $f(v_i')$ contains some vertex in $f_i(Y_i)$.
Recall that $f_i(v_i) \in V(C)$ and $f_i(v_i') \in V(C')$, where $C$ and $C'$ are two connected components of $F-f_i(X_i \cap \beta(t))$.
Thus, $\pi$ must contain some vertex in $f_i(X_i \cap \beta(t))$.
Let $z \in V(\pi) \cap f_i(X_i \cap \beta(t))$ be the vertex that is closest to $f_i(v_i)$ on $\pi$, and $\pi'$ be the sub-path of $\pi$ between $f_i(v_i)$ and $z$.
Then $\pi' \in \varPi$.
The largest vertex in $f_i^{-1}(V(\pi'))$ is contained in $Y_i$, which implies that $\pi'$ (and hence $\pi$) contains some vertex in $f_i(Y_i)$.
\hfill $\lhd$
\medskip

Now we are ready to bound $|I^* \backslash I|$.
We charge each index $i \in I^* \backslash I$ to the set $Y_i$ in Claim~2.
Note that $Y_i \subseteq X_i$ and $Y_i \subseteq \text{WR}_\gamma(G,\sigma,v)$ for some $v \in \beta(t)$ by Claim~2, which implies that $Y_i$ is a subset of $\text{WR}_\gamma(G,\sigma,v)$ of size at most $\gamma$.
As $|\text{WR}_\gamma(G,\sigma,v)| \leq \text{wcol}_\gamma(G,\sigma) = \text{wcol}_\gamma(G)$, the number of distinct $Y_i$'s is bounded by $\text{wcol}_\gamma^\gamma(G) \cdot |\beta(t)|$, which is $\eta^{O(1)} \cdot |\beta(t)|$.
It suffices to show that for each $v \in \beta(t)$, each subset $Y \subseteq \text{WR}_\gamma(G,\sigma,v)$ of size at most $\gamma$ can get charged at most $\delta^{O(1)}$ times.
Assume $Y$ gets charged $d = |\mathcal{F}| \gamma \cdot (|\mathcal{F}| \cdot (\gamma^{2 \gamma+2} \delta)^\gamma \gamma!)^\gamma \gamma! = \delta^{O(1)}$ times, in order to deduce a contradiction.
We may assume that the indices charged to $Y$ are $1,\dots,d$, without loss of generality.
These indices correspond to the heavy cores $(X_1,F_1,f_1),\dots,(X_d,F_d,f_d)$.
Set 
\begin{equation*}
    s = d/(|\mathcal{F}| \gamma) = (|\mathcal{F}| \cdot (\gamma^{2 \gamma+2} \delta)^\gamma \gamma!)^\gamma \gamma!.
\end{equation*}
By the Pigeonhole principle, there exist $F \in \mathcal{F}$ and $J \subseteq [d]$ with $|J| = s$ such that $F = F_j$ for all $j \in \mathcal{J}$ and $|X_i| = |X_j|$ for all $i,j \in J$.
Without loss of generality, assume $F_1 = \cdots = F_s = F$ and $|X_1| = \cdots = |X_s|$.
Set $p = |\mathcal{F}| \cdot (\gamma^{2 \gamma+1} \delta)^\gamma \gamma!$.
By Fact~\ref{fact-goodsunflower}, there exists $J \subseteq [s]$ with $|J| = p$ such that $\{X_j: j \in J\}$ form a sunflower with core $K \subseteq V(G)$ and $(f_i)_{|K} = (f_j)_{|K}$ for all $i,j \in J$.
Again, we can assume $J = \{1,\dots,p\}$ without loss of generality.

Let $f:K \rightarrow V(F)$ be the map such that $f = (f_1)_{|K} = \cdots = (f_p)_{|K}$.
By Corollary~\ref{cor-hitfew}, each vertex in $V(G) \backslash U$ can hit at most $|\mathcal{F}| \cdot (\gamma^{2 \gamma+1} \delta)^\gamma \gamma! -1 = p-1$ sets in $\{X_1,\dots,X_p\}$, because $(X_1,F,f_1),\dots,(X_p,F,f_p)$ are $U$-minimal.
So we have $K \subseteq U$.
On the other hand, we have $Y \subseteq K$, simply because every $i \in [p]$ is charged to $Y$ and hence $Y \subseteq X_i$ for all $i \in [p]$.
The heavy cores $(X_1,F,f_1),\dots,(X_p,F,f_p)$ satisfy the conditions in Lemma~\ref{lem-heavysep}.
Thus, by Lemma~\ref{lem-heavysep}, for any $i \in [p]$ and any set $\mathcal{C}$ of connected components of $F - f(K)$, $(X_i^\mathcal{C},F,f_i^\mathcal{C})$ is a $(\mathcal{F},\delta)$-heavy core in $G$, where $X_i^{\mathcal{C}} = K \cup (\bigcup_{C \in \mathcal{C}} f_i^{-1}(V(C)))$ and $f_i^{\mathcal{C}}$ is the restriction of $f_i$ to $X_i^{\mathcal{C}}$.
Take an arbitrary $i \in [p]$.
By Claim~2, there exist $v_i,v_i' \in X_i \backslash U$ such that $f_i(v_i)$ and $f_i(v_i')$ do not lie in the same connected component of $F - f_i(Y)$.
As $Y \subseteq K$, $f_i(v_i)$ and $f_i(v_i')$ also do not lie in the same connected component of $F - f_i(K)$.
Note that $v_i,v_i' \notin K$, since $K \subseteq U$ and $v_i,v_i' \in X_i \backslash U$.
Therefore, $f_i(v_i)$ and $f_i(v_i')$ are both vertices of $F - f_i(K)$.
Let $C$ be the connected component of $F - f_i(K)$ containing $f_i(v_i)$, and $\mathcal{C} = \{C\}$.
Then $f_i(v_i') \notin V(C)$, which implies that $X_i^{\mathcal{C}} \subsetneq X_i$ and hence $(X_i^\mathcal{C},F,f_i^\mathcal{C}) \prec (X_i,F,f_i)$.
We have $X_i^\mathcal{C} \nsubseteq U$, as $v_i \in X_i^\mathcal{C} \backslash U$.
It follows that $(X_i,F,f_i)$ is not $U$-minimal, contradicting with our assumption.
So we conclude that $Y$ can get charged at most $d-1 = \delta^{O(1)}$ times, which implies $|I^* \backslash I| = (\eta \delta)^{O(1)} |\beta(t)|$ and $|\beta^*(t)| = (\eta \delta)^{O(1)} |\beta(t)|$.
Recall that $|\beta(t)| = (\eta \delta)^{O(1)} \cdot k^\rho$ since the width of $(T,\beta)$ is $(\eta \delta)^{O(1)} \cdot k^\rho$, and thus $|\beta^*(t)| = (\eta \delta)^{O(1)} \cdot k^\rho$.
As a result, $\mathbf{tw}(G^*) = (\eta \delta)^{O(1)} \cdot k^\rho$, completing the proof.
\end{proof}

In the rest of this section, we introduce some additional definitions and properties for heavy cores.
Let $U \subseteq V(G)$.
For each $F \in \mathcal{F}$, we say an $F$-copy $(H,\pi)$ in $G$ is \textit{$U$-redundant} if there exists another $F$-copy $(H',\pi')$ in $G$ such that $V(H')\nsubseteq U$ and $V(H') \backslash U \subsetneq V(H) \backslash U$.
We say an $(\mathcal{F},\delta)$-heavy core $(X,F,f)$ in $G$ is \textit{$U$-redundant} if \textit{every} $F$-copy $(H,\pi)$ in $G$ satisfying $(X,F,f) \preceq (V(H),F,\pi)$ is \textit{$U$-redundant}.
The $U$-redundancy of heavy cores satisfies the following monotonicity with the partial order $\prec$.

\begin{fact} \label{fact-redundancy}
    Let $\delta>0$ be an integer, and $(X,F,f),(Y,F,g)$ be $(\mathcal{F},\delta)$-heavy cores in $G$ such that $(X,F,f) \prec (Y,F,g)$.
    For any $U \subseteq V(G)$, if $(X,F,f)$ is $U$-redundant, then $(Y,F,g)$ is $U$-redundant.
\end{fact}
\begin{proof}
Suppose $(X,F,f)$ is $U$-redundant.
We show that every $F$-copy $(H,\pi)$ in $G$ satisfying $(Y,F,g) \preceq (V(H),F,\pi)$ is $U$-redundant, which implies $(Y,F,g)$ is $U$-redundant.
Consider such an $F$-copy $(H,\pi)$.
Since $(X,F,f) \prec (Y,F,g)$, we also have $(X,F,f) \preceq (V(H),F,\pi)$.
As $(X,F,f)$ is $U$-redundant, by definition, $(H,\pi)$ is $U$-redundant.
\end{proof}

Now we are ready to define the last notion we need, active heavy cores.
We say an $(\mathcal{F},\delta)$-heavy core $(X,F,f)$ in $G$ is \textit{$U$-active} if it is $U$-minimal but not $U$-redundant.
An important property of active heavy cores is the following lemma, which essentially states that being an $\mathcal{F}$-hitting set disjoint from $U$, it suffices to hit all $U$-active heavy cores.

\begin{lemma} \label{lem-hitactive}
    Let $G$ be a graph, $\mathcal{F}$ be a finite set of graphs, $\delta > 0$ be an integer, and $U \subseteq V(G)$ such that $G[U]$ is $\mathcal{F}$-free.
    For any set $S \subseteq V(G) \backslash U$, if $S \cap X \neq \emptyset$ for every $U$-active $(\mathcal{F},\delta)$-heavy core $(X,F,f)$ in $G$, then $S$ is an $\mathcal{F}$-hitting set of $G$.
\end{lemma}
\begin{proof}
Suppose $S \subseteq V(G) \backslash U$ satisfies $S \cap X \neq \emptyset$ for every $U$-active $(\mathcal{F},\delta)$-heavy core $(X,F,f)$ in $G$.
We prove $S$ is an $\mathcal{F}$-hitting set of $G$ by contradiction.
Assume that $S$ is not an $\mathcal{F}$-hitting set.
Let $F \in \mathcal{F}$ such that $G-S$ is not $F$-free, and pick an $F$-copy $(H,\pi)$ in $G-S$ that minimizes $|V(H) \backslash U|$.
Note that $|V(H) \backslash U| \geq 1$, because $G[U]$ is $\mathcal{F}$-free and thus $V(H) \nsubseteq U$.
By Fact~\ref{fact-Fcopyheavy}, $(V(H),F,\pi)$ is a heavy core.
Since $S \cap V(H) = \emptyset$, $(V(H),F,\pi)$ is not $U$-active.
We distinguish two cases: $(V(H),F,\pi)$ is not $U$-redundant and $(V(H),F,\pi)$ is $U$-redundant.

If $(V(H),F,\pi)$ is not $U$-redundant, then it is not $U$-minimal, for otherwise it is $U$-active.
So there exists a $U$-minimal heavy core $(X,F,f)$ of $G$ such that $(X,F,f) \prec (V(H),F,\pi)$.
By Fact~\ref{fact-redundancy}, $(X,F,f)$ is not $U$-redundant either.
Thus, $(X,F,f)$ is $U$-active and thus $S \cap X \neq \emptyset$ by our assumption.
But this contradicts the fact that $S \cap V(H) = \emptyset$, since $X \subseteq V(H)$.

If $(V(H),F,\pi)$ is $U$-redundant, then $(H,\pi)$ is $U$-redundant.
So there exists another $F$-copy $(H',\pi')$ in $G$ such that $V(H') \backslash U \subsetneq V(H) \backslash U$.
Since $S \subseteq V(G) \backslash U$ and $S \cap V(H) = \emptyset$, we have $S \cap V(H') = \emptyset$.
This contradicts the fact that $(H,\pi)$ is the $F$-copy with $S \cap V(H) = \emptyset$ that minimizes $|V(H) \backslash U|$, as $(H',\pi')$ is also an $F$-copy satisfying $S \cap V(H') = \emptyset$ and $|V(H') \backslash U| < |V(H) \backslash U|$.
\end{proof}

\subsection{The branching algorithm} \label{sec-branchalg}

Based on the notions and results established in Section~\ref{sec-heavy}, we are now ready to present our branching algorithm for generating the collections $\mathcal{X}_1,\dots,\mathcal{X}_t$ in Theorem~\ref{thm-subtwbranch}.

As our actual algorithm is somehow involved, we shall first describe a very simple branching algorithm that can generate collections $\mathcal{X}_1,\dots,\mathcal{X}_t$ satisfying all conditions in Theorem~\ref{thm-subtwbranch}, except that the number $t$ and the running time of the algorithm are unbounded.
Let $\delta > 0$ be a parameter to be determined.
Imagine there is some (unknown) $\mathcal{F}$-hitting set $S$ of $G$ such that $|S| \leq k$.
Basically, the branching algorithm guesses whether $S$ hits each $(\mathcal{F},\delta)$-heavy cores in $G$.
It maintains a set $U \subseteq V(G)$ and a collection $\mathcal{X}$ of subsets of $V(G)$.
The vertices in $U$ are ``undeletable'' vertices, namely, the vertices that are not supposed to be in $S$.
On the other hand, the sets in $\mathcal{X}$ are supposed to be hit by $S$.
Initially, $U = \emptyset$ and $\mathcal{X} = \emptyset$.
The algorithm finds a $U$-active heavy core $(X,F,f)$ in $G$ such that $X \notin \mathcal{X}$, and then guesses whether $X$ is hit by $S$ or not.
If it decides that $X$ is hit by $S$ (which we call the ``yes'' branch), then $X$ is added to $\mathcal{X}$.
On the other hand, if it decides that $X$ is not hit by $S$ (which we call the ``no'' branch), then the vertices in $X$ are added to $U$.
The algorithm tries both branches and in each branch it proceeds with the new $U$ and $\mathcal{X}$.
The branching terminates when there does not exist a $U$-active heavy core $(X,F,f)$ with $X \notin \mathcal{X}$.
At this point, if $G[U]$ is $\mathcal{F}$-free and $\{X \backslash U: X \in \mathcal{X}\}$
admits a hitting set of size at most $k$\footnote{Of course we cannot afford testing whether $\{X \backslash U: X \in \mathcal{X}\}$ admits a size-$k$ hitting set or not. But as we do not care about the running time of this simple algorithm, let us ignore this technical detail at this point.}, then we construct a new collection $\mathcal{X}_i = \{X \backslash U: X \in \mathcal{X}\}$.
The algorithm then goes back to the last level and tries other branches.
The entire algorithm terminates when it exhausts all branches in all levels.

Consider the collections $\mathcal{X}_1,\dots,\mathcal{X}_t$ constructed during the branching procedure.
First, we show that for any $S \subseteq V(G)$ with $|S| \leq k$, $S$ is an $\mathcal{F}$-hitting set of $G$ iff $S$ is a hitting set of $\mathcal{X}_i$ for some $i \in [t]$.
Suppose $S$ is a hitting set of $\mathcal{X}_i$.
Recall that $\mathcal{X}_i$ is constructed when we are in the situation that $X \notin \mathcal{X}$ for every $U$-active heavy core $(X,F,f)$, and $\mathcal{X}_i = \{X \backslash U: X \in \mathcal{X}\}$.
Thus, by Lemma~\ref{lem-hitactive}, $S$ is an $\mathcal{F}$-hitting set.
On the other hand, suppose $S$ is an $\mathcal{F}$-hitting set.
When the branching algorithm makes the right decision for $S$ in every step (i.e., it chooses the ``yes'' branch whenever the current $X$ is hit by $S$ and chooses the ``no'' branch whenever the current $X$ is not hit by $S$), we finally terminate with some $U$ and $\mathcal{X}$ such that $S \cap U = \emptyset$ and $S$ is a hitting set of $\mathcal{X}$.
We have $G[U]$ is $\mathcal{F}$-free, since $S \cap U = \emptyset$ and $S$ is an $\mathcal{F}$-hitting set.
Furthermore, $\{X \backslash U: X \in \mathcal{X}\}$ admits a hitting set of size at most $k$, which is $S$.
So in the algorithm, we constructed a collection $\mathcal{X}_i = \{X \backslash U: X \in \mathcal{X}\}$ and $S$ is a hitting set of $\mathcal{X}_i$.
Next, we bound the treewidth of the Gaifman graph of each $\mathcal{X}_i$.
Recall that the collection $\mathcal{X}_i$ is constructed only when $\{X \backslash U: X \in \mathcal{X}\}$ admits a hitting set of size at most $k$, and we set $\mathcal{X}_i = \{X \backslash U: X \in \mathcal{X}\}$.
Every $X \in \mathcal{X}$ corresponds to a heavy core $(X,F,f)$ which is $U'$-active for some $U' \subseteq U$ (because $U$ is expanding during the branching procedure), which implies that $(X,F,f)$ is $U'$-minimal and thus $U$-minimal, 
{because $X\nsubseteq U$ as $S\cap X\neq \emptyset$ and $S\cap U=\emptyset$ along with another property needed for being $U$-minimal which follows from the fact that $(X,F,f)$ is $U'$-minimal}.
Thus, by Lemma~\ref{lem-twGaifman}, the Gaifman graph of $\{X \backslash U: X \in \mathcal{X}\}$ has treewidth $\delta^{O(1)} k^\rho$.
If we choose $\delta = k^\varepsilon$ for a sufficiently small $\varepsilon > 0$, then the treewidth is sublinear in $k$.
To summarize, this simple branching algorithm satisfies all conditions of Theorem~\ref{thm-subtwbranch} except the bounds on $t$ and the running time.

To achieve the desired bounds on the number of collections constructed and the running time, the key is to guarantee that along any path in the branching tree from the root to a leaf, the numbers of ``yes'' decisions and ``no'' decisions are bounded.
It turns out that bounding the number of ``yes'' decisions is easy, even without changing the above algorithm.
To see this, consider a successful leaf of the branching tree (i.e., the leaf where a new collection $\mathcal{X}_i$ is constructed).
At that leaf, our sets $U$ and $\mathcal{X}$ satisfy that $\{X \backslash U: X \in \mathcal{X}\}$ admits a hitting set of size at most $k$.
The sets in $\mathcal{X}$ are all $U$-minimal, as argued before.
By Corollary~\ref{cor-hitfew}, any vertex in $V(G) \backslash U$ can hit at most $O(\delta^\gamma)$ sets in $\mathcal{X}$.
Equivalently, any vertex of $G$ can hit at most $O(\delta^\gamma)$ sets in $\{X \backslash U: X \in \mathcal{X}\}$.
Thus, we have $|\mathcal{X}| = O(\delta^\gamma k)$, i.e., the number of ``yes'' decisions we made is $O(\delta^\gamma k)$.
The difficult part is to have a bounded number of ``no'' decisions.
Note that in order to have a subexponential running time, the number of ``no'' decisions must be sublinear in $k$.
To achieve this, new ideas are needed to substantially modify the above algorithm.

In what follows, we present our actual branching algorithm.
Compared to the above simple algorithm, the main difference in our actual algorithm is that, when it decides a set $X$ is not hit by the solution $S$, instead of simply adding the vertices in $X$ to $U$ and going to the next level, it further guesses an additional set $P$ of vertices that are not in $S$ and adds the vertices in $P$ to $U$ as well.
As such, at each stage of branching, we have one ``yes'' branch and \textit{multiple} ``no'' branches (corresponding to different guesses for the set $P$).
The number of sets $P$ to be guessed (and thus the number of ``no'' branches) will be bounded by $\delta^{O(1)}$.
Furthermore, by properly defining these sets, we can guarantee that the number of ``no'' decisions on any successful path in the branching tree is $O(k/\delta)$.
Next, we describe the branching algorithm in detail, which is shown in Algorithm~\ref{alg-branch}.
Let $\delta,\theta_\textsf{yes},\theta_\textsf{no}$ be parameters to be determined, where $\delta$ is used for defining heavy cores and $\theta_\textsf{yes}$ (resp., $\theta_\textsf{no}$) is the total budget for the ``yes'' decisions (resp., ``no'' decisions) we can make in the branching procedure.
These parameters will depend on $\mathcal{G}$, $\mathcal{F}$, and $k$.
We use $\varGamma$ to denote the set of all $(\mathcal{F},\gamma)$-heavy cores in $G$.
The core of our algorithm is the procedure \textsc{Branch} (line~4-15), which has four arguments $U$, $\mathcal{X}$, $d_\mathsf{yes}$, $d_\mathsf{no}$.
Here, $U$ and $\mathcal{X}$ are the same as in the simple algorithm, while $d_\mathsf{yes}$ and $d_\mathsf{no}$ are the remaining budget for the ``yes'' and ``no'' decisions, respectively.
In the initial call of \textsc{Branch} (line~2), we have $U = \emptyset$, $\mathcal{X} = \emptyset$, $d_\mathsf{yes} = \theta_\mathsf{yes}$, and $d_\mathsf{no} = \theta_\mathsf{no}$.
In the procedure \textsc{Branch}, we first check whether $d_\mathsf{yes} < 0$ or $d_\mathsf{no} < 0$ (i.e., whether we have made too many ``yes'' or ``no'' decisions), and if so, we do not branch further and go back to the last level (line~5).
In line~6, we find a $U$-active heavy core $(X,F,f) \in \varGamma$ such that $X \notin \mathcal{X}$.
If such a heavy core exists, we branch on it (line~7-12).
Otherwise, we go to line~13 to decide whether a new collection should be constructed.
Let us first explain the part of line~13-15, and the branching part will be discussed later.
Recall that in the simple algorithm, the conditions for constructing a new collection are \textbf{(i)} $G[U]$ is $\mathcal{F}$-free and \textbf{(ii)} $\{X \backslash U: X \in \mathcal{X}\}$ admits a hitting set of size at most $k$.
As aforementioned, we cannot afford checking if condition (ii) holds or not.
Fortunately, all sets in $\{X \backslash U: X \in \mathcal{X}\}$ are of size at most $\gamma$.
In this case, one can compute a $\gamma$-approximation solution for the minimum hitting set of $\{X \backslash U: X \in \mathcal{X}\}$ (in polynomial time), which is sufficient for our purpose.
(Indeed, we can simply compute a maximal sub-collection $\mathcal{Z} \subseteq \{X \backslash U: X \in \mathcal{X}\}$ such that the sets in $\mathcal{Z}$ are disjoint, and then return $\bigcup_{Z \in \mathcal{Z}} Z$ as the hitting set.
It is a hitting set because of the maximality of $\mathcal{Z}$, and it is a $\gamma$-approximation solution since $|\bigcup_{Z \in \mathcal{Z}} Z| \leq \gamma |\mathcal{Z}|$ and any hitting set of $\{X \backslash U: X \in \mathcal{X}\}$ is of size at least $|\mathcal{Z}|$ as it has to contain one vertex in each $Z \in \mathcal{Z}$.)
The sub-routine $\textsc{Hit}(\{X \backslash U: X \in \mathcal{X}\})$ is just such an algorithm, which returns the size of a $\gamma$-approximation solution for the minimum hitting set of $\{X \backslash U: X \in \mathcal{X}\}$.
If the size returned by \textsc{Hit} is at most $\gamma k$ (and $G[U]$ is $\mathcal{F}$-free), we construct a new collection $\mathcal{X}_t$ (line~14-15).
If the size is greater than $\gamma k$, then $\{X \backslash U: X \in \mathcal{X}\}$ does not have a hitting set of size $k$ and we do not need to construct a new collection.

\begin{algorithm}[h]
    \caption{\textsc{Generate}$(G,\mathcal{F},k)$ \Comment{$\mathcal{F}$ consists of connected graphs}}
    \begin{algorithmic}[1]
        \State $t \leftarrow 0$
        \State $\textsc{Branch}(\emptyset,\emptyset,\theta_\mathsf{yes},\theta_\mathsf{no})$
        \State \textbf{return} $\mathcal{X}_1,\dots,\mathcal{X}_t$
        \smallskip
        \Procedure{\textsc{Branch}}{$U,\mathcal{X},d_\textsf{yes},d_\textsf{no}$}
            \If{$d_\textsf{yes} < 0$ or $d_\textsf{no} < 0$}{ \textbf{return}}
            \EndIf
            \If{there exists $U$-active $(X,F,f) \in \varGamma$ with $X \notin \mathcal{X}$}
                \State $\textsc{Branch}(U,\mathcal{X} \cup \{X\},d_\textsf{yes}-1,d_\textsf{no})$ \Comment{``yes'' branch}
                \State $\Delta \leftarrow \gamma^{|X|}\delta$
                \State $(H_1,\pi_1),\dots,(H_\Delta,\pi_\Delta) \leftarrow$ a witness of $(X,F,f)$
                \State $\mathcal{P} \leftarrow \textsc{Config}((X,F,f),(H_1,\pi_1),\dots,(H_\Delta,\pi_\Delta))$
                \For{every $P \in \mathcal{P}$}
                    \State $\textsc{Branch}(U \cup X \cup P,\mathcal{X},d_\textsf{yes},d_\textsf{no}-1)$ \Comment{``no'' branches}
                \EndFor
            \Else{ \textbf{if} $G[U]$ is $\mathcal{F}$-free and $\textsc{Hit}(\{X \backslash U: X \in \mathcal{X}\}) \leq \gamma k$ \textbf{then}}
                \State $t \leftarrow t+1$
                \State $\mathcal{X}_t = \{X \backslash U: X \in \mathcal{X}\}$
            \EndIf
        \EndProcedure
        \smallskip
        \Function{\textsc{Config}}{$(X,F,f),(H_1,\pi_1),\dots,(H_\Delta,\pi_\Delta)$}
            \State $C_1,\dots,C_t \leftarrow$ connected components of $F - f(X)$
            \State $V_{i,j} \leftarrow \pi_j^{-1}(V(C_i))$ for $i \in [t]$ and $j \in [\Delta]$        
            \State $\mathcal{J} = \{J \subseteq [\Delta]: |J| \leq \gamma\}$
            \State $\mathcal{P} \leftarrow \{\bigcup_{i=1}^t \bigcup_{j \in J_i} V_{i,j}: J_1,\dots,J_t \in \mathcal{J}\}$
            \State \textbf{return} $\mathcal{P}$
        \EndFunction
    \end{algorithmic}
    \label{alg-branch}
\end{algorithm}

Next, we discuss the branching part (line~7-12).
Recall that $(X,F,f) \in \varGamma$ is a $U$-active heavy core and $X \notin \mathcal{X}$.
In the ``yes'' branch (line~7), we simply add $X$ to $\mathcal{X}$ and recurse; the value of $d_\mathsf{yes}$ is increased by 1 since we made one more ``yes'' decision.
The ``no'' branches are more complicated (line~8-12).
As aforementioned, we need to guess a set $P$ of vertices, and add them to $U$ together with $X$.
To this end, we first compute a witness $(H_1,\pi_1),\dots,(H_\Delta,\pi_\delta)$ of $(X,F,f)$ where $\Delta = \gamma^{|X|} \delta$ (line~9).
Then in line~10, the sub-routine \textsc{Config} computes a collection $\mathcal{P}$ consisting of all guesses of $P$, based on the heavy core $(X,F,f)$ and its witness $(H_1,\pi_1),\dots,(H_\Delta,\pi_\delta)$.
Every $P \in \mathcal{P}$ corresponds to a ``no'' branch (line~11-12), in which we add the vertices in $X \cup P$ to $U$ and recurse; the value of $d_\mathsf{no}$ is increased by 1 since we made one more ``no'' decision.
The collection $\mathcal{P}$ is constructed as follows.
Suppose the graph $F-f(X)$ consists of connected components $C_1,\dots,C_t$ (line~17).
The graphs $H_1,\dots,H_\Delta$ are all isomorphic to $F$ with the isomorphisms $\pi_1,\dots,\pi_\Delta$.
Thus, for every $i \in [t]$ and $j \in [\Delta]$, $C_i$ is isomorphic to an induced subgraph of $H_j$, whose vertex set is $V_{i,j} = \pi_j^{-1}(V(C_i))$.
Ideally, we would like to guess for every $i \in [t]$ and $j \in [\Delta]$, whether $V_{i,j}$ is hit by the solution $S$ or not, and let $P$ be the union of the $V_{i,j}$'s not hit by $S$ in our guess.
However, this will result in $2^{t \Delta}$ guesses, which we cannot afford.
Instead, what we do is to guess for every $i \in [t]$, up to $\gamma$ sets $V_{i,j}$ not hit by $S$, and let $P$ be the union of these sets.
Formally, let $\mathcal{J} = \{J \subseteq [\Delta]: |J| \leq \gamma\}$.
Then each $P \in \mathcal{P}$ is defined by $t$ index sets $J_1,\dots,J_t \in \mathcal{J}$ as $P = \bigcup_{i=1}^t \bigcup_{j \in J_i} V_{i,j}$ (line~20).
Since $t \leq \gamma$ and $\Delta \leq \gamma^\gamma \delta = O(\delta)$, we have $|\mathcal{P}| = \delta^{O(1)}$, which bounds the number of ``no'' branches.

\subsection{Analysis} \label{sec-branchana}
In this section, we analyze the correctness of Algorithm~\ref{alg-branch}.
We shall show that the collections $\mathcal{X}_1,\dots,\mathcal{X}_t$ returned by Algorithm~\ref{alg-branch} satisfy the conditions in Theorem~\ref{thm-subtwbranch}, when the parameters $\delta,\theta_\mathsf{yes},\theta_\mathsf{no}$ are chosen properly.
How to implement Algorithm~\ref{alg-branch} in the desired running time will be discussed in the next section.

We first bounded $t$.
It suffices to bound the number of leaves in the branching tree (equivalently, the recursion tree of \textsc{Branch}).
Let $\Lambda(a,b)$ denote the maximum number of leaves in the recursion tree when calling \textsc{Branch} with $d_\mathsf{yes} = a$ and $d_\mathsf{yes} = b$.
Then we have the recurrence
\begin{equation*}
    \Lambda(a,b) \leq \Lambda(a-1,b) + \delta^{O(1)} \cdot \Lambda(a,b-1),
\end{equation*}
since the number of ``no'' branches is bounded by $\delta^{O(1)}$, as argued in the previous section.
The recurrence solves to $\Lambda(a,b) = \delta^{O(b)} \cdot \binom{a+b}{b}$.
Since $\binom{a+b}{b} = (a+b)^{O(b)}$ and $t \leq \Lambda(\theta_\mathsf{yes},\theta_\mathsf{no})$, we have $t = (\theta_\mathsf{yes} \delta+\theta_\mathsf{no} \delta)^{O(\theta_\mathsf{no})}$.
Therefore, as long as $\theta_\mathsf{yes},\delta$ are both polynomial in $k$ and $\theta_\mathsf{no}$ is sublinear in $k$, $t$ is subexponential in $k$.

Next, we bound the treewidth of the Gaifman graph of each $\mathcal{X}_i$.
Consider the moment we construct $\mathcal{X}_i$.
We observe that at this moment, for every $X \in \mathcal{X}$, there exists a heavy core $(X,F,f)$ which is $U$-minimal.
Indeed, $X$ is added to $\mathcal{X}$ when we consider a $U'$-active heavy core $(X,F,f) \in \varGamma$, where $U'$ denotes the set ``$U$'' at that point.
By definition, $(X,F,f)$ is $U'$-minimal.
We have $U' \subseteq U$, since the set ``$U$'' is expanding along a path in the branching tree.
Thus, $(X,F,f)$ is also $U$-minimal.
As the sets in $\mathcal{X}$ are defined by $U$-minimal heavy cores and $\{X \backslash U: X \in \mathcal{X}\}$ admits a hitting set of size at most $\gamma k$, we can apply Lemma~\ref{lem-twGaifman} to deduce that the Gaifman graph of $\mathcal{X}_i = \{X \backslash U: X \in \mathcal{X}\}$ has treewidth $(\eta \delta)^{O(1)} \cdot k^\rho$.

Finally, we show that a set $S \subseteq V(G)$ with $|S| \leq k$ is an $\mathcal{F}$-hitting set of $G$ iff it is a hitting set of $\mathcal{X}_i$ for some $i \in [t]$, if the parameters $\theta_\mathsf{yes}$ and $\theta_\mathsf{no}$ are chosen properly.
The ``if'' direction is straightforward, while showing the ``only if'' direction is much more difficult.
Suppose $S$ is a hitting set of $\mathcal{X}_i$.
When constructing the collection $\mathcal{X}_i$, we have $X \in \mathcal{X}$ for all $U$-active heavy core $(X,F,f) \in \varGamma$, and we set $\mathcal{X}_i = \{X \backslash U: X \in \mathcal{X}\}$.
Furthermore, $S \cap X\neq \emptyset$ for all $X \in \mathcal{X}$, because $S$ is a hitting set of $\mathcal{X}_i$.
Thus, by Lemma~\ref{lem-hitactive}, $S$ is an $\mathcal{F}$-hitting set of $G$.

The rest of this section is dedicated to showing the ``only if'' direction.
Suppose $S$ is an $\mathcal{F}$-hitting set of $G$ and $|S| \leq k$.
Consider an internal node of the branching tree, where we are branching on a heavy core $(X,F,f) \in \varGamma$ which is $U$-active and satisfies $X \notin \mathcal{X}$ for the current $U$ and $\mathcal{X}$.
At this node, we have one ``yes'' branch and multiple ``no'' branches.
Each ``no'' branch corresponds to a set $P \in \mathcal{P}$, where $\mathcal{P}$ is the collection computed in line~10 of Algorithm~\ref{alg-branch}.
We classify these branches into \textit{$S$-correct branches} and \textit{$S$-wrong branches} as follows.
If $S \cap X \neq \emptyset$, then the ``yes'' branch is the only $S$-correct branch and all ``no'' branches are $S$-wrong.
If $S \cap X = \emptyset$, then the ``yes'' branch is $S$-wrong, and a ``no'' branch is $S$-correct iff the corresponding set $P \in \mathcal{P}$ satisfies \textbf{(i)} $S \cap P = \emptyset$ and \textbf{(ii)} $S \cap P' \neq \emptyset$ for any $P' \in \mathcal{P}$ with $P \subsetneq P'$ (in other words, the set $P$ we guess is a maximal set in $\mathcal{P}$ that is disjoint from $S$).
Note that in either case, there exists at least one $S$-correct branch, and in the case $S \cap X = \emptyset$, there can possibly be multiple $S$-correct branches.
Now consider a path in the branching tree from the root.
We say the path is \textit{$S$-successful} if it chooses an $S$-correct branch in every step.
It is clear that at any node of an $S$-successful path, the set $U$ and the collection $\mathcal{X}$ always satisfy that $S \cap U = \emptyset$ and $S$ is a hitting set of $\mathcal{X}$.
\begin{observation} \label{obs-thetayes}
    Along an $S$-successful path in the branching tree, the total number of ``yes'' decisions made is at most $\theta_\mathsf{yes} = |\mathcal{F}| \cdot (\gamma^{2 \gamma+1} \delta)^\gamma \gamma! \cdot k$.
\end{observation}
\begin{proof}
When we reach the end of the $S$-successful path, the size of $\mathcal{X}$ is exactly equal to the number of ``yes'' decisions we made.
As argued before, for every $X \in \mathcal{X}$, there exists a $U$-minimal heavy core $(X,F,f) \in \varGamma$.
Since $S \cap U = \emptyset$, by Corollary~\ref{cor-hitfew}, each vertex in $S$ can hit at most $|\mathcal{F}| \cdot (\gamma^{2 \gamma+1} \delta)^\gamma \gamma!$ sets in $\mathcal{X}$.
But $S$ is a hitting set of $\mathcal{X}$.
So we must have $|\mathcal{X}| \leq |\mathcal{F}| \cdot (\gamma^{2 \gamma+1} \delta)^\gamma \gamma! \cdot k$.
\end{proof}

To bound the total number of ``no'' decisions on an $S$-successful path is more difficult.
Fix an ordering $\sigma$ of $V(G)$ such that $\text{wcol}_{2\gamma}(G,\sigma) = \text{wcol}_{2\gamma}(G)$.
The reason for why we choose $2\gamma$ as the distance instead of $\gamma$ will be clear later.
For $v \in V(G)$, define $\lambda_S(v) = |\{u \in S: v \in \text{WR}_\gamma(G,\sigma,u)\}|$.
Then we define a set
\begin{equation*}
    R = \{v \in V(G): \lambda_S(v) \geq \delta - \gamma - \text{wcol}_\gamma(G)\}.
\end{equation*}
Observe that $|R| \leq \frac{\text{wcol}_{2\gamma}(G)}{\delta - \gamma - \text{wcol}_{2\gamma}(G)} \cdot |S|$.
Indeed, since $|\text{WR}_\gamma(G,\sigma,u)| \leq \text{wcol}_{2\gamma}(G,\sigma) = \text{wcol}_{2\gamma}(G)$ for all $u \in S$, we have $\sum_{v \in V(G)} \lambda_S(v) = \sum_{u \in S} |\text{WR}_\gamma(G,\sigma,u)| \leq \text{wcol}_{2\gamma}(G) \cdot |S|$.
By an averaging argument, we deduce that $|R| \leq \frac{\text{wcol}_{2\gamma}(G)}{\delta - \gamma - \text{wcol}_{2\gamma}(G)} \cdot |S|$.

\begin{observation} \label{obs-xinR}
    Let $(X,F,f) \in \varGamma$ be a heavy core, and $x \in X$ be the largest vertex under $\sigma$.
    If $S \cap X = \emptyset$, then $x \in R$.
    Furthermore, if $\delta > \textnormal{wcol}_{2\gamma}(G)$, then $x \in \textnormal{WR}_\gamma(G,\sigma,u)$ for all $u \in X$.
\end{observation}
\begin{proof}
Let $(H_1,\pi_1),\dots,(H_\Delta,\pi_\Delta)$ be a witness of $(X,F,f)$ where $\Delta = \gamma^{|X|} \delta$.
So $V(H_1),\dots,V(H_\Delta)$ form a sunflower with core $X$.
Define $v_i$ as the largest vertex in $V(H_i)$ under $\sigma$, for $i \in [\Delta]$.
Observe that $v_i \in \text{WR}_\gamma(G,\sigma,x)$ for all $i \in [\Delta]$.
Indeed, $H_i$ is isomorphic to $F \in \mathcal{F}$ and by assumption all graphs in $\mathcal{F}$ are connected.
Since $x \in X \subseteq V(H_i)$ and $v_i \in V(H_i)$, there is a path in $H_i$ connecting $x$ and $v_i$ whose length is at most $\delta$.
The largest vertex on this path is $v_i$, and thus $v_i \in \text{WR}_\gamma(G,\sigma,x)$.
Set $I = \{i \in [\Delta]: v_i >_\sigma x\}$.
The following properties of $I$ hold.
\begin{enumerate}[(i)]
    \item $v_i = x$ for all $i \in [\Delta] \backslash I$.
    Indeed, since $x \in X \subseteq V(H_i)$, we have $v_i \geq_\sigma x$ for all $i \in [\Delta]$ and thus $v_i = x$ for all $i \in [\Delta] \backslash I$.
    \item $x \in \text{WR}_\gamma(G,\sigma,u)$ for all $u \in V(H_i)$ where $i \in [\Delta] \backslash I$.
    Indeed, as shown in property (i), for all $i \in [\Delta] \backslash I$, we have $v_i = x$ and thus $x$ is the largest vertex in $V(H_i)$.
    Since $H_i$ is connected, there exists a path in $H_i$ connecting $u$ and $x$ of length at most $\gamma$ on which $x$ is the largest vertex, which implies $x \in \text{WR}_\gamma(G,\sigma,u)$.
    \item $|I| \leq \text{wcol}_{2\gamma}(G)$.
    To see this, observe that $v_i \in V(H_i) \backslash X$ for all $i \in I$, since $x$ is the largest vertex in $X$ and $v_i >_\sigma x$.
    As $V(H_1) \backslash X,\dots,V(H_r) \backslash X$ are disjoint, the vertices $v_i$ for $i \in [r]$ are distinct, i.e., $|\{v_i: i \in I\}| = |I|$.
    Furthermore, because $v_1,\dots,v_\Delta \in \text{WR}_\gamma(G,\sigma,x)$, we have $|I| \leq |\text{WR}_\gamma(G,\sigma,x)| \leq \text{wcol}_{2\gamma}(G,\sigma) = \text{wcol}_{2\gamma}(G)$.
\end{enumerate}

Now we consider the set $S$ such that $S \cap X = \emptyset$.
Since $S$ is an $\mathcal{F}$-hitting set of $G$, $S \cap V(H_i) \neq \emptyset$ for all $i \in [\Delta]$.
Furthermore, as $S \cap X = \emptyset$, we have $S \cap (V(H_i) \setminus X) \neq \emptyset$ for all $i \in [\Delta]$.
Let $u_i \in S \cap (V(H_i) \setminus X)$.
Note that $u_1,\dots,u_\Delta$ are distinct because $V(H_1) \backslash X,\dots,V(H_r) \backslash X$ are disjoint.
By property (ii) above, $x \in \text{WR}_\gamma(G,\sigma,u_i)$ for all $i \in [\Delta] \backslash I$.
Therefore, $\lambda_S(x) \geq \Delta - |I|$.
Since $|I| \leq \text{wcol}_{2\gamma}(G)$ by property (iii) above,
\begin{equation*}
    \lambda_S(x) \geq \Delta - |I| \geq \delta-\gamma-\text{wcol}_{2\gamma}(G),
\end{equation*}
and hence we have $x \in R$.
To see the second statement, assume $\delta > \text{wcol}_{2\gamma}(G)$.
So we have $\Delta > \text{wcol}_{2\gamma}(G)$ and thus $[\Delta] \backslash I \neq \emptyset$ by property (iii).
Let $i \in [\Delta] \backslash I$.
By property (ii), $x \in \text{WR}_\gamma(G,\sigma,u)$ for all $u \in V(H_i)$ and in particular all $u \in X$.
\end{proof}

\begin{observation} \label{obs-manyavl}
    Let $(X,F,f) \in \varGamma$ be a heavy core, $z \in X$, and $(H_1,\pi_1),\dots,(H_\Delta,\pi_\Delta)$ be a witness of $(X,F,f)$ where $\Delta = \gamma^{|X|} \delta$.
    Also, let $C_z$ be the connected component of $F - f(\{x \in X: x >_\sigma z\})$ containing $f(z)$, and $C'$ be a connected component of $F - f(X)$ contained in $C_z$.
    If $z \notin R$, then there exists $J \subseteq [\Delta]$ with $|J| = \gamma$ such that $S \cap \pi_j^{-1}(V(C')) = \emptyset$ for all $j \in J$.
\end{observation}
\begin{proof}
Suppose $z \notin R$.
Consider the indices $i \in [\Delta]$ such that $(V(H_i) \backslash X) \cap \text{WR}_\gamma(G,\sigma,z) \neq \emptyset$.
As $V(H_1) \backslash X,\dots,V(H_\Delta) \backslash X$ are disjoint, the number of such vertices is at most $|\text{WR}_\gamma(G,\sigma,z)| \leq \text{wcol}_{2\gamma}(G,\sigma) = \text{wcol}_{2\gamma}(G)$.
So without loss of generality, we assume $(V(H_i) \backslash X) \cap \text{WR}_\gamma(G,\sigma,z) = \emptyset$ for all $i \in [\Delta - \text{wcol}_{2\gamma}(G)]$.
Let $I = \{i \in [\Delta - \text{wcol}_{2\gamma}(G)]: S \cap \pi_i^{-1}(V(C')) \neq \emptyset\}$.
It suffices to show that $|I| \leq \Delta - \text{wcol}_{2\gamma}(G) - \gamma$.
Indeed, if this is the case, we can set $J = [\Delta - \text{wcol}_{2\gamma}(G)] \backslash I$, which satisfies the desired properties.
Assume $I = [r]$ without loss of generality.
For each $i \in [r]$, we pick a vertex $v_i \in S \cap \pi_i^{-1}(V(C'))$.
Note that $\pi_1^{-1}(V(C')),\dots,\pi_r^{-1}(V(C'))$ are disjoint, since $\pi_1^{-1}(V(C')) \subseteq V(H_i) \backslash X$ for all $i \in [r]$.
So the vertices $v_1,\dots,v_r$ are distinct.
We claim that $z \in \text{WR}_\gamma(G,\sigma,v_i)$ for all $i \in [r]$.
Recall that $C_z$ is the connected component of $F - f(\{x \in X: x >_\sigma z\})$ containing $f(z)$, and $C'$ is contained in $C_z$.
So there exists a simple path $p$ in $C_z$ connecting $f(z)$ and $f_i(v_i)$.
In $H_i$, there exists a (unique) path $\tilde{p}$ connecting $z$ and $v_i$ such that $\pi_i$ induces an isomorphism between $\tilde{p}$ and $p$.
Let $z'$ be the largest vertex on $\tilde{p}$ under the ordering $\sigma$.
Note that $z' \in \text{WR}_\gamma(G,\sigma,z)$, as the largest vertex on the sub-path of $\tilde{p}$ between $z$ and $z'$ is $z'$ itself.
It follows that $z' \notin V(H_i) \backslash X$, because $(V(H_i) \backslash X) \cap \text{WR}_\gamma(G,\sigma,z) = \emptyset$ by assumption.
So $z' \in X$.
But $f(z') \in V(C_z)$ and $C_z$ is a connected component of $F - f(\{x \in X: x >_\sigma z\})$, which implies that $z' \leq_\sigma z$.
On the other hand, $z' \geq_\sigma z$ as $z$ is the largest vertex on $\tilde{p}$.
So we have $z' = z$, and the existence of the path $\tilde{p}$ implies that $z \in \text{WR}_\gamma(G,\sigma,v_i)$.
Now $z \in \text{WR}_\gamma(G,\sigma,v_i)$ for all $i \in [r]$, and $v_1,\dots,v_r$ are distinct.
This implies that $\lambda_S(z) \geq r$.
However, $z \notin R$ by our assumption, and thus $\lambda_S(z) \leq \delta - \gamma - \text{wcol}_{2\gamma}(G)$.
As a result, we have $r \leq \delta - \gamma - \text{wcol}_{2\gamma}(G) \leq \Delta - \gamma - \text{wcol}_{2\gamma}(G)$.
\end{proof}

Consider an $S$-successful path.
In order to bound the number of ``no'' decisions on the path, we charge each ``no'' decision to a $4$-tuple $(F,W,u,v)$ as follows, where $F \in \mathcal{F}$, $W \subseteq V(F)$, $u \in V(F)$, and $v \in R$.
Suppose the ``no'' decision is made for a heavy core $(X,F,f) \in \varGamma$, which is $U$-active for the set $U$ at that point.
By definition, $X \nsubseteq U$ (while after the ``no'' decision is made we will add the vertices in $X$ to $U$).
Also, $S \cap X = \emptyset$, since the path is $S$-successful.
Pick an arbitrary vertex $y \in X \backslash U$.
We define $v$ as the smallest vertex in $X \cap R \cap \text{WR}_\gamma(G,\sigma,y)$ under the ordering $\sigma$.
Note that such a vertex always exists if $\delta > \textnormal{wcol}_{2\gamma}(G)$, because the largest vertex $x \in X$ satisfies that $x \in R$ and $x \in \text{WR}_\gamma(G,\sigma,y)$, by Observation~\ref{obs-xinR}.
Set $W = f(X)$ and $u = f(y)$.
We then charge the ``no'' decision to the tuple $(F,W,u,v)$.
Let $\varPhi$ be the set of all such tuples.
Note that the total number of such tuples is bounded by
\begin{equation*}
    |\varPhi| \leq |\mathcal{F}| \cdot 2^\gamma \gamma \cdot |R| \leq \frac{|\mathcal{F}| \cdot 2^\gamma \gamma \cdot \text{wcol}_{2\gamma}(G)}{\delta - \gamma - \text{wcol}_{2\gamma}(G)} \cdot |S| \leq \frac{|\mathcal{F}| \cdot 2^\gamma \gamma \cdot \text{wcol}_{2\gamma}(G)}{\delta - \gamma - \text{wcol}_{2\gamma}(G)} \cdot k.
\end{equation*}
So it suffices to guarantee that each tuple does not get charged too many times.
\begin{observation} \label{obs-thetano}
    Every $4$-tuple $(F,W,u,v) \in \varPhi$ can get charged by at most $(\gamma^\gamma (\gamma+\textnormal{wcol}_{2\gamma}(G)))^\gamma \gamma!$ times by the ``no'' decisions along an $S$-successful path.
    In particular, the total number of ``no'' decisions made along an $S$-successful path is at most $\theta_\mathsf{no} = \frac{\gamma|\mathcal{F}| \cdot (2\gamma^\gamma (\gamma+\textnormal{wcol}_{2\gamma}(G)))^\gamma \gamma! \cdot \textnormal{wcol}_{2\gamma}(G)}{\delta - \gamma - \textnormal{wcol}_{2\gamma}(G)} \cdot k$.
\end{observation}
\begin{proof}
Assume $(F,W,u,v)$ gets charged $r > (\gamma^\gamma (\gamma+\textnormal{wcol}_{2\gamma}(G)))^\gamma \gamma!$ times, in order to deduce a contradiction.
Consider the ``no'' decisions of the $S$-successful path charged to $(F,W,u,v)$, and let $(X_1,F,f_1),\dots,(X_r,F,f_r) \in \varGamma$ be their corresponding heavy cores.
For each $i \in [r]$, there is a unique vertex $y_i \in X_i$ such that $f_i(y_i) = u$.
By our charging rule, we have $W = f_i(X_i)$ and $v$ is the smallest vertex in $X_i \cap R \cap \text{WR}_\gamma(G,\sigma,y_i)$ for all $i \in [r]$.
Furthermore, $y_i \in X_i \backslash U_i$ for all $i \in [r]$, where $U_i$ denotes the set $U$ when we branch on the heavy core $(X_i,F,f_i)$.
Set $p = \gamma+\textnormal{wcol}_{2\gamma}(G)$.
By Fact~\ref{fact-goodsunflower}, there exists $I \in [r]$ with $|I| = p$ such that $\{X_i: i \in I\}$ form a sunflower with core $K \subseteq V(G)$ and $(f_i)_{|K} = (f_j)_{|K}$ for all $i,j \in I$.
Without loss of generality, assume $I = [p]$.
Let $f:K \rightarrow V(F)$ be the map such that $f = (f_1)_{|K} = \cdots = (f_p)_{|K}$.
We have $v \in K$, since $v \in X_i$ for all $i \in [p]$.
Now consider the indices $i \in [p]$ such that $(X_i \backslash K) \cap \text{WR}_{2\gamma}(G,\sigma,v) = \emptyset$.
Observe that the number of such indices is at least $p - \textnormal{wcol}_{2\gamma}(G) = \gamma$.
Indeed, the sets $X_1 \backslash K,\dots,X_p \backslash K$ are disjoint, and thus there can be at most $|\text{WR}_{2\gamma}(G,\sigma,v)|$ indices $i \in [p]$ satisfying $(X_i \backslash K) \cap \text{WR}_{2\gamma}(G,\sigma,v) \neq \emptyset$.
Since 
\begin{equation*}
    |\text{WR}_{2\gamma}(G,\sigma,v)| \leq \text{wcol}_{2\gamma}(G,\sigma) = \text{wcol}_{2\gamma}(G),
\end{equation*}
the number of $i \in [p]$ satisfying $(X_i \backslash K) \cap \text{WR}_{2\gamma}(G,\sigma,v) = \emptyset$ is at least $p - \textnormal{wcol}_{2\gamma}(G) = \gamma$.
Without loss of generality, assume $(X_i \backslash K) \cap \text{WR}_{2\gamma}(G,\sigma,v) = \emptyset$ for all $i \in [\gamma]$.
Also, we can assume that the heavy cores $(X_1,F,f_1),\dots,(X_\gamma,F,f_\gamma)$ are sorted in the order along the $S$-successful path, i.e., we first made the ``no'' decision for $(X_1,F,f_1)$, followed by $(X_2,F,f_2)$, $(X_3,F,f_3)$, and so forth.
Under this assumption, we have the following conditions.
\begin{enumerate}[(i)]
    \item $U_1 \subseteq \cdots \subseteq U_\gamma$, because $U$ is expanding along any path in the branching tree.
    \item $X_i \subseteq U_{i+1}$ for all $i \in [\gamma-1]$, since after we make the ``no'' decision for $(X_i,F_i,f_i)$, the vertices in $X_i$ are all added to $U$ and thus appear in $U_{i+1}$.
    \item $(X_i,F,f_i)$ is $U_i$-active for all $i \in [\gamma]$, because we always branch on $U$-active heavy cores.
\end{enumerate}
Our plan is to show that that $(X_\gamma,F,f_\gamma)$ is $U_\gamma$-redundant and thus not $U_\gamma$-active, which contradicts condition (iii) above.
Consider an $F$-copy $(H,\pi)$ in $G$ such that $(X_\gamma,F,f_\gamma) \preceq (V(H),F,\pi)$.
We want to show that $(H,\pi)$ is $U_\gamma$-redundant, i.e., there exists another $F$-copy $(H',\pi')$ in $G$ such that $V(H') \backslash U_\gamma \subsetneq V(H) \backslash U_\gamma$.
Recall that $y_i \in X_i \backslash U_i$ is the (unique) vertex satisfying $f_i(y_i) = u$, and $y_i \in X_i \backslash U_i$.
It follows that $y_2 \notin K$, since $K \subseteq X_1 \subseteq U_2$ but $y_2 \notin U_2$.
Thus, $u \notin f_2(K) = f(K)$.
Let $C$ be the connected component of $F-f(K)$ that contains $u$.
\medskip

\noindent
\textit{Claim 1.}
For any $i \in [\gamma]$ and $c \in f_i^{-1}(V(C))$, we have $c \leq_\sigma v$.
\medskip

\noindent
\textit{Proof.}
Let $c \in f_i^{-1}(V(C))$ be the largest vertex under the ordering $\sigma$.
It suffices to show that $c \leq_\sigma v$.
Assume $c >_\sigma v$.
As $u \in V(C)$, there exists a simple path in $C$ connecting $u$ and $f_i(c)$, which corresponds to a simple path in $H_i[f_i^{-1}(V(C))]$ connecting $y_i$ and $c$.
Since the latter path is in $H_i[f_i^{-1}(V(C))]$, $a$ is the largest vertex on it, which implies $c \in \text{WR}_\gamma(G,\sigma,y_i)$.
On the other hand, $v \in \text{WR}_\gamma(G,\sigma,y_i)$ by construction.
It follows that $c \in \text{WR}_{2\gamma}(G,\sigma,v)$, because $c >_\sigma v$.
As $c \in f_i^{-1}(V(C)) \subseteq X_i \backslash K$, this contradicts our assumption that $(X_i \backslash K) \cap \text{WR}_{2\gamma}(G,\sigma,v) = \emptyset$.
\hfill $\lhd$
\medskip

Define $A = N_F[V(C)]$ and $B = V(F) \backslash V(C)$.
Then $A \backslash B = V(C)$ and $B \backslash A = V(F) \backslash A$.
So there is no edge of $G$ between $A \backslash B$ and $B \backslash A$.
Note that $|V(C)| \geq 1$ as $u \in V(C)$, and $|A| \geq |V(C)|+1 \geq 2$, because $F$ is connected and hence $N_F(V(C)) \neq \emptyset$.
So we have $|V(H) \backslash \pi^{-1}(A)| = |V(F) \backslash A| \leq \gamma-2$.
Since $X_1 \backslash K,\dots,X_{\gamma-1} \backslash K$ are disjoint, there exists an index $s \in [\gamma-1]$ such that $(X_s \backslash K) \cap (V(H) \backslash \pi^{-1}(A)) = \emptyset$.
\medskip

\noindent
\textit{Claim 2.}
$X_s \cap \pi^{-1}(B) \subseteq K$.
\medskip

\noindent
\textit{Proof.}
Since $(X_s \backslash K) \cap (V(H) \backslash \pi^{-1}(A)) = \emptyset$, we have $X_s \cap \pi^{-1}(B \backslash A) \subseteq K$.
On the other hand, 
\begin{equation*}
    A \cap B = A \backslash V(C) \subseteq f(K) = f_\gamma(K) = \pi(K).
\end{equation*}
which implies $X_s \cap \pi^{-1}(A \cap B) \subseteq \pi^{-1}(A \cap B) \subseteq K$.
Therefore, $X_s \cap \pi^{-1}(B) \subseteq K$.
\hfill $\lhd$
\medskip

Let $(H_1,\pi_1),\dots,(H_\Delta,\pi_\Delta)$ be a witness of $(X_s,F,f_s)$, where $\Delta = \gamma^{|X_s|} \delta$.
By definition, the sets $V(H_1),\dots,V(H_\Delta)$ form a sunflower with core $X_s$ and $f_s = (\pi_1)_{|X_s} = \cdots = (\pi_\Delta)_{|X_s}$.
Recall that $W = f_i(X_i)$ for all $i \in [\gamma]$ and in particular $W = f_s(X_s)$.
Consider the connected components of $F - W$.
Since $f(K) = f_s(K) \subseteq W$, every connected component of $F - W$ is either contained in $C$ or disjoint from $C$.
The connected components contained in $C$ have the following property.
\medskip

\noindent
\textit{Claim 3.}
For any connected component $C'$ of $F - W$ contained in $C$, there exists $J \subseteq [\Delta]$ with $|J| = \gamma$ such that $S \cap \pi_j^{-1}(V(C')) = \emptyset$ for all $j \in J$.
\medskip

\noindent
\textit{Proof.}
Let $z$ be the largest vertex in $f_s^{-1}(V(C))$ and $C_z$ be the connected component of $F - f_s(\{x \in X_s: x >_\sigma z\})$ that contains $f_s(z)$.
We observe that $C$ is contained in $C_z$.
Indeed, $f_s^{-1}(V(C)) \cap \{x \in X_s: x >_\sigma z\} = \emptyset$, which implies $V(C) \cap f_s(\{x \in X_s: x >_\sigma z\}) = \emptyset$ and thus $C$ is contained in some connected component of $F - f_s(\{x \in X_s: x >_\sigma z\})$.
As $f(z) \in V(C)$ and $f(z) \in V(C_z)$, we know that $C$ is contained in $C_z$.
Next, we show that $z \notin R$.
By Claim~1, $z \leq_\sigma v$.
Note that $z \neq v$, since $v \in K$ but $z \notin K$, where the latter follows from the fact that $V(C) \cap f(K) = \emptyset$ and thus $f_s^{-1}(V(C)) \cap K = \emptyset$.
Thus, $z <_\sigma v$.
Furthermore, we have $z \in \text{WR}_\gamma(G,\sigma,y_s)$.
To see this, recall that $f_s(y_s) = u \in V(C)$, and thus there exists a simple path $\pi$ in $C$ connecting $u$ and $z$ such that $z$ is the largest vertex in $f_s^{-1}(V(\pi))$.
The length of $\pi$ is at most $\gamma$, which implies $z \in \text{WR}_\gamma(G,\sigma,y_s)$ by Fact~\ref{fact-connect}.
Now $z \in X_s \cap \text{WR}_\gamma(G,\sigma,y_s)$.
Recall that $v$ is the smallest vertex in $X_s \cap R \cap \text{WR}_\gamma(G,\sigma,y_s)$ by our charging rule.
As $z <_\sigma v$, we must have $z \notin R$, for otherwise $z$ is a vertex in $X_s \cap R \cap \text{WR}_\gamma(G,\sigma,y_s)$ smaller than $v$.
Applying Observation~\ref{obs-manyavl}, we directly obtain the desired set $J \subseteq [\Delta]$.
\hfill $\lhd$
\medskip

Consider the ``no'' decision made for $(X_s,F,f_s)$.
We computed the collection $\mathcal{P}$ in line~10 of Algorithm~\ref{alg-branch}, and suppose we chose the ``no'' branch corresponding to $P \in \mathcal{P}$.
As we are on an $S$-successful path, we have $S \cap P = \emptyset$ and $S \cap P' \neq \emptyset$ for all $P' \in \mathcal{P}$ with $P \subsetneq P'$.
In other words, $P$ is a maximal set in $\mathcal{P}$ that is disjoint from $S$.
Recall the construction of $\mathcal{P}$.
Let $C_1,\dots,C_t$ be the connected components of $F - f_s(X_s) = F - W$, and $\mathcal{J} = \{J \subseteq [\Delta]: |J| \leq \gamma\}$.
Then $\mathcal{P} = \{\bigcup_{i=1}^t \bigcup_{j \in J_i} V_{i,j}: J_1,\dots,J_t \in \mathcal{J}\}$, where $V_{i,j} = \pi_j^{-1}(V(C_i))$.
Suppose $P = \bigcup_{i=1}^t \bigcup_{j \in J_i} V_{i,j}$ for $J_1,\dots,J_t \in \mathcal{J}$.
Without loss of generality, assume $C_1,\dots,C_{t'}$ are contained in $C$ and $C_{t'+1},\dots,C_t$ are disjoint from $C$.
By Claim~3, for each $i \in [t']$, there exist at least $\gamma$ indices $j \in [\Delta]$ such that $S \cap V_{i,j} = \emptyset$.
Because of the maximality of $P$, we must have $|J_i| = \gamma$ for all $i \in [t']$.
Indeed, if $|J_i| < \gamma$, then there exists $j \in [\Delta] \backslash J_i$ with $S \cap V_{i,j} = \emptyset$ (as $S$ is disjoint from at least $\gamma$ $V_{i,j}$'s).
In this case, the set $P' = P \cup V_{i,j}$ is also in $\mathcal{P}$ and disjoint from $S$, which contradicts the maximality of $P$.
Note that the sets $V_{i,j}$ are disjoint from each other for all $i \in [t]$ and $j \in [\Delta]$.
Since $|V(H) \backslash A| \leq \gamma - 2$ and $|J_i| = \gamma$ for each $i \in [t']$, there exists $j_i \in J_i$ such that $V_{i,j_i} \cap (V(H) \backslash A) = \emptyset$.
Set $X = f_s^{-1}(A \cap W) = f_s^{-1}(A)$.
Define $H' = (\bigcup_{i=1}^{t'} H_{j_i}[V_{i,j_i} \cup X]) \cup H[\pi^{-1}(B)]$.
Clearly, $V(H') = (\bigcup_{i=1}^{t'} V_{i,j_i}) \cup X \cup \pi^{-1}(B)$.
Now define a map $\pi':V(H') \rightarrow V(F)$ as
\begin{equation*}
    \pi'(a) = \left\{
    \begin{array}{ll}
        \pi_{j_i}(a) & \text{if } a \in V_{i,j_i}, \\
        f_s(a) & \text{if } a \in X, \\
        \pi(a) & \text{if } a \in \pi^{-1}(B).
    \end{array}
    \right.
\end{equation*}
Observe that $\pi'$ is well-defined.
Indeed, $V_{1,j_1},\dots,V_{t',j_{t'}}$ are pairwise disjoint and are all disjoint from $X$ and $\pi^{-1}(B)$.
Furthermore, $X \cap \pi^{-1}(B) \subseteq X_s \cap \pi^{-1}(B) \subseteq K$ by Claim~2 and $(f_s)_{|K} = f_{|K} = \pi_{|K}$.
Thus, $\pi'$ is well-defined.

Next, we show that $\pi'$ is an isomorphism between $H'$ and $F$.
The set $V(H')$ is the disjoint union of $V_{1,j_1},\dots,V_{t',j_{t'}},X,\pi^{-1}(B) \backslash X$.
On the other hand, $V(F)$ is the disjoint union of $V(C_1),\dots,V(C_{t'}),f_s(X),B \backslash f_s(X)$ by our construction.
The map $\pi'$ is a bijection between $V_{i,j_i}$ and $V(C_i)$ when restricted to $V_{i,j_i}$ for all $i \in [t']$, a bijection between $X$ and $f_s(X)$ when restricted to $X$, and a bijection between $\pi^{-1}(B) \backslash X$ and $B \backslash f_s(X)$.
Therefore, $\pi'$ is bijective.
To see it is an isomorphism, observe that $F[A] = \bigcup_{i=1}^{t'} F[V(C_i) \cup f_s(X)]$, because every edge of $F[A]$ is in some $C_i$, or with both endpoints in $f_s(X) = A \cap W$, or with one endpoint in some $C_i$ and the other point in $f_s(B)$.
It follows that $F = F[A] \cup F[B] = (\bigcup_{i=1}^{t'} F[V(C_i) \cup f_s(X)]) \cup F[B]$, as there is no edge in $F$ between $A \backslash B$ and $B \backslash A$.
For each $i \in [t']$, $\pi'_{|V_{i,j_i} \cup X}$ is an isomorphism between $H_{j_i}[V_{i,j_i} \cup X]$ and $F[V(C_i) \cup f_s(X)]$, since $\phi_{|V_{i,j_i} \cup X} = (\pi_{j_i})_{|V_{i,j_i} \cup X}$.
Also, $\pi'_{|\pi^{-1}(B)}$ is an isomorphism between $H[\pi^{-1}(B)]$ and $F[B]$, because $\pi'_{|\pi^{-1}(B)} = \pi_{|\pi^{-1}(B)}$.
So by Lemma~\ref{lem-isom}, $\pi'$ is an isomorphism between $H'$ and $F$, which implies that $(H',\pi')$ is an $F$-copy in $G$.

Finally, we show that $V(H') \backslash U_\gamma \subsetneq V(H) \backslash U_\gamma$, which implies $(H,\pi)$ is $U_\gamma$-redundant.
We have $V_{i,j_i} \subseteq P$ for all $i \in [t']$, since $j_i \in J_i$.
The vertices in $P$ are added to the set $U$ after we make the ``no'' decision for $(X_s,F,f_s)$.
Thus, $\bigcup_{i=1}^{t'} V_{i,j_i} \subseteq U_{s+1} \subseteq U_\gamma$.
Also, the vertices in $X_s$ is are added to the set $U$ after the ``no'' decision.
As $X \subseteq X_s$, we have $X \subseteq U_\gamma$ and thus $(\bigcup_{i=1}^{t'} V_{i,j_i}) \cup X \subseteq U_\gamma$.
It follows that $V(H') \backslash U_\gamma \subseteq \pi^{-1}(B) \subseteq V(H)$ and hence $V(H') \backslash U_\gamma \subseteq V(H) \backslash U_\gamma$.
To further show that $V(H') \backslash U_\gamma \subsetneq V(H) \backslash U_\gamma$, recall that $y_\gamma \in X_\gamma \backslash U_\gamma \subseteq V(H) \backslash U_\gamma$.
However, $y_\gamma \notin V(H') \backslash U_\gamma$, because $\pi(y_\gamma) = u \notin B$ but $V(H') \backslash U_\gamma \subseteq \pi^{-1}(B)$.
Thus, $V(H') \backslash U_\gamma \subsetneq V(H) \backslash U_\gamma$ and $(H,\pi)$ is $U_\gamma$-redundant.
As a result, $(X_\gamma,F,f_\gamma)$ is $U_\gamma$-redundant, contradicting the $U_\gamma$-activeness of $(X_\gamma,F,f_\gamma)$.
\end{proof}

Now we can complete the proof of the correctness of Algorithm~\ref{alg-branch}.
We use the values in Observations~\ref{obs-thetayes} and~\ref{obs-thetano} as our parameters $\theta_\mathsf{yes}$ and $\theta_\mathsf{no}$, respectively.
Recall that our goal is to show any $\mathcal{F}$-hitting set $S \subseteq V(G)$ of $G$ with $|S| \leq k$ is a hitting set of $\mathcal{X}_i$ for some $i \in [t]$.
Consider an $S$-successful path in the branching tree from the root to a leaf.
Note that such an $S$-successful path always exists because any internal node in the branching tree has an $S$-correct branch.
When we reach the end of the path (i.e., the leaf), the number of ``yes'' decisions (``no'' decisions) made is at most $\theta_\mathsf{yes}$ (resp., $\theta_\mathsf{no}$).
Thus, we do not return in line~5 of Algorithm~\ref{alg-branch}.
So the only reason for why this is a leaf of the branching tree is that there does not exist any $U$-active heavy core $(X,F,f) \in \varGamma$ with $X \notin \mathcal{X}$.
As such, the algorithm proceeds to line~13 to check whether a new collection $\mathcal{X}_i$ should be constructed.
Since we are on an $S$-successful path, $S \cap U = \emptyset$ and $S$ is a hitting set of $\mathcal{X}$ for the current $U$ and $\mathcal{X}$.
The former implies that $G[U]$ is $\mathcal{F}$-free, as $S$ is an $\mathcal{F}$-hitting set of $G$, while the latter implies that $\textsc{Hit}(\{X \backslash U: X \in \mathcal{X}\}) \leq \gamma k$, because $S$ is a hitting set of $\{X \backslash U: X \in \mathcal{X}\}$ of size $k$ and \textsc{Hit} is a $\gamma$-approximation algorithm.
So the conditions in line~13 hold and a new collection $\mathcal{X}_i = \{X \backslash U: X \in \mathcal{X}\}$ is constructed in line~15, of which $S$ is a hitting set.

At the end, we recall the bound $t = (\theta_\mathsf{yes} \delta+\theta_\mathsf{no} \delta)^{O(\theta_\mathsf{no})}$ and the treewidth bound $(\eta\delta)^{O(1)} \cdot k^\rho$ of the Gaifman graphs of $\mathcal{X}_1,\dots,\mathcal{X}_t$ obtained at the beginning of this section.
By Observation~\ref{obs-thetayes}, we have $\theta_\mathsf{yes} = \delta^{O(1)} \cdot k$.
Furthermore, we have $\theta_\mathsf{no} = \eta^{O(1)} \cdot (k/\delta)$ by Observation~\ref{obs-thetano}, because $\text{wcol}_{2\gamma}(G) = \eta^{O(1)}$ by Lemma~\ref{lem-nablawcol}.
Therefore, $t = (\eta \delta)^{\eta^{O(1)} \cdot (k/\delta)}$.
Also, the size of each $\mathcal{X}_i$ is bounded by the depth of the branching tree, which is bounded by $\theta_\mathsf{yes} \delta+\theta_\mathsf{no} = (\eta\delta)^{O(1)} k$.
We remark that the $\eta$-factors do not matter in the special case $\mu = 0$ we are considering.
However, they will be important when we extend our proof to the general case in Section~\ref{sec-generalbranch}.

Finally, we analyze the time complexity of the branching algorithm.
Observe that the internal nodes of the branching tree have degree at least 3.
So the size of the branching tree is bounded by the number of its leaves, which is $(\theta_\mathsf{yes} \delta+\theta_\mathsf{no} \delta)^{O(\theta_\mathsf{no})} = (\eta \delta)^{\eta^{O(1)} \cdot (k/\delta)}$.
Therefore, it suffices to bounded the time cost at each node of the branching tree.
In other words, we need to show how to efficiently implement each call of \textsc{Branch} (without the recursive calls).
The most time-consuming step is to find a $U$-active heavy core (line~6) as well as its witness (line~9) or decide such a heavy core does not exist.
All the other steps can be done in $n^{O(1)}$ time.
To find a $U$-active heavy core, we first compute the $F$-copies in $G$ for all $F \in \mathcal{F}$.
This takes $n^{O(1)}$ time.
Then we compute all heavy cores in $G$.
This can be done by considering every triple $(X,F,f)$ where $X \subseteq V(G)$ is a subset of size at most $\gamma$, $F \in \mathcal{F}$, and $f: X \rightarrow V(F)$ is an injective map.
The number of such triples is $n^{O(1)}$.
For each such triple $(X,F,f)$, we can test if it is a heavy core in $n^{O(\delta)}$ time by guessing its witness, which consists of $O(\delta)$ $F$-copies in $G$.
Note that the number of heavy cores in $G$ is $n^{O(1)}$.
With the heavy cores (and the $F$-copies) in hand, it is straightforward to find the $U$-active ones in $n^{O(1)}$ time.
Thus, line~6 and line~9 can be done in $n^{O(\delta)}$ time.
The overall time complexity of our algorithm is then $(\eta \delta)^{\eta^{O(1)} \cdot (k/\delta)} \cdot n^{O(\delta)}$.
The choice of $\delta$ will be discussed in the next section.

\subsection{The general case} \label{sec-generalbranch}

In the previous sections, we have presented the branching algorithm for the special case where $\mu = 0$ and $\mathcal{F}$ consists of connected graphs.
In this section, we shall extend the result to the general case and complete the proof of Theorem~\ref{thm-subtwbranch}.

We first show how to handle a general $\mathcal{F}$ (which can contain any graphs).
For a graph $G$, we denote by $G^+$ the graph obtained from $G$ by adding a single vertex $o$ with edges $(o,v)$ for all $v \in V(G)$.
For a set $\mathcal{G}$ of graphs, we define $\mathcal{G}^+ = \{G^+: G \in \mathcal{G}\}$.
\begin{fact} \label{fact-graph+}
    Let $\mathcal{F}$ be a finite set of graphs.
    For any graph $G$, a subset $S \subseteq V(G)$ is an $\mathcal{F}$-hitting set of $G$ if and only if $S$ is an $\mathcal{F}^+$-hitting set of $G^+$.
\end{fact}
\begin{proof}
Suppose $V(G^+) \backslash V(G) = \{o\}$. 
Fix $F \in \mathcal{F}$, and we want to show $G-S$ contains a subgraph isomorphic to $F$ if and only if $G^+ - S$ contains a subgraph isomorphic to $F^+$.
Clearly, if $G-S$ contains a subgraph $H$ isomorphic to $F$, then $G^+ - S$ contains a subgraph $H'$ isomorphic to $F^+$, which is obtained by taking $H$ and the vertex $o$, together with the edges connecting $o$ and the vertices of $H$.
To see the other direction, assume $G^+ - S$ contains a subgraph $H'$ isomorphic to $F^+$.
If $V(H') \subseteq V(G)$, then we are done, because $H'$ contains a subgraph isomorphic to $F$ and so does $G$.
Assume $V(H') \subsetneq V(G)$, then $o \in V(H')$.
We show that $H' - \{o\}$ contains a subgraph $H$ isomorphic to $F$.
Let $\pi': V(H') \rightarrow V(F^+)$ be an isomorphism between $H'$ and $F^+$.
Define a map $\pi: V(H') \backslash \{o\} \rightarrow V(F)$ as
\begin{equation*}
    \pi(v) = \left\{
    \begin{array}{ll}
        \pi'(v) & \text{if } \pi'(v) \in V(F), \\
        \pi'(o) & \text{if } \pi'(v) \notin V(F).
    \end{array}
    \right.
\end{equation*}
Note that $\pi$ is well-defined.
Indeed, if $\pi'(v) \notin V(F)$, then $\pi'(v)$ is the single vertex $o' \in V(F^+) \backslash V(F)$ and we must have $\pi'(o) \neq o'$ since $\pi'$ is bijective, which implies $\pi'(o) \in V(F)$.
Furthermore, $\pi$ is bijective.
To see this, consider a vertex $u \in V(F)$.
If $u = \pi'(v)$ for some $v \in V(H') \backslash \{o\}$, then $u = \pi(v)$.
If $u = \pi'(o)$, then $u = \pi(v)$ for the unique vertex $v \in V(H') \backslash \{o\}$ satisfying $\pi'(v) = o'$.
In either case, $u$ is in the image of $\pi$.
Thus, $\pi$ is surjective, and thus also injective because $|V(H') \backslash \{o\}| = |V(F)|$.
Finally, observe that for any $(u,v) \in E(H' - \{o\})$, we have $(\pi(u),\pi(v)) \in E(F)$.
Indeed, if $\pi'(u),\pi'(v) \in V(F)$, then $(\pi(u),\pi(v)) = (\pi'(u),\pi'(v)) \in E(F)$.
Otherwise, assume $\pi'(u) \in \pi'(v) \notin V(F)$ without loss of generality.
Then $(\pi(u),\pi(v)) = (\pi'(o),\pi'(v)) \in E(F)$, since $\pi'(o)$ is neighboring to all vertices in $V(F)$ by the fact that $o$ is neighboring to all vertices in $V(H'-\{o\})$.
This property implies that $F$ is isomorphic to a subgraph of $H' - \{o\}$.
Therefore, $G-S$ contains a subgraph isomorphic to $F$.
\end{proof}

Now consider a graph class $\mathcal{G} \subseteq \mathcal{G}(\eta,0,\rho)$ and a finite set $\mathcal{F}$ of graphs (which are not necessarily connected).
Note that $\mathcal{G}^+ \subseteq \mathcal{G}(\eta+1,0,\rho)$ and all graphs in $\mathcal{F}^+$ are connected.
Therefore, we know that our branching algorithm can be applied to $\mathcal{G}^+$ and $\mathcal{F}^+$ (based on what we have shown in the previous sections).
For any graph $G \in V(G)$ of $n$ vertices, one can construct in $(\eta\delta)^{\eta^{O(1)} \cdot (k/\delta + \delta^3)} \cdot n^2 + O(n^4)$ time $t = (\eta\delta)^{\eta^{O(1)} \cdot (k/\delta)}$ collections $\mathcal{Y}_1,\dots,\mathcal{Y}_t \subseteq V(G^+)$ satisfying two conditions: \textbf{(i)} for any $S \subseteq V(G^+)$ with $|S| \leq k$, $S$ is an $\mathcal{F}^+$-hitting set of $G^+$ iff $S$ is a hitting set of $\mathcal{Y}_i$ for some $i \in [t]$, and \textbf{(ii)} the Gaifman graph of each $\mathcal{Y}_i$ has treewidth $(\eta\delta)^{O(1)} \cdot k^\rho$.
Now define $\mathcal{X}_i = \{Y \backslash \{o\}: Y \in \mathcal{Y}_i\}$ for $i \in [t]$, where $o$ is the single vertex in $V(G^+) \backslash V(G)$.
Clearly, the Gaifman graph of each $\mathcal{X}_i$ also has treewidth $(\eta\delta)^{O(1)} \cdot k^\rho$.
We claim that for any $S \subseteq V(G)$ with $|S| \leq k$, $S$ is an $\mathcal{F}$-hitting set of $G$ iff $S$ is a hitting set of $\mathcal{X}_i$ for some $i \in [t]$.
If $S$ is an $\mathcal{F}$-hitting set of $G$, then it is an $\mathcal{F}^+$-hitting set of $G^+$ by Fact~\ref{fact-graph+}.
By property (i) of the collections $\mathcal{Y}_1,\dots,\mathcal{Y}_t$, $S$ is a hitting set of $\mathcal{Y}_i$ for some $i \in [t]$.
Since $o \notin S$, $S$ is also a hitting set of $\mathcal{X}_i$.
On the other hand, if $S$ is a hitting set of $\mathcal{X}_i$ for some $i \in [t]$, then it is a hitting set of $\mathcal{Y}_i$.
Again by property (i) of $\mathcal{Y}_1,\dots,\mathcal{Y}_t$, $S$ is an $\mathcal{F}^+$-hitting set of $G^+$.
Fact~\ref{fact-graph+} then implies that $S$ is an $\mathcal{F}$-hitting set of $G$.
In other words, $\mathcal{X}_1,\dots,\mathcal{X}_t$ satisfy the desired conditions in Theorem~\ref{thm-subtwbranch}, and thus the theorem also holds for $\mathcal{G}$ and $\mathcal{F}$.
It follows that our branching algorithm applies to a general $\mathcal{F}$.
For convenience, we summarize what we have shown so far in the lemma below.
\begin{lemma} \label{lem-mu=0}
Let $\mathcal{G} \subseteq \mathcal{G}(\eta,0,\rho)$ where $\rho < 1$.
Also, let $\mathcal{F}$ be a finite set of graphs.
Then for a graph $G \in \mathcal{G}$ of $n$ vertices and parameters $k,\delta \in \mathbb{N}$, one can construct in $(\eta\delta)^{\eta^{O(1)} \cdot (k/\delta)} \cdot n^{O(\delta)}$ time $t = (\eta\delta)^{\eta^{O(1)} \cdot (k/\delta)}$ collections $\mathcal{X}_1,\dots,\mathcal{X}_t$ of subsets of $V(G)$ satisfying the following conditions.
\begin{itemize}
    \item For any $S \subseteq V(G)$ with $|S| \leq k$, $S$ is an $\mathcal{F}$-hitting set of $G$ iff $S$ hits $\mathcal{X}_i$ for some $i \in [t]$.
    \item The Gaifman graph of $\mathcal{X}_i$ has treewidth $(\eta\delta)^{O(1)} \cdot k^\rho$ for all $i \in [t]$.
    \item $|\mathcal{X}_i| = (\eta\delta)^{O(1)} \cdot k$ for all $i \in [t]$.
\end{itemize}
\end{lemma}
\vspace{0.1cm}

Next, we show how to extend our result to a general $\mathcal{G} \subseteq \mathcal{G}(\eta,\mu,\rho)$.
The basic idea is to keep finding large cliques in the input graph and branching on those cliques, until there is no large clique in the graph (at that point, the graph admits small separators).
This approach is standard, and similar ideas have been repeatedly used in the literature \cite{ErlebachJS05,FominLS18,lokshtanov2022subexponential,lokshtanov2023framework,van2006better}.
\begin{lemma} \label{lem-cliquebranch}
Let $\mathcal{G} \subseteq \mathcal{G}(\eta,\mu,\rho)$ where $\eta,\mu \geq 0$ and $\rho < 1$.
Also, let $\mathcal{F}$ be a finite set of graphs and $c \in (0,1)$.
Then for a graph $G \in \mathcal{G}$ of $n$ vertices and parameters $k,\tau \in \mathbb{N}$, one can construct in $n^{O(\tau+k/\tau)}$ time $r = \tau^{O(k/\tau)}$ subsets $S_1,\dots,S_r \subseteq V(G)$ satisfying the following conditions.
\begin{itemize}
    \item For any $\mathcal{F}$-hitting set $S \subseteq V(G)$ of $G$ with $|S| \leq k$, $S_i \subseteq S$ for some $i \in [r]$.
    \item $\omega(G-S_i) \leq \tau + \gamma$ for all $i \in [r]$.
    \item $|S_i| \leq k$ for all $i \in [r]$.
\end{itemize}
\end{lemma}
\vspace{0.1cm}
\begin{proof}
The branching algorithm for computing the sets $S_1,\dots,S_r \subseteq V(G)$ is presented in Algorithm~\ref{alg-clique}.
The core of the algorithm is the sub-routine \textsc{Branch}, which maintains a set $V$ of vertices of $G$.
In the initial call of \textsc{Branch} (line~2), we have $V = \emptyset$.
When doing $\textsc{Branch}(V)$, we first check whether $\omega(G - V) \leq \tau + \gamma$ (line~5).
If this is the case, we construct a new set $S_r = V$ (line~7) and return to the last level.
Otherwise, we check whether $|V| \geq k$, and if so, we do not branch further and return to the last level (line~9).
In line~10, the sub-routine $\textsc{Find}(G-V,K_{\tau+\gamma})$ returns a $K_{\tau+ \gamma}$-copy $(H,\pi)$ in $G-V$, where $K_{\tau+ \gamma}$ is the complete graph of $\tau+ \gamma$ vertices.
Then we branch on $H$ as follows.
We enumerate all subsets of $V' \subseteq V(H)$ whose size is at least $\tau$ (line~11).
For each such $V'$, we recursively call $\textsc{Branch}(V \cup V')$.

\begin{algorithm}[h]
    \caption{\textsc{CliqueReduce}$(G,k)$}
    \begin{algorithmic}[1]
        \State $r \leftarrow 0$
        \State $\textsc{Branch}(\emptyset)$
        \State \textbf{return} $S_1,\dots,S_r$
        \smallskip
        \Procedure{Branch}{$V$}
            \If{$\omega(G - V) \leq \tau + \gamma$}
                \State $r \leftarrow r+1$
                \State $S_r \leftarrow V$
                \State \textbf{return}
            \EndIf
            \vspace{0.08cm}
            \If{$|V| \geq k$}{ \textbf{return}}
            \EndIf
            \State $(H,\pi) \leftarrow \textsc{Find}(G-V,K_{\tau+\gamma})$
            \For{every $V' \subseteq V(H)$ with $|V'| > \tau$}
                \State $\textsc{Branch}(V \cup V')$
            \EndFor
        \EndProcedure
    \end{algorithmic}
    \label{alg-clique}
\end{algorithm}

In what follows, we first analyze the time complexity of Algorithm~\ref{alg-clique} and prove its correctness.
Consider the branching tree of the algorithm.
The graph $H$ computed in line~10 has $\tau+ \gamma$ vertices and thus the number of subsets of $V(H)$ with size larger than $\tau$ is $O(\tau^\gamma)$, which also bounds the degree of the branching tree. 
The depth of the branching tree is $O(k/\tau)$.
Indeed, the size of $V$ increases by at least $\tau$ after each branching step.
When $|V| \geq k$, we stop going further.
So the depth of the branching tree is at most $k/\tau$.
It follows that the size of the branching tree is $\tau^{O(k/\tau)}$, and hence $r = \tau^{O(k/\tau)}$.
The most time-consuming steps in a call of \textsc{Branch} are line~5 and line~10, where we need to check whether $G-V$ contains a clique of size $\tau + \gamma$, i.e., a $K_{\tau+ \gamma}$-copy, and if so, find such a copy.
Both steps can be done in $n^{O(\tau+\gamma)}$ time.
Therefore, the overall time complexity of Algorithm~\ref{alg-clique} is $n^{O(\tau+k/\tau)}$.

Now we show that the sets $S_1,\dots,S_r$ constructed satisfy the desired properties.
We have $|S_i| \leq k$ for all $i \in [r]$ just be construction.
It is also clear that $\omega(G-S_i) \leq \tau+ \gamma$ for all $i \in [r]$, as we construct a new set $S_r = V$ only when $\omega(G - V) \leq \tau+ \gamma$ (line~5).
To see the first property, consider an $\mathcal{F}$-hitting set $S \subseteq V(G)$ of $G$ with $|S| \leq k$.
Consider an internal node of the branching tree, where we are branching on a $K_{\tau+ \gamma}$-copy $(H,\pi)$.
At this node, we have $O(\tau^\gamma)$ branches, each of which corresponds to a subset $V' \subseteq V(H)$ with $|V'|>\tau$.
We say a branch is \textit{$S$-correct} if the corresponding set $V'$ satisfies $V' \subseteq S$.
We claim that there exists at least one branch that is $S$-correct.
Indeed, we must have $S \cap V(H) > \tau$, for otherwise $G-S$ contains the clique $H-S$ of size at least $\gamma$ and hence is not $\mathcal{F}$-free, which contradicts the fact that $S$ is an $\mathcal{F}$-hitting set of $G$.
Therefore, there exists $V' \subseteq V(H)$ satisfying $|V'|>\tau$ and $V' \subseteq S$, which implies that at least one branch is $S$-correct.
We say a path in the branching tree from the root is \textit{$S$-successful} if it chooses an $S$-correct branch in every step.
It is clear that at any node on an $S$-successful path, the set $V$ maintained by \textsc{Branch} satisfies $V \subseteq S$.
Consider an $S$-successful path in the branching tree from the root to a leaf.
Note that such an $S$-successful path exists because every internal node has at least one $S$-correct branch.
At the end of the path, we have $V \subseteq S$.
If $V = S$, then $\omega(G - V) = \omega(G - S) < \gamma$, since $S$ is an $\mathcal{F}$-hitting set of $G$.
Therefore, a new set $S_i = V$ is constructed in line~7 and $S_i = V \subseteq S$.
If $V \subsetneq S$, then $|V| < k$.
In this case, we must have $\omega(G - V) < \tau+ \gamma$ and a new set $S_i = V$ is constructed, which satisfies $S_i \subseteq S$.
Indeed, if this was not the case, then we shall not return to the last level in line~9 (as $|V| < k$) and thus will keep branching on a new $K_{\tau+ \gamma}$-copy, which contradicts the fact that we are the end of the path (i.e., at a leaf node).
As a result, $S_1,\dots,S_r$ satisfy both of the desired properties.
\end{proof}

Now we apply the above lemma to complete the proof of Theorem~\ref{thm-subtwbranch}.
Without loss of generality, we can assume that $\mathcal{G} = \mathcal{G}(\eta,\mu,\rho)$.
Given a graph $G \in \mathcal{G}$ and a parameter $k \in \mathbb{N}$, we apply the algorithm in the above lemma (with a parameter $\tau$ to be determined later) to generate sets $S_1,\dots,S_r \subseteq V(G)$ for $r = \tau^{O(k/\tau)}$.
Let $G_i = G - S_i$ for $i \in [r]$.
Note that $G_1,\dots,G_r \in \mathcal{G}$.
By Lemma~\ref{lem-cliquebranch}, $\omega(G_i) \leq \tau+\gamma$ for all $i \in [r]$.
Therefore, $G_1,\dots,G_r$ admit $(\eta (\tau+\gamma)^\mu,0,\rho)$-separators and thus $G_1,\dots,G_r \in \mathcal{G}' = \mathcal{G}(\eta (\tau+\gamma)^\mu,0,\rho)$.
Set $\tilde{\eta} = \eta (\tau+\gamma)^\mu$.
We now apply Lemma~\ref{lem-mu=0} on $\mathcal{G}'$ to compute, for each $G_i$, $t_i = (\tilde{\eta}\delta)^{\tilde{\eta}^{O(1)} \cdot (k/\delta)}$ collections $\mathcal{X}_1^{(i)},\dots,\mathcal{X}_{t_i}^{(i)}$ of subsets of $V(G_i)$.
Each $\mathcal{X}_j^{(i)}$ for $j \in [t_i]$ has size $(\tilde{\eta}\delta)^{O(1)} \cdot k = (\eta\tau\delta)^{O(1)} \cdot k$.
The time cost for each $G_i$ is $(\tilde{\eta}\delta)^{\tilde{\eta}^{O(1)} \cdot (k/\delta)} \cdot n^{O(\delta)}$.
Applying Lemma~\ref{lem-cliquebranch} takes $n^{O(\tau + k/\tau)}$ time.
Thus, the total time cost is
\begin{equation*}
    (\tilde{\eta}\delta)^{\tilde{\eta}^{O(1)} \cdot (k/\delta)} \cdot rn^{O(\delta)} + n^{O(\tau + k/\tau)} = \tau^{\tau^{O(1)} \cdot (k/\delta) + O(k/\tau)} \cdot n^{O(\delta)} + n^{O(\tau+k/\tau)},
\end{equation*}
which is in turn bounded by $n^{\tau^{O(1)} \cdot (k/\delta) + O(\delta+k/\tau)}$.
Let $L = \{(i,j): i \in [r] \text{ and } j \in [t_i]\}$, $t = \sum_{i=1}^r t_i$, and $f:[t] \rightarrow L$ be an arbitrary bijection.
Note that $t = (\tilde{\eta}\delta)^{\tilde{\eta}^{O(1)} \cdot (k/\delta)} \cdot \tau^{O(k/\tau)} = \tau^{\tau^{O(1)} \cdot (k/\delta)+O(k/\tau)}$.
For each $i \in [t]$, we write $\alpha(i)$ as the first index in the pair $f(i)$ and $\beta(i)$ as the second index in the pair $f(i)$.
Then for each $i \in [t]$, we define a collection $\mathcal{X}_i = \{\{v\}: v \in S_{\alpha(i)}\} \cup \mathcal{X}_{\beta(i)}^{(\alpha(i))}$ of subsets of $V(G)$.
The following observation shows that $\mathcal{X}_1,\dots,\mathcal{X}_t$ satisfy the properties we want.
\begin{observation}
    For all $i \in [t]$, we have $|\mathcal{X}_i| = (\eta\tau\delta)^{O(1)} \cdot k$ and the Gaifman graph of $\mathcal{X}_i$ has treewidth $(\tau \delta)^{O(1)} \cdot k^{\rho}$.
    Furthermore, for any $S \subseteq V(G)$ with $|S| \leq k$, $S$ is an $\mathcal{F}$-hitting set of $G$ iff $S$ hits $\mathcal{X}_i$ for some $i \in [t]$.
\end{observation}
\begin{proof}
The bound for the sizes of $\mathcal{X}_1,\dots,\mathcal{X}_t$ follows directly from the fact that $|S_i| \leq k$ for all $i \in [r]$ and $|\mathcal{X}_j^{(i)}| = (\eta\tau\delta)^{O(1)} \cdot k$ for all $i \in [r]$ and $j \in [t_i]$.
Note that the Gaifman graph of each $\mathcal{X}_i$ is the Gaifman graph of $\mathcal{X}_{\beta(i)}^{(\alpha(i))}$ together with some isolated vertices corresponding to the vertices in $S_{\alpha(i)}$, because the sets in $\mathcal{X}_{\beta(i)}^{(\alpha(i))}$ are all disjoint from $S_{\alpha(i)}$.
Therefore, the Gaifman graphs of $\mathcal{X}_i$ and $\mathcal{X}_{\beta(i)}^{(\alpha(i))}$ have the same treewidth, which is
\begin{equation*}
    (\tilde{\eta} \delta)^{O(1)} \cdot k^\rho = (\tau \delta)^{O(1)} \cdot k^{\rho},
\end{equation*}
by Lemma~\ref{lem-mu=0}.
To verify the second condition, we first consider the ``if'' direction.
Suppose $S$ hits $\mathcal{X}_i$ for some $i \in [t]$.
Then $S_{\alpha(i)} \subseteq S$ and $S \cap V(G_{\alpha(i)})$ hits $\mathcal{X}_{\beta(i)}^{(\alpha(i))}$.
By Lemma~\ref{lem-mu=0}, $S \cap V(G_{\alpha(i)})$ is an $\mathcal{F}$-hitting set of $G_{\alpha(i)}$.
Therefore, $S = S_{\alpha(i)} \cup (S \cap V(G_{\alpha(i)}))$ is an $\mathcal{F}$-hitting set of $G$, since $G-S = G_{\alpha(i)} - (S \cap V(G_{\alpha(i)}))$.
Next, consider the ``only if'' direction.
Suppose $S$ is an $\mathcal{F}$-hitting set of $G$.
By Lemma~\ref{lem-cliquebranch}, there exists $\alpha \in [r]$ such that $S_\alpha \subseteq S$.
Then $S \cap V(G_\alpha)$ is an $\mathcal{F}$-hitting set of $G_\alpha$, since $G-S = G_\alpha - (S \cap V(G_\alpha))$.
Now by Lemma~\ref{lem-mu=0}, $S \cap V(G_\alpha)$ hits $\mathcal{X}_{\beta}^{(\alpha)}$ for some $\beta \in [t_\alpha]$.
Let $i \in [t]$ be the unique instance such that $f(i) = (\alpha,\beta)$.
Then $S$ hits $\mathcal{X}_i = \{\{v\}: v \in S_{\alpha}\} \cup \mathcal{X}_{\beta}^{(\alpha)}$, because $S_\alpha \subseteq S$ and $S \cap V(G_\alpha)$ hits $\mathcal{X}_{\beta}^{(\alpha)}$.
\end{proof}

Set $\tau = k^\varepsilon$ and $\delta = k^{\varepsilon'}$ for sufficiently small (compared to $\mathcal{F}$, $\gamma$, $\eta$, $\mu$, and $\rho$) constants $\varepsilon,\varepsilon' > 0$ such that $\varepsilon'$ is much larger than $\varepsilon$.
Then one can easily verify from the bounds we achieved above that $t = 2^{k^{1-\Omega(1)}}$ and $|\mathcal{X}_i| = k^{O(1)}$ for all $i \in [t]$.
Also, the time complexity for constructing $\mathcal{X}_1,\dots,\mathcal{X}_t$ is bounded by $n^{k^{1-\Omega(1)}}$, and the Gaifman graphs of $\mathcal{X}_1,\dots,\mathcal{X}_t$ have treewidth $k^{1-\Omega(1)}$, since $\rho < 1$.
Therefore, the constant $c<1$ in Theorem~\ref{thm-subtwbranch} exists.
This completes the proof of Theorem~\ref{thm-subtwbranch}, which further implies Corollary~\ref{cor-branching}.



\section{Subgraph hitting algorithms} \label{sec-subhit}

Given Corollary~\ref{cor-branching}, to prove Theorem~\ref{thm-main}, we only need an efficient algorithm for solving general (weighted) hitting-set instances whose Gaifman graphs have bounded treewidth.

\begin{lemma} \label{lem-treeDP}
    Given a collection $\mathcal{X}$ of sets and a function $w: \bigcup_{X \in \mathcal{X}} X \rightarrow \mathbb{R}_{\geq 0}$, a minimum-weight subset $S \subseteq \bigcup_{X \in \mathcal{X}} X$ under the weight function $w$ satisfying that $|S| \leq k$ and $S$ hits $\mathcal{X}$ can be computed in $2^{O(\mathbf{tw}(G^*))} \cdot (n+|\mathcal{X}|)$ time, where $n = |\bigcup_{X \in \mathcal{X}} X|$ and $G^*$ is the Gaifman graph of $\mathcal{X}$.
\end{lemma}
\begin{proof}
Let ${\mathcal X}=\{X_1,\ldots,X_{\ell}\}$.
Each $X_i$ forms a clique in $G^*$, and thus $|X_i| \leq \mathbf{tw}(G^*)+1$.
So we can construct $G^*$ in $\mathbf{tw}^2(G^*) \cdot \ell$ time by considering every $X_i$ and add the edges $(u,v)$ for $u,v \in X_i$.
Let $H$ be the graph constructed from $G^*$ as follows.
For each $i\in [\ell]$, we add a vertex $x_i$ and add edges between $x_i$ and vertices in $X_i$. That is, $V(H)=V(G^*)\cup \{x_i: i\in [\ell]\}$ and $E(H)=E(G^*)\cup \{(x_i,y): i\in [\ell] \text{ and } y\in X_i\}$.
Notice that $G^*$ is an induced subgraph of $H$.
Observe that a subset $S \subseteq V(G^*)$ hits $\mathcal{X}$ iff $S$ dominates all vertices in $\{x_i: i\in [\ell]\}$ in $H$.
We claim that $\mathbf{tw}(H)\leq 1+ \mathbf{tw}(G^*)$.
To see this, let $(T,\beta)$ be a tree decomposition of $G^*$ of width $\mathbf{tw}(G^*)$.
Since $X_i$ is a clique in $G^*$ for each $i\in [\ell]$, there is a node $t_i \in T$ such that $X_i \subseteq \beta(t_i)$.
We add to $T$ a new node $t_i'$ with an edge connecting to $t_i$ whose bag is $\beta(t_i') = X_i \cup \{x_i\}$, for every $i \in [\ell]$.
Let the new tree be $T'$, and it is clear that $(T',\beta)$ is a tree decomposition of $H$ of width at most  $\mathbf{tw}(G^*)+1$.
Next, we define a weight function $w':V(H) \rightarrow \mathbb{R}_{\geq 0}$ as follows.
For each $x\in V(H)$, we set $w'(x)=w(x)$ if $x\in V(G^*)$, and set $w'(x)= \infty$ otherwise.
We solve the {\sc Weighted Subset Dominating Set} instance $(H,w',\{x_i:i \in [\ell]\},k)$, which aims to compute a minimum-weight set $S \subseteq V(H)$ under the weight function $w'$ such that $|S| \leq k$ and $S$ dominates all vertices in $\{x_i:i \in [\ell]\}$~\cite{CyganFKLMPPS15}.
If the weight of $S$ is not infinity, then $S \subseteq V(G^*)$ and it is a minimum-weight set satisfying that $|S| \leq k$ and $S$ hits $\mathcal{X}$.
Otherwise, we can conclude that $\mathcal{X}$ has no hitting set of size at most $k$.
\end{proof}

Now we are ready to prove Theorem~\ref{thm-main}.
Given an instance $(G,w,k)$ of \textsc{Weighted $\mathcal{F}$-Hitting} where $G \in \mathcal{G} \subseteq \mathcal{G}(\eta,\mu,\rho)$ for $\rho < 1$, we simply compute the collections $\mathcal{X}_1,\dots,\mathcal{X}_t$ in Corollary~\ref{cor-branching} for the graph $G$.
This step takes $2^{O(k^c)} \cdot n + O(m)$ time and $t = 2^{O(k^c)}$, where $c<1$ is the constant in the corollary.
The size of each $\mathcal{X}_i$ is bounded by $k^{O(1)}$.
Then we use Lemma~\ref{lem-treeDP} to compute a minimum-weight (under the weight function $w$) set $S_i \subseteq V(G)$ satisfying that $|S_i| \leq k$ and $S_i$ hits $\mathcal{X}_i$ for each $i \in [t]$, which takes $2^{O(k^c)}$ time, because the treewidth of the Gaifman graph of $\mathcal{X}_i$ is $O(k^c)$ by Corollary~\ref{cor-branching} and $|\mathcal{X}_i| = k^{O(1)}$.
As $t = 2^{O(k^c)}$, the overall time cost for computing $S_1,\dots,S_t$ is $2^{O(k^c)} \cdot n + O(m)$.
Finally, we return the set $S_i$ with the minimum total weight among $S_1,\dots,S_t$.
We claim that $S_i$ is an optimal solution of the instance $(G,w,k)$.
First, as $S_i$ hits $\mathcal{X}_i$, it is an $\mathcal{F}$-hitting set of $G$, by the first condition in Corollary~\ref{cor-branching}.
Also, $|S_i| \leq k$ by our construction.
To see the optimality of $S_i$, let $S^* \subseteq V(G)$ be an optimal solution of $(G,w,k)$.
By Corollary~\ref{cor-branching}, $S^*$ hits $\mathcal{X}_j$ for some $j \in [t]$.
So we have
\begin{equation*}
    \sum_{v \in S_i} w(v) \leq \sum_{v \in S_j} w(v) \leq \sum_{v \in S^*} w(v).
\end{equation*}
Therefore, $S_i$ is an optimal solution.
This completes the proof of Theorem~\ref{thm-main}.

\section{Conclusion and open problems} \label{sec-conclusion}

In this paper, we studied the (\textsc{Weighted}) \textsc{$\mathcal{F}$-Hitting} problem, in which we are given a (vertex-weighted) graph $G$ together with a parameter $k$, and the goal is to compute a set $S \subseteq V(G)$ of vertices (with minimum total weight) such that $|S| \leq k$ and $G - S$ does not contain any graph in $\mathcal{F}$ as a subgraph.
We gave a general framework for designing subexponential FPT algorithms for \textsc{\textsc{Weighted} $\mathcal{F}$-Hitting} on a large family of graph classes that admit ``small'' separators relative to the graph size and the maximum clique size.
Such graph classes include all classes of polynomial expansion and many important classes of geometric intersection graphs, on which subexponential algorithms are widely studied.
The algorithms obtained from our framework runs in $2^{O(k^c)} \cdot n + O(m)$ time, where $n = |V(G)|$ and $m = |E(G)|$.
The technical core of our framework is a subexponential branching algorithm that reduces an instance of \textsc{$\mathcal{F}$-Hitting} (on the aforementioned graph classes) to $2^{O(k^c)}$ instances of the general hitting-set problem, where the Gaifman graph of each instance has treewidth $O(k^c)$, for some constant $c < 1$.

We now propose several open problems for future research.
First, as mentioned in the introduction, one can study the ``induced'' variant of the \textsc{$\mathcal{F}$-Hitting} problem, called \textsc{Induced $\mathcal{F}$-Hitting}, which aims to delete $k$ vertices from an input graph $G$ such that the resulting graph does not contain any $F \in \mathcal{F}$ as an \textit{induced} subgraph.
It is interesting to ask whether \textsc{Induced $\mathcal{F}$-Hitting} can be solved in $2^{O(k^c)} \cdot n^{f(\mathcal{F})}$ time for $c<1$ and some function $f$ on the graph classes considered in this paper.
For more restrictive graph classes (such as graph classes of polynomial expansion), one can even seek for algorithms with running time $2^{O(k^c)} \cdot n$ for $c < 1$.
Second, the constant $c$ in the running time $2^{O(k^c)} \cdot n + O(m)$ of our algorithms is only slightly smaller than 1.
How to further improve the running time would be a natural problem.
Finally, one can investigate the problem with parameters other than the solution size $k$, such as the treewidth $t$ (or other width parameters) of the input graph $G$.
Although Cygan et al.~\cite{cygan2017hitting} showed that solving \textsc{$\mathcal{F}$-Hitting} requires $2^{t^{f(\mathcal{F})}} \cdot n$ time for some function $f$, this hardness result was proved on general graphs.
An interesting open question here is whether \textsc{$\mathcal{F}$-Hitting} on the graph classes considered in this paper can be solved in $2^{O(t^c)} \cdot n$ time for a constant $c$ \textit{independent} of $\mathcal{F}$, or even subexponential time in $t$.

\bibliographystyle{plainurl}
\bibliography{my_bib}

\appendix

\section{Important instances of {\sc $\mathcal{F}$-Hitting}}

In this section, we summarize some important instances of the \textsc{$\mathcal{F}$-Hitting} problem that are studied independently in the literature and provide their background.

\paragraph{Vertex cover.}
The simplest instance of \textsc{$\mathcal{F}$-Hitting} is the \textsc{Vertex Cover} problem, where $\mathcal{F} = \{K_2\}$.
{\sc Vertex Cover} is widely considered to be the most extensively studied problem in parameterized complexity, and is often considered a testbed for the introduction of new notions and techniques~\cite{CyganFKLMPPS15}.
For {\sc Vertex Cover} on general graphs and on bounded degree graphs, there has been long races to achieve the best possible running time; see~\cite{balasubramanian1998improved,
buss1993nondeterminism,chandran2004refined,chen2001vertex,
DBLP:journals/tcs/ChenKX10,xiao2017exact} for a few of the papers.
The best known parameterized algorithm the bound $1.28^k\cdot n^{O(1)}$~\cite{DBLP:journals/tcs/ChenKX10}.
Furthermore, it is well-known that {\sc Vertex Cover} admits a kernel with $2k$ vertices~\cite{chen2001vertex}.
The kernel, together with the standard $2^{O(w)} n^{O(1)}$-time DP algorithm for {\sc Vertex Cover} on graphs of treewidth $w$, directly implies a subexponential FPT algorithm for {\sc Vertex Cover} on graph classes that admit sublinear balanced separators.
Indeed, graphs with sublinear balanced separators have treewidth sublinear in the number of vertices, and therefore applying the DP algorithm to the linear kernel takes $2^{o(k)} n^{O(1)}$ time.
Based on this, it is not difficult to achieve the result of Theorem~\ref{thm-main} for {\sc Vertex Cover}.
Of course, this approach does not work for a general \textsc{$\mathcal{F}$-Hitting} problem, because both pieces --- the linear kernel and the DP algorithm with exponential running time in the treewidth --- are missing when going beyond {\sc Vertex Cover}.

\paragraph{Path transversal.}
In the \textsc{Path transversal} problem, the goal is to hit all paths in the graph of certain length.
Thus, this is a class of \textsc{$\mathcal{F}$-Hitting} instances where $\mathcal{F} = \{P_\ell\}$ for some $\ell \in \mathbb{N}$.
{\sc Vertex Cover} is a special case of {\sc Path Transversal} when $\ell = 2$.
In some literature, {\sc Path Transversal} is also known as \textsc{Vertex $\ell$-Path Cover}.
The problem was intensively studied in parameterized complexity: various parameterized algorithms \cite{bai2019improved,DBLP:conf/mfcs/CervenyS19,vcerveny2023generating,DBLP:journals/jcss/ShachnaiZ17,tsur2019path,DBLP:journals/dam/Tsur21,DBLP:journals/dam/TuJ16,DBLP:journals/jco/TuWYC17} and kernels \cite{vcerveny2024kernels,DBLP:journals/jcss/FellowsGMN11,DBLP:journals/jcss/Xiao17,DBLP:conf/tamc/XiaoK17} were developed.
Specifically, the case $\ell = 3$ received much attentions.
It was known that this case is solvable in polynomial time on various graph classes~\cite{DBLP:journals/arscom/BoliacCL04,DBLP:journals/dam/BrauseK17,DBLP:journals/amc/BresarKKS19,DBLP:journals/mp/CameronH06,DBLP:journals/ipl/LozinR03,DBLP:journals/dam/OrlovichDFGW11} such as $P_5$-free graphs~\cite{DBLP:journals/dam/BrauseK17}, and on general graphs, various exponential-time algorothms were developed~\cite{chang20141,2018moderately,kardovs2011computing,xiao2017exact}.
Furthermore, bidimesnionality theory yields subexponential FPT algorithms for {\sc Path Transversal} on planar graphs (or more generally, minor-free graphs), unit-disk graphs, and map graphs~\cite{demaine2005subexponential,fomin2019decomposition,FominLPSZ19,DBLP:conf/soda/FominLS12}. 

\paragraph{Short cycle hitting.}
In the \textsc{Short Cycle Hitting} problem, the goal is to hit all cycles in the graph of certain length(s).
Thus, this is a class of \textsc{$\mathcal{F}$-Hitting} instances where $\mathcal{F} = \{C_\ell\}$ for some $\ell \in \mathbb{N}$ or more generally, $\mathcal{F} = \{C_{\ell_1},\dots,C_{\ell_r}\}$ for $\ell_1,\dots,\ell_r \in \mathbb{N}$.
The restriction of the case where $k=3$ (i.e., when we aim to hit triangles), known as \textsc{Triangle Hitting}, admits a subexponential FPT algorithm on disk graphs~\cite{lokshtanov2022subexponential}, and thus also on planar graphs and unit-disk graphs.
Furthermore, the complexity of {\sc Short Cycle Hitting} parameterized by treewidth was studied in~\cite{pilipczuk2011problems}, and a dichotomy concerning its membership in P was given in~\cite{groshaus2009cycle}.

\paragraph{Degree modulator.}
In the {\sc Degree Modulator} problem, the goal is to delete vertices to reduce the maximum degree of the graph to a certain number.
Thus, this is a class of \textsc{$\mathcal{F}$-Hitting} instances where $\mathcal{F} = \{K_{1,\ell+1}\}$ for some $\ell \in \mathbb{N}$.
{\sc Vertex Cover} is a special case of {\sc Degree modulator} when $\ell = 0$.
FPT algorithms, kernels, and W[1]-hardness results for {\sc Degree modulator} (sometimes referred to by a different name) parameterized by $\ell$, the solution size, or a structural measure such as treewidth of the graph were given in~\cite{betzler2012bounded,DBLP:conf/isaac/DessmarkJL93,DBLP:journals/jcss/FellowsGMN11,ganian2021structural,DBLP:conf/soda/LokshtanovMRSZ21,DBLP:journals/jco/MoserNS12,DBLP:journals/dam/NishimuraRT05}.

\paragraph{Clique hitting and biclique hitting.}
These are classses of \textsc{$\mathcal{F}$-Hitting} instances where $\mathcal{F} = \{K_\ell\}$ and $\mathcal{F} = \{K_{\ell,\ell'}\}$ for some $\ell,\ell' \in \mathbb{N}$, respectively.
Parameterized algorithms (and also approximation algorithms) for \textsc{Clique Hitting} on perfect graphs were given in~\cite{fiorini2018approximability}.
Very recently, Berthe et al.~\cite{berthe2024kick} gave subexponential parameterized algorithms for \textsc{Clique Hitting} on a wide family of graph classes.
The parameterized complexity of \textsc{Biclique Hitting} was studied in~\cite{goldmann2021parameterized}.

\paragraph{Component order connectivity.}
In the \textsc{Component Order Connectivity} problem, the goal is to delete vertices to make each connected component of the graph having size smaller than or equal to a certain number.
Thus, this is a class of \textsc{$\mathcal{F}$-Hitting} instances where $\mathcal{F}$ consists of all connected graphs of $\ell$ vertices (or equivalently, all trees of $\ell$ nodes) for some $\ell \in \mathbb{N}$.
{\sc Vertex Cover} is a special case of \textsc{Component Order Connectivity} when $\ell=2$.
Drange et al.~\cite{drange2016computational} provided an FPT algorithm and a polynomial kernel for~\textsc{Component Order Connectivity}, and in addition studied the complexity of the problem on special graph classes. 
Improved kernels were given by \cite{kumar20172lk,DBLP:journals/jcss/Xiao17a}.
Bidimesnionality theory yields subexponential FPT algorithms for \textsc{Component Order Connectivity} on planar graphs (or more generally, minor-free graphs), unit-disk graphs, and map graphs~\cite{demaine2005subexponential,fomin2019decomposition,FominLPSZ19,DBLP:conf/soda/FominLS12}.
Additional results can be found in the survey~\cite{gross2013survey}.

\paragraph{Treedepth modulator.}
In the \textsc{Treedepth Modulator} problem, the goal is to delete vertices to reduce the treedepth of the graph to a certain number.
Thus, this is a class of \textsc{$\mathcal{F}$-Hitting} instances where $\mathcal{F}$ consists of all minimal graphs whose treedepth is $\ell+1$ for some $\ell \in \mathbb{N}$ (it was known that $\mathcal{F}$ is finite~\cite{dvovrak2012forbidden}).
The computation and utility of treedepth modulators were studied from the perspective of parameterized complexity (see, e.g., \cite{bougeret2019much,eiben2022lossy,gajarsky2017kernelization}).

\end{document}